\documentclass[11pt]{article}
\usepackage{etoolbox}
\usepackage{microtype}
\usepackage{graphicx}
\usepackage{latexsym}
\usepackage{amssymb}
\usepackage{amsmath}
\usepackage{amsthm}
\usepackage{forloop}
\usepackage[paper=letterpaper,margin=1in]{geometry}
\usepackage[ruled,vlined,linesnumbered,algo2e]{algorithm2e}
\usepackage{paralist}
\usepackage{nicefrac}
\usepackage{thmtools}
\usepackage{array}
\usepackage{float,subfig}
\usepackage{multirow}
\usepackage{xcolor}
\usepackage[sort&compress,numbers]{natbib}
\usepackage[hyphens]{url}
\usepackage{hyperref}
\hypersetup{
	colorlinks = true,
	linkcolor  = blue!80!black,
	citecolor  = blue!80!black,
	urlcolor   = blue!80!black,
}
\usepackage[capitalise,noabbrev]{cleveref}

\newcommand{\AlgFTG}{{\textsc{\textsc{FastThresholdGreedy}}}\xspace}
\newcommand{\AlgBR}{{\textsc{\textsc{BoostRatio}}}\xspace}
\newcommand{\AlgG}{{\textsc{\textsc{Greedy}}}\xspace}
\newcommand{\AlgLG}{{\textsc{\textsc{LazyGreedy}}}\xspace}
\newcommand{\AlgSG}{{\textsc{\textsc{StochasticGreedy}}}\xspace}
\newcommand{\AlgBG}{{\textsc{\textsc{BarrierGreedy}}}\xspace}
\newcommand{\AlgDG}{{\textsc{\textsc{DensityGreedy}}}\xspace}
\newcommand{\AlgMRT}{{\textsc{\textsc{DynamicMrt}}}\xspace}
\newcommand{\AlgFast}{{\textsc{\textsc{Fast}}}\xspace}
\newcommand{\FTG}{{\textsc{\textsc{FTG}}}\xspace}
\newcommand{\FTGP}{{\textsc{\textsc{FTGP}}}\xspace}
\newtoggle{anonymize}
\togglefalse{anonymize}

\newtheorem{theorem}{Theorem}[section]
\newtheorem{lemma}[theorem]{Lemma}

\newtheorem{corollary}[theorem]{Corollary}

\newtheorem{proposition}[theorem]{Proposition}
\newtheorem{observation}[theorem]{Observation}

\newcommand{\defcal}[1]{\expandafter\newcommand\csname c#1\endcsname{{\mathcal{#1}}}}
\newcommand{\defbb}[1]{\expandafter\newcommand\csname b#1\endcsname{{\mathbb{#1}}}}
\newcounter{calBbCounter}
\forLoop{1}{26}{calBbCounter}{
    \edef\letter{\Alph{calBbCounter}}
    \expandafter\defcal\letter
		\expandafter\defbb\letter
}

\newcommand{\eps}{\varepsilon}
\newcommand{\ie}{{\it i.e.}}

\newcommand{\nnR}{{\bR_{\geq 0}}}
\newcommand{\SetExtract}{\texttt{SetExtract}}
\newcommand{\BigAlg}{\texttt{BigElementsAlg}}
\newcommand{\SMKS}{{\texttt{SMKS}}}
\newcommand{\SMK}{{\texttt{SMK}}}
\newcommand{\SMC}{{\texttt{SMC}}}
\newcommand{\SI}{{\texttt{Set-\allowbreak Identification}}}
\newcommand{\USM}{{\texttt{USM}}}
\newcommand{\characteristic}{{\mathbf{1}}}
\newcommand{\SMCFull}{{\texttt{Submodular Maximization subject to a Cardinality Constraint}}}
\newcommand{\SMKFull}{{\texttt{Submodular Maximization subject to a Knapsack Constraint}}}
\newcommand{\USMFull}{{\texttt{Unconstrained Submodular Maximi\-zation}}}
\newcommand{\SMKSFull}{{\texttt{Submodular Maximization subject to Knapsacks and a Set System Constraint}}}

\title{Submodular Maximization in Clean Linear Time}
\author{
		Wenxin Li\thanks{Department of Electrical and Computer Engineering, Ohio State University. Email: \texttt{wenxinliwx.1@gmail.com}.} \and
		Moran Feldman\thanks{Department of Computer Science, University of Haifa, Israel. Email: \texttt{moranfe@cs.haifa.ac.il}.} \and
		Ehsan Kazemi\thanks{Google, Z\"{u}rich. Email: \texttt{ehsankazemi@google.com}.} \and
		Amin Karbasi\thanks{School of Engineering and Applied Science, Yale University. Email: \texttt{amin.karbasi@yale.edu}.}
}

\begin{document}
\date{}
\maketitle
\pagenumbering{arabic}

\begin{abstract}
In this paper, we provide the first deterministic algorithm that achieves the tight $1-1/e$ approximation guarantee for submodular maximization under a cardinality (size) constraint while making a number of queries that scales only linearly with the size of the ground set $n$. To complement our result, we also show strong information-theoretic lower bounds. More specifically, we show that when the maximum cardinality allowed for a solution is constant, no algorithm making a sub-linear number of function evaluations can guarantee any constant approximation ratio. Furthermore,  when the constraint allows the selection of a constant fraction of the ground set, we show that any algorithm making fewer than $\Omega(n/\log(n))$ function evaluations cannot perform better than an algorithm that simply outputs a uniformly random subset of the ground set of the right size. We then provide a variant of our deterministic algorithm for the more general knapsack constraint, which is the first linear-time algorithm that achieves $1/2$-approximation guarantee for this constraint. Finally, we extend our results to the general case of maximizing a monotone submodular function subject to the intersection of a $p$-set system and multiple knapsack constraints. We extensively evaluate the performance of our algorithms on multiple real-life machine learning applications, including movie recommendation, location summarization, twitter text summarization and video summarization.

\end{abstract}
\section{Introduction}

One of today's most ubiquitous challenges is making an informed decision while facing a massive amount of data. 
Indeed, many machine learning models feature inherently
discrete decision variables: from phrases in a corpus to objects in an image. Similarly, nearly all aspects
of the machine learning pipeline involve discrete tasks, from data summarization \cite{lin2011class} and sketching \cite{muthukrishnan2005data} to
feature selection \cite{das2011submodular} and model explanation \cite{ribeiro2016should}.

The study of how to make near-optimal decisions from a
massive pool of possibilities is at the heart of combinatorial optimization. To make use of combinatorial optimization in machine learning, it is of great importance to understand which discrete formulations can be solved efficiently, and how to do so. Many of these problems
are notoriously hard, and even those that are theoretically tractable may not scale to large instances. However,
the problems of practical interest are often much more well-behaved, and possess extra structures that make
them amenable to exact or approximate optimization techniques. Just as convexity has been a celebrated
and well-studied condition under which continuous optimization is tractable, submodularity is a condition allowing
discrete objectives to be optimized \cite{fujishige91old}. Informally, submodularity captures the intuitive diminishing
returns notion, where the added value of an element (e.g., image, sensor, etc.) to a context (data set
of images, sensor network, etc.) decreases as the context in which it is considered increases.

Submodular functions have  attracted significant interest from the machine
learning community, and have led to the development of algorithms with near-optimal solutions for a wide range
of applications, including 
outbreak detection \cite{leskovec2007cost},  graph 
cuts in computer vision \cite{jegelka2011submodularity},  image and video 
summarization \cite{tschiatschek2014learning, mitrovic2018data, feldman2018less, harshaw2021power},  active  learning \cite{guillory10interactive,golovin11,esfandiari2021adaptivity},  compressed sensing and structured sparsity 
\cite{bach2010structured, elenberg2018restricted}, fairness in machine learning 
\cite{celis2016fair, kazemi2018scalable}, recommendation \cite{mitrovic2019adaptive},  human brain parcellation \cite{salehi2017submodular}, model training \cite{mirzasoleiman2020coresets, mirzasoleiman2016distributed},  and learning causal structures \cite{zhou2016causal, sussex2021near}. For recent surveys on the applications of submodular functions in machine learning and artificial intelligence, we refer the interested reader to \cite{tohidi2020submodularity, bilmes2022submodularity}.

In their seminal works, \citet{nemhauser1978best} and \citet{nemhauser1078analysis} showed that a greedy algorithm achieves a tight $1-1/e$ approximation guarantee for maximizing  a monotone submodualr function subject to cardinality constraint ({\SMC}). Since then, there has been a significant effort to develop faster algorithms and extend the approximation guarantee to more general settings (such as non-monotone functions, or more complex constraints). In particular, it has been observed that, for many machine learning applications, running the greedy algorithm is too expensive as it requires $O(nk)$ function evaluations, where $n$ is the size of the ground set and $k$ is the maximum cardinality allowed by the constraint. Many works attempted to solve this issue: first, using heuristics such as lazy greedy \cite{minoux1978accelerated} or multi-stage greedy \cite{wei2014fast} that aim to  reduce the number of function evaluations and the run-time of each function evaluation, respectively; and later, by suggesting new algorithms that guarantee $(1-1/e-\eps)$-approximation using provably fewer function evaluations. Currently, there is a known \emph{randomized} algorithm achieving this approximation ratio in $O_\eps(n)$ time~\cite{buchbinder2017comparing,mirzasoleiman2015lazier}; however, the state-of-the-art \emph{deterministic} algorithm only manages to achieve this approximation ratio using $O_\eps(n \log \log n)$ function evaluations~\cite{huang2018multi}. In this work, we propose the first deterministic linear-time algorithm that guarantees $(1-1/e-\eps)$ approximation via $O_\eps(n)$ function evaluations.\footnote{See Section~\ref{ssc:independent_works} for a discussion of two recent works that independently obtained similar results.}

Since the linear time complexity of our algorithm is the minimum required to read all the elements of the ground set, it is difficult to believe that any algorithm can improve over it while guaranteeing a reasonable approximation ratio. However, the situation is somewhat different when one is interested in the number of function evaluations made by the algorithm. An algorithm making a sub-linear number of function evaluations can be interesting, even if its total time complexity is linear, because the cost of every function evaluation is quite high in some applications. Since an algorithm can gather information about many elements of the ground set with a single function evaluation, it is not a priori clear that such an algorithm cannot exist. However, we are able to prove that this is indeed the case in at least two important regimes. First, we show that when the maximum cardinality allowed for the solution is constant, no algorithm making a sub-linear number of function evaluations can guarantee any constant approximation ratio for {\SMC}.\footnote{This inapproximability result also appears in~\cite{kuhnle2021quick}, which was published after appearance of an earlier arXiv version of this paper\iftoggle{anonymize}{}{~\cite{li2020note}} that included this result.} Second, when the constraint allows the selection of a constant fraction of the ground set, we show that any algorithm making fewer than $\Omega(n/\log(n))$ function evaluations cannot perform better than an algorithm that simply outputs a uniformly random subset of the ground set of the right size. Our technique also leads to a similar inapproximability result for unconstrained submodular maximization. These lower bounds, along with the corresponding linear-time algorithms, nearly complete our understanding of the minimum query complexity of their respective problems. It is important to note that (as is often the case in the field of submodular maximization) our inapproximability results are based on information theoretic arguments rather than relying on any complexity assumptions.

In many applications, the elements of the ground set may have different costs, and instead of selecting a set of size $k$, we might be interested in selecting a set whose total cost does not exceed a predefined budget.  We extend our algorithm to this problem---formally, we we aim here to maximize a monotone submodular function subject to a knapsack constraint ({\SMK})---and prove that this extended algorithm is the first linear time algorithm achieving $(1/2-\eps)$-approximation using $O_\eps(n)$ function evaluations. We note that some nearly-linear time algorithms have been previously proposed for this problem. Ene and Nguyen~\cite{ene2019nearly} showed such an algorithm achieving $(1-1/e-\eps)$-approximation, but unfortunately, the number of function evaluation used by their algorithm scales poorly with $\eps$ (namely, $O((1/\eps)^{1/\eps})$) which makes is impractical even for moderate values of $\eps$. This has motivated Yaroslavtsev et al.~\cite{yaroslavtsev2020bring} to design an algorithm called \texttt{Greedy+} that achieve $(\nicefrac{1}{2} - \eps)$-approximation and has a much more practical nearly-linear time complexity. Nevertheless, our algorithm manages to combine the approximation guarantee of \texttt{Greedy+} with a clean linear time.

Finally, we consider the more general setting where we aim to maximize a monotone submodular function subject to a $p$-set system and $d$ knapsack constraints ({\SMKS}). We study the
tradeoff between the time complexity and approximation ratio for this problem. In particular, our results improve
over the state-of-the-art approximation for nearly-linear time algorithms~\cite{badanidiyuru2014fast}.

We extensively evaluate the performance of our algorithms using real-life and synthetic  experiments, including movie recommendation, location summarization, vertex cover, twitter text summarization and video summarization.

\subsection{Comparison with Recent Independent Works} \label{ssc:independent_works}

Most of our results have never appeared in another paper. However, in the case of maximizing a monotone submodualr function subject to cardinality constraint (\SMC), there are two recent independent works that overlap our results and we would like to discuss. Let us begin with the work of \citet{kuhnle2021quick}, which presents an algorithm that is similar to our algorithm (and in particular, fully matches the theoretical properties we prove for our {\SMC} algorithm). Formally, an earlier version\iftoggle{anonymize}{}{~\cite{li2020note}} of this paper appeared on arXiv and achieved linear query complexity, but using nearly-linear time; and later, the arXiv version of \citet{kuhnle2021quick} appeared as \citep{kuhnle2020quick} and achieved both linear query and time complexities. However, both arXiv versions were posted long after the results presented in them were first obtained, and neither group was aware of the work of the other group until recently (personal communication~\citep{kuhnle2022personal}), and therefore, our work and the work of~\citep{kuhnle2021quick} should be considered as independent of each other, with neither clearly preceding the other in time.

In addition to the above works, very recently, a journal version of~\cite{huang2018multi} appeared as~\cite{huang2022multi} and claimed another deterministic algorithm that achieves linear time, but with a worse dependence on $\eps$ compared to our work and the work of~\citet{kuhnle2021quick}.

\subsection{Additional Related Work} \label{ssc:related_work}

The work of \citet{badanidiyuru2014fast} is probably the first work whose main motivation was to find algorithms with guaranteed approximation ratios for submodular maximization problems that provably enjoy faster time complexities. This work introduced the thresholding technique, which became a central component of basically every fast deterministic implementation of the greedy algorithm considered to date, including our algorithms. In particular, using this technique, \citet{badanidiyuru2014fast} described an $O_\eps(n \log n)$ time algorithm for maximizing a monotone submodular function subject to a cardinality constraint.

\citet{sviridenko2004note} was the first to obtain $(1 - 1/e)$-approximation for maximizing a monotone submodular function subject to a knapsack constraint. His original algorithm was very slow (having a time complexity of $O(n^5)$), but it can be improved to run in $\tilde{O}(n^4)$ time using the thresholding technique of~\cite{badanidiyuru2014fast}. The algorithm of \citet{sviridenko2004note} is based on making $3$ guesses. Recently, \citet{feldman2021practical} showed that the optimal approximation ratio of $1 - 1/e$ can be obtained using only $2$ guesses, leading to $\tilde{O}(n^3)$ time. Furthermore, \citet{feldman2021practical} also showed that a single guess suffices to guarantee a slightly worse approximation ratio of $0.6174$ and requires only $\tilde{O}(n^2)$ time. As mentioned above, there exists also an impractical nearly-linear time algorithm for the problem due to Ene and Nguyen~\cite{ene2019nearly}. Another fast algorithm was suggested by~\cite{badanidiyuru2014fast}, but an error was later found in their analysis (see~\cite{ene2019nearly} for details). We also note that multiple recent works aimed to bound the performance of a natural (and historically important) greedy algorithm~\cite{feldman2021practical,kulik2021faster,tang2021revisiting}.

When considering fast algorithms for combinatorial optimization problems, it is important to consider also (semi-)streaming algorithms since algorithms of this type are usually naturally fast. Among the first semi-streaming algorithms for submodular maximization was the work of \citet{badanidiyuru2014streaming}, who obtained $(1/2 - \eps)$-approximation for maximizing a submodular function subject to a cardinality constraint. While this work was later slightly improved~\cite{kazemi2019submodular} and extended to non-monotone functions~\cite{alaluf2020optimal}, it was also proved that no semi-streaming algorithm can provide a better than $1/2$-approximation for this problem~\cite{feldman2020}. The more general case of a knapsack constraint was first considered in the semi-streaming setting by~\citet{yu2018streaming}. Improved algorithms for this problem were later found by~\cite{huang2021improved,huang2020streaming} for single pass and by~\cite{huang2022multi,yaroslavtsev2020bring} for multiple passes.

Maximizing a monotone submodular function subject to a $p$-set system constraint was first studied by Fisher et al.~\cite{fisher1978analysis}, who showed that a natural greedy algorithm obtains $(p + 1)^{-1}$-approximation for this problem. Many years later, \citet{badanidiyuru2014fast} used their thresholding technique to get a nearly-linear time algorithm guaranteeing $(p + 2d + 1 + \eps)^{-1}$-approximation even when there are $d$ knapsack constraints (in addition to the $p$-set system constraint). Recently, \citet{badanidiyuru2020submodular} presented a new technique that allows them to improve over the last approximation ratio. They explicitly consider the implications of their ideas only for the regime $d \leq p$, and achieve in this regime $[2(p + 1 + \eps)]^{-1}$-approximation and $(p + 1 + \eps)^{-1}$-approximation in $\tilde{O}(nr/\eps)$ and $\tilde{O}(n^3r^2/\eps)$ time, respectively, where $r \leq n$ is the rank of $p$-set system. We note, however, that some of the ideas of~\cite{badanidiyuru2020submodular} can also be applied to the regime $d > p$, leading to $(p + d + 1)^{-1}$-approximation in $\tilde{O}(n^3)$ time.\footnote{To be more exact, to get the $(p + d + 1)^{-1}$-approximation one has to apply the idea used in the proof of Theorem~8 of~\cite{badanidiyuru2020submodular} to the $(p + 2d + 1 + \eps)^{-1}$-approximation algorithm due to~\cite{badanidiyuru2014fast}.} One can observe that the improvement in the approximation ratio obtained by the new results of~\cite{badanidiyuru2020submodular} comes at the cost of an increased time complexity. Our results in Section~\ref{sec:SMKS} further study this tradeoff between the time complexity and approximation ratio and improve over the state-of-the-art approximation for nearly-linear time algorithms.

\paragraph{Paper Organization.} Section~\ref{sec:preliminaries} presents the definitions and notation that we use. Our algorithms for submodular maximization subject to a cardinality or knapsack constraint appear in Section~\ref{sec:SMK}; and Section~\ref{sec:inapproximability} complements these algorithms by presenting our inapproximability results for a cardinality constraint and unconstrained submodular maximization. Our result for the general case of the intersection of a $p$-set system constraint with multiple knapsack constraints can be found in Section~\ref{sec:SMKS}. Finally, Section~\ref{sec:experiments} describes our experiments and their results.

\section{Preliminaries} \label{sec:preliminaries}

In this section we present the definitions and notation that we use throughout the paper.

\paragraph{Set Functions.} A set function $f\colon 2^\cN \to \bR$ over a ground set $\cN$ is a function that assigns a numeric value to every subset of $\cN$. Given such a function, an element $u \in \cN$ and a set $S \subseteq \cN$, we denote by $f(u \mid S) \triangleq f(S \cup \{u\}) - f(S)$ the marginal contribution of $u$ to $S$ with respect to $f$. The set function $f$ is submodular if $f(u \mid S) \geq f(u \mid T)$ for every two sets $S \subseteq T \subseteq \cN$ and element $u \in \cN \setminus T$, and it is monotone if $f(u \mid S) \geq 0$ for every element $u \in \cN$ and set $S \subseteq \cN$. Since we are interested in multiplicative approximation guarantees, all the set functions considered in this paper are also non-negative. Additionally, given two sets $S, T \in \cN$, we often use the shorthand $f(T \mid S) \triangleq f(S \cup T) - f(T)$.

\paragraph{Set Systems.} A set system $\cM$ is an ordered pair $(\cN, \cI)$, where $\cN$ is a ground set and $\cI$ is a non-empty subset of $2^\cN$ that obeys the following property: if $S \subseteq T \subseteq \cN$ and $T \in \cI$, then $S$ also belongs to $\cI$. It is customary to refer to the sets of $\cI$ as independent (and to sets that do not belong to $\cI$ as dependent). A base of a set $S \subseteq \cN$ is a subset of $S$ that is independent and is not included in any other set having these properties. The set system $\cM = (\cN, \cI)$ is called $p$-set system for an integer $p \geq 1$ if for every $S \subseteq \cN$ the ratio between the sizes of the largest base of $S$ and the smallest base of $S$ is at most $p$.

\paragraph{Additional Notation.} Given a set $S$ and element $u$, we often use $S + u$ as a shorthand for $S \cup \{u\}$.
\section{Cardinality or Knapsack Constraint} \label{sec:SMK}

In this section we consider the {\SMKFull} problem (\SMK). In this problem we are given a non-negative monotone submodular function $f\colon 2^\cN \to \nnR$, a non-negative cost function $c\colon \cN \to \nnR$ and a budget $B > 0$. We say that a set $S \subseteq \cN$ is \emph{feasible} if it obeys $c(S) \leq B$ (where $c(S) \triangleq \sum_{u \in S} c(u)$), and the objective of the problem is to find a set maximizing $f$ among all feasible sets $S \subseteq \cN$. For simplicity, we assume below that $c(u) \in (0, B]$ for every element $u \in \cN$ and that $B = 1$. These assumptions are without loss of generality because (i) every element of $\cN$ whose cost exceeds $B$ can be simply removed, (ii) every element of $\cN$ whose cost is $0$ can be safely assumed to be part of the optimal solution, and (iii) one can scale the costs to make the budget $B$ equal to $1$.

When the knapsack constraint of {\SMK} happens to be simply a cardinality constraint, we get the {\SMCFull} problem (\SMC). Formally, in this problem we have an integer parameter $1 \leq k \leq n$, and we are allowed to output any set whose size is at most $k$. One can observe that {\SMC} is a special case of {\SMK} in which $c(u) = 1/k$ for every element $u \in \cN$.

The first step in the algorithms we develop for {\SMK} and {\SMC} is getting an estimate $\Gamma$ of the value of an optimal solution up to a constant factor. We show how this can be done in Section~\ref{sec:estimate}. Then, in Section~\ref{sec:greedy} we present a variant of the threshold-greedy algorithm of~\cite{badanidiyuru2014fast} that, unlike the original algorithm of~\cite{badanidiyuru2014fast}, uses the estimate $\Gamma$ to reduce the number of thresholds that need to be considered. This already leads to our algorithmic result for {\SMC}, which is given by the next theorem.
\begin{theorem} \label{thm:cardinality_alg}
For every $\eps > 0$, there exists a deterministic $(1 - 1/e - \eps)$-approximation algorithm for {\SMCFull} (\SMC) using $O(n/\eps)$ time.
\end{theorem}

Our algorithmic result for {\SMK} is given by Theorem~\ref{thm:knapsack}. To prove this theorem one must combine the threshold-greedy variant from Section~\ref{sec:greedy} with a post-processing step presented in Section~\ref{sec:plus}.
\begin{theorem} \label{thm:knapsack}
For every $\eps > 0$, there exists a deterministic $(\nicefrac{1}{2} - \eps)$-approximation algorithm for {\SMKFull} (\SMK) that uses $O(n \eps^{-1} \log \eps^{-1})$ time.
\end{theorem}

\subsection{Estimating the Value of an Optimal Solution} \label{sec:estimate}

In this section our objective is to design an algorithm (given as Algorithm~\ref{alg:estimate}) that produces a value $\Gamma$ that obeys $\Gamma \leq f(OPT) \leq s \cdot \Gamma$ for some constant $s \geq 1$, where $OPT$ is an arbitrary optimal solution.

\begin{algorithm2e}[th]
\caption{\textsc{Estimating} $f(OPT)$} \label{alg:estimate}
Let $S \gets \varnothing$.\\
\For{every element $u \in \cN$}
{
	\If{$\tfrac{f(u \mid S)}{c(u)} \geq f(S)$}
	{
		Add $u$ to $S$.
	}
}
\Return{$f(S) / 4$}.
\end{algorithm2e}

Clearly Algorithm~\ref{alg:estimate} can be implemented to work in $O(n)$ time. Therefore, we concentrate on bounding the quality of the estimate it produces. The following observation shows that Algorithm~\ref{alg:estimate} does not underestimate $OPT$ by more than a constant factor. Let $\tilde{S}$ denote the final value of the set $S$ in Algorithm~\ref{alg:estimate}.
\begin{observation} \label{obs:under_estimation}
$f(\tilde{S}) \geq f(OPT) / 2$.
\end{observation}
\begin{proof}
Consider an arbitrary element $u \in OPT \setminus \tilde{S}$. The fact that $u$ was never added to $S$ by Algorithm~\ref{alg:estimate} implies that at the time in which $u$ was processed by Algorithm~\ref{alg:estimate} the following inequality held
\[
	\frac{f(u \mid S)}{c(u)} \leq f(S)
	\enspace.
\]
Since elements are only ever added to $S$, by the monotonicity and submodularity of $f$, the last inequality implies
\[
	\frac{f(u \mid \tilde{S})}{c(u)} \leq f(\tilde{S})
	\enspace.
\]

Using the last inequality, we now get
\[
	f(OPT \mid \tilde{S})
	\leq
	\sum_{u \in OPT \setminus \tilde{S}} \mspace{-9mu} f(u \mid \tilde{S})
	\leq
	f(\tilde{S}) \cdot \sum_{u \in OPT \setminus \tilde{S}} \mspace{-9mu} c(u)
	=
	f(\tilde{S}) \cdot c(OPT \setminus \tilde{S})
	\leq
	f(\tilde{S})
	\enspace,
\]
where the first inequality follows from the submodularity of $f$, and the last inequality holds because $OPT$ is a feasible solution and $OPT \setminus \tilde{S}$ is a subset of it. Plugging the definition of $f(OPT \mid \tilde{S})$ into the last inequality, and rearranging, yields
\[
	f(\tilde{S})
	\geq
	\frac{f(OPT \cup \tilde{S})}{2}
	\geq
	\frac{f(OPT)}{2}
	\enspace,
\]
where the second inequality follows from the monotonicity of $f$.
\end{proof}

To prove that Algorithm~\ref{alg:estimate} also does not over estimate $f(OPT)$, we need to define some additional notation. From this point on, it will be convenient to assume that the elements of $\cN$ are ordered in the order in which Algorithm~\ref{alg:estimate} considers them. Given a set $T \subseteq \cN$, let $T^{\geq 1}$ be the minimal suffix of $T$ (according to the above order) whose value according to $c$ is at least $1$ (unless $c(T) < 1$, in which case $T^{\geq 1} = T$). Clearly $f(\tilde{S}^{\geq 1}) \leq f(\tilde{S})$ by the monotonicity of $f$. However, it turns out that $f(\tilde{S}^{\geq 1})$ is never much smaller than $f(\tilde{S})$.
\begin{lemma} \label{lem:suffix_exponential}
$f(\tilde{S}^{\geq 1}) \geq f(\tilde{S}) / 2$.
\end{lemma}
\begin{proof}
If $c(\tilde{S}^{\geq 1}) < 1$, then $\tilde{S}^{\geq 1} = \tilde{S}$, which makes the lemma trivial. Therefore, we may assume in the rest of the proof that $c(\tilde{S}^{\geq 1}) \geq 1$.

The fact that an element $u \in \tilde{S}^{\geq 1}$ was added to $S$ by Algorithm~\ref{alg:estimate} implies that if, we denote by $S_u$ the set $S$ immediately before $u$ was processed, then we have the following inequality.
\[
	\frac{f(u \mid S_u)}{c(u)}
	\geq
	f(S_u)
	\enspace.
\]
Since $\tilde{S}^{\geq 1}$ is a suffix of $\tilde{S}$, this inequality implies
\[
	f(\tilde{S}^{\geq 1} \mid \tilde{S} \setminus \tilde{S}^{\geq 1})
	=
	\sum_{u \in \tilde{S}^{\geq 1}} f(u \mid S_u)
	\geq
	\sum_{u \in \tilde{S}^{\geq 1}} c(u) \cdot f(S_u)
	\geq
	c(\tilde{S}^{\geq 1}) \cdot f(\tilde{S} \setminus \tilde{S}^{\geq 1})
	\geq
	f(\tilde{S} \setminus \tilde{S}^{\geq 1})
	\enspace,
\]
where the penultimate inequality holds by the monotonicity of $f$, and the last inequality follows from our assumption that $c(\tilde{S}^{\geq 1}) \geq 1$.

Plugging the definition of $f(\tilde{S}^{\geq 1} \mid \tilde{S} \setminus \tilde{S}^{\geq 1})$ into the last inequality gives
\[
	f(\tilde{S})
	\geq
	2f(\tilde{S} \setminus \tilde{S}^{\geq 1})
	\geq
	2[f(\tilde{S}) - f(\tilde{S}^{\geq 1})]
	\enspace,
\]
where the second inequality holds by the submodularity (and non-negativity) of $f$. The lemma now follows by rearranging the last inequality.
\end{proof}
\begin{corollary} \label{cor:over_estimation}
$f(\tilde{S}) \leq 4 \cdot f(OPT)$.
\end{corollary}
\begin{proof}
Since $\tilde{S}^{\geq 1}$ is the minimal suffix of $\tilde{S}$ whose size according to $c$ is at least $1$, if we denote by $u$ the first element of this suffix, then both $\{u\}$ and $\tilde{S}^{\geq 1} - u$ are feasible solutions. Hence,
\[
	f(OPT)
	\geq
	\max\{f(\{u\}), f(\tilde{S}^{\geq 1} - u)\}
	\geq
	\frac{f(\{u\}) + f(\tilde{S}^{\geq 1} - u)}{2}
	\geq
	\frac{f(\tilde{S}^{\geq 1})}{2}
	\geq
	\frac{f(\tilde{S})}{4}
	\enspace,
\]
where the penultimate inequality follows from the submodularity of $f$, and the last inequality follows from Lemma~\ref{lem:suffix_exponential}. The corollary now follows by multiplying the last inequality by $4$.
\end{proof}

Recall now that the output of Algorithm~\ref{alg:estimate} is $f(\tilde{S}) / 4$. Therefore, Observation~\ref{obs:under_estimation} and Corollary~\ref{cor:over_estimation} imply together the following proposition.
\begin{proposition} \label{prop:estimation}
Algorithm~\ref{alg:estimate} runs in $O(n)$ time, and if we denote by $\Gamma$ its output, then $\Gamma \leq f(OPT) \leq 8 \cdot \Gamma$.
\end{proposition}

\subsection{Fast Threshold Greedy Algorithm} \label{sec:greedy}

In this section we present a variant, given as Algorithm~\ref{alg:threshold_greedy}, of the threshold-greedy algorithm of~\cite{badanidiyuru2014fast} that uses the estimate $\Gamma$ calculated by Algorithm~\ref{alg:estimate} to reduce the number of thresholds that need to be considered. Specifically, our algorithm considers an exponentially decreasing series of thresholds between $8\alpha \cdot \Gamma$ and $\Gamma / e$, where $\alpha \geq 1$ is a parameter of the algorithm. For every threshold considered, the algorithm adds to its solution every element such that (i) adding the element to the solution does not violate feasibility, and (ii) the density of the element with respect to the current solution exceeds the threshold. In addition to the parameter $\alpha$, Algorithm~\ref{alg:threshold_greedy} gets a quality control parameter $\eps \in (0, 1)$.
\begin{algorithm2e}[th]
\caption{\textsc{Fast Threshold Greedy}($\eps, \alpha$)} \label{alg:threshold_greedy}
Let $\Gamma$ be the output of Algorithm~\ref{alg:estimate}.\\
Let $\tau \gets 8\alpha \cdot \Gamma$, $h \gets 0$ and $S_0 \gets \varnothing$.\\
\While{$\tau > (1 - \eps) \Gamma / e$}
{
	\For{every element $u \in \cN$}
	{
		\If{$u \not \in S_h$, $c(u) + c(S_h) \leq 1$ and $\tfrac{f(u \mid S_h)}{c(u)} \geq \tau$\label{line:condition}}
		{
			Let $u_{h + 1} \gets u$ and $S_{h + 1} \gets S_h + u_{h + 1}$.\\
			Increase $h$ by $1$.
		}
	}
	Update $\tau \gets (1 - \eps)\tau$.
}
\Return{$S_h$}.
\end{algorithm2e}

We begin the analysis of Algorithm~\ref{alg:threshold_greedy} by bounding its time complexity.

\begin{observation} \label{obs:time_threshold_greedy}
The time complexity of Algorithm~\ref{alg:threshold_greedy} is $O(n \eps^{-1} \log \alpha)$.
\end{observation}
\begin{proof}
Recall that Algorithm~\ref{alg:estimate} runs in $O(n)$ time by Proposition~\ref{prop:estimation}, and observe that every iteration of the outer loop of Algorithm~\ref{alg:threshold_greedy} runs in $O(n)$ time. Therefore, to prove the lemma, it suffices to argue that the number of iterations made by this outer loop is $O(\eps^{-1} \log \alpha)$. Formally, the number of iterations made by this loop is upper bounded by
\[
	2 + \log_{(1 - \eps)^{-1}} \left(\frac{8\alpha \cdot \Gamma}{(1 - \eps)\Gamma / e}\right)
	=
	3 - \frac{\ln(8\alpha \cdot e)}{\ln(1 - \eps)}
	\leq
	3 + \frac{4 + \ln \alpha}{\eps}
	=
	O(\eps^{-1} \log \alpha)
	\enspace.
	\qedhere
\]
\end{proof}

Our next objective is to analyze the approximation guarantee of Algorithm~\ref{alg:threshold_greedy}. We do this by describing two lower bounds on the performance of this algorithm. The following lemma gives the simpler among these bounds. Let $\ell$ be the value of $k$ when Algorithm~\ref{alg:threshold_greedy} terminated.
\begin{lemma} \label{lem:unfilled_bound}
For every set $T \subseteq \cN$, if $\max_{u \in T} c(u) \leq 1 - c(S_\ell)$, then $f(S_\ell) \geq f(T) - \frac{c(T)}{e} \cdot f(OPT)$. In particular, $\max_{u \in OPT} c(u) \leq 1 - c(S_\ell)$ implies $f(S_\ell) \geq (1 - 1/e) \cdot f(OPT)$.
\end{lemma}
\begin{proof}
The second part of the lemma follows from the first part by plugging $T = OPT$ since $OPT$, as a feasible set, obeys $c(OPT) \leq 1$. Therefore, we concentrate below on proving the first part of the lemma.

Since elements are only ever added by Algorithm~\ref{alg:threshold_greedy} to its solution set $S_k$, the condition of the lemma implies that whenever an element $u \in T \setminus S_\ell$ was considered on Line~\ref{line:condition} of Algorithm~\ref{alg:threshold_greedy}, its density was too low compared to the value of the threshold $\tau$ at time. In particular, this is true for the threshold $\tau \leq \Gamma / e$ in the last iteration of the outer loop of Algorithm~\ref{alg:threshold_greedy}, and therefore,
\[
	\frac{f(u \mid S_u)}{c(u)}
	\leq
	\frac{\Gamma}{e}
	\qquad
	\forall\; e \in T \setminus S_\ell
	\enspace,
\]
where $S_u$ represents here the set $S_k$ immediately before $u$ is processed by Algorithm~\ref{alg:threshold_greedy} in the last iteration of its outer loop.

By the submodularity of $f$, the last inequality holds also when $S_u$ is replaced with $S_\ell$, and therefore,
\[
	f(T)
	\leq
	f(T \cup S_\ell)
	\leq
	f(S_{\ell}) + \sum_{u \in T \setminus S_{\ell}} \mspace{-9mu} f(u \mid S_\ell)
	\leq
	f(S_{\ell}) + \frac{\Gamma}{e} \cdot c(T \setminus S_{\ell})
	\leq
	f(S_{\ell}) + \frac{c(T) \cdot \Gamma}{e}
	\enspace,
\]
where the first inequality follows from the monotonicity of $f$, and the second inequality follows from the submodularity of $f$. The lemma now follows from the last inequality by recalling that Proposition~\ref{prop:estimation} guarantees $\Gamma \leq f(OPT)$.
\end{proof}

The following lemma gives the other lower bound we need on the performance of Algorithm~\ref{alg:threshold_greedy}. 
\begin{lemma} \label{lem:greedy_bound}
For every set $\varnothing \neq T \subseteq \cN$ and integer $0 \leq h < \ell$, if $\max_{u \in T} c(u) \leq 1 - c(S_h)$, then \[\frac{f(S_{h + 1}) - f(S_h)}{c(u_{h + 1})} \geq \min\left\{(1 - \eps) \cdot \frac{f(T \mid S_h)}{c(T)}, \alpha \cdot f(OPT)\right\}\enspace.\]
\end{lemma}
\begin{proof}
Let $\tau_{h + 1}$ be the value of $\tau$ when Algorithm~\ref{alg:threshold_greedy} selects $u_{h + 1}$. If $\tau_{h + 1} = 8\alpha \cdot \Gamma \geq \alpha \cdot f(OPT)$, then the lemma follows immediately from the condition used on Line~\ref{line:condition} of the algorithm. Therefore, we assume below that $\tau_{h + 1} < 8\alpha \cdot \Gamma$, which implies that no element of $T \setminus S_h$ was added to the solution of Algorithm~\ref{alg:threshold_greedy} when the threshold $\tau$ was equal to $(1 - \eps)^{-1}\tau_{h + 1}$. Therefore, since elements are only ever added by Algorithm~\ref{alg:threshold_greedy} to its solution, the submodularity of $f$ guarantees
\[
	\frac{f(u \mid S_h)}{c(u)}
	<
	\frac{\tau_{h + 1}}{1 - \eps}
	\qquad
	\forall\; u \in T \setminus S_h
	\enspace.
\]

We now observe that the lemma follows immediately from the monotonicity of $f$ when $T \subseteq S_h$. Therefore, we need to consider only the case of $T \not \subseteq S_h$, and in this case the previous inequality implies
\begin{align*}
	\frac{f(S_{h + 1}) - f(S_h)}{c(u_{h + 1})}
	={} &
	\frac{f(u_{h + 1} \mid S_h)}{c(u_{h + 1})}
	\geq
	\tau_{h + 1}
	=
	\frac{\sum_{u \in T \setminus S_h} c(u) \cdot \tau_{h + 1}}{c(T \setminus S_h)}\\
	>{} &
	(1 - \eps) \cdot \frac{\sum_{u \in T \setminus S_h} f(u \mid S_h)}{c(T \setminus S_h)}
	\geq
	(1 - \eps) \cdot \frac{f(T \mid S_h)}{c(T)}
	\enspace,
\end{align*}
where the last inequality follows from the submodularity of $f$.
\end{proof}

We are now ready to prove the following theorem, which implies Theorem~\ref{thm:cardinality_alg}.
\begin{theorem}
If we choose $\alpha = 1$, then, for every $\eps \in (0, 1)$, Algorithm~\ref{alg:threshold_greedy} runs in $O(n/\eps)$ time and guarantees $(1 - 1/e - \eps)$-approximation for ${\SMC}$.
\end{theorem}
\begin{proof}
The time complexity of Algorithm~\ref{alg:threshold_greedy} was already proved in Observation~\ref{obs:time_threshold_greedy}. Therefore, we concentrate here on analyzing the approximation guarantee of this algorithm. Let $k$ be the maximum cardinality allowed by the cardinality constraint (\ie, $c(u) = 1 / k$ for every $u \in \cN$). We need to consider two cases depending on the value of $\ell$. If $\ell < k$, \ie, Algorithm~\ref{alg:threshold_greedy} outputs a solution that is smaller than the maximum cardinality allowed, then the condition $\min_{u \in OPT} c(u) \leq 1 - c(S_\ell)$ of Lemma~\ref{lem:unfilled_bound} holds (because $c(S_\ell) = \ell / k \leq 1 - 1 / k$), and therefore,
\[
	f(S_\ell) \geq (1 - 1/e) \cdot f(OPT)
	\enspace.
\]

Assume now that $\ell = k$. In this case we need to use Lemma~\ref{lem:greedy_bound}. Since the condition of this lemma holds for every integer $0 \leq h < \ell$ and $T = OPT$ (we may assume that $OPT \neq \varnothing$ because otherwise Algorithm~\ref{alg:threshold_greedy} is clearly optimal), the lemma implies, for every such $k$,
\[
	\frac{f(S_{h + 1}) - f(S_h)}{c(u_{h + 1})}
	\geq
	\min\left\{(1 - \eps) \cdot \frac{f(OPT \mid S_h)}{c(OPT)}, f(OPT)\right\}
	\geq
	(1 - \eps) \cdot f(OPT \mid S_h)
	\enspace,
\]
where the second inequality holds because the feasibility of $OPT$ implies $c(OPT) \leq 1$, and the submodularity and non-negativity of $f$ imply $f(OPT \mid S_h) \leq f(OPT \mid \varnothing) \leq f(OPT)$. Rearranging the last inequality yields
\begin{align*}
	[1 - (1 - \eps) &{}\cdot c(u_{h + 1})] \cdot [f(OPT) - f(S_h)]\\
	\geq{} &
	[1 - (1 - \eps) \cdot c(u_{h + 1})] \cdot f(OPT) + (1 - \eps) \cdot c(u_{h + 1}) \cdot f(OPT \cup S_h) - f(S_{h + 1})
	\enspace,
\end{align*}
which by the monotonicity of $f$, the observation that $f(S_{h}) \leq f(OPT)$ because $S_{h}$ is a feasible solution, and the inequality $1 - x \leq e^{-x}$ (which holds for every $x \in \bR$) implies
\[
	e^{-(1 - \eps) \cdot c(u_{k + 1})} \cdot [f(OPT) - f(S_k)]
	\geq
	f(OPT) - f(S_{k + 1})
	\enspace.
\]
Finally, by combining the last inequality for all integers $0 \leq h < \ell$, we get
\begin{align*}
	f(OPT) - f(S_{\ell})
	\leq{} &
	e^{-(1 - \eps) \cdot \sum_{0 \leq h < \ell} c(u_{h + 1})} \cdot [f(OPT) - f(S_0)]\\
	={} &
	e^{-(1 - \eps)} \cdot [f(OPT) - f(S_0)]
	\leq
	e^{-(1 - \eps)} \cdot f(OPT)
	\leq
	(e^{-1} + \eps) \cdot f(OPT)
	\enspace,
\end{align*}
where the equality holds since $\sum_{0 \leq h < \ell} c(u_{h + 1}) = c(S_\ell) = 1$ by the definition of the case we consider, and the penultimate inequality follows from the non-negativity of $f$. The theorem now follows by rearranging the last inequality.
\end{proof}

\subsection{General Knapsack Constraint in Clean Linear Time} \label{sec:plus}

In this section we prove Theorem~\ref{thm:knapsack}. Unfortunately, Algorithm~\ref{alg:threshold_greedy} by itself does not guarantee any constant approximation guarantee for general knapsack constraints.\footnote{To see this, consider an input instance consisting of an element $u$ of size $\delta > 0$, an additional element $w$ of size $1$ and an objective function $f(S) = 2\delta \cdot |S \cap \{u\}| + |S \cap \{w\}|$. For small enough $\eps$ and $\delta$ values, Algorithm~\ref{alg:threshold_greedy} will pick the solution $\{u\}$ of value $2\delta$ instead of the much more valuable solution $\{w\}$ whose value is $1$.} However, this can be fixed using the post-processing step described by Algorithm~\ref{alg:post_processing}. Recall that Algorithm~\ref{alg:threshold_greedy} maintains a solution $S_k$ that grows as the algorithm progresses. The post-processing step takes $O(\eps^{-1} \log \eps^{-1})$ snapshots of this solution at various times during the execution of Algorithm~\ref{alg:threshold_greedy}, and then tries to augment each one of these snapshots with a single element.
\begin{algorithm2e}
\caption{\textsc{Fast Threshold Greedy + Post-Processing}$(\eps)$} \label{alg:post_processing}
Execute Algorithm~\ref{alg:threshold_greedy} with $\alpha = \eps^{-1}$. Store all the sets $S_0, S_1, \dotsc, S_\ell$ produced by this execution.\\
\For{$i = 0$ \KwTo $\lfloor \log_{1 + \eps} \eps^{-1} \rfloor$}
{
	Let $S^{(i)}$ be the set $S_k$ for the maximal $k$ value for which $c(S_k) \leq \eps (1 + \eps)^i$.\\
	\If{there exists an element $u \in \cN$ such that $c(u) \leq 1 - c(S^{(i)})$}
	{
		Let $u^{(i)}$ be an element maximizing $f(u^{(i)} \mid S^{(i)})$ among all the elements obeying the condition on the previous line.\\
		Let $S^{(i)+} \gets S^{(i)} + u^{(i)}$.
	}
	\Else
	{
		Let $S^{(i)+} \gets S^{(i)}$.
	}
}
\Return{the set maximizing $f$ in $\{S_\ell\} \cup \{\{u\} \mid u \in \cN\} \cup \{S^{(i)+} \mid 0 \leq i  \leq \lfloor \log_{1 + \eps} \eps^{-1} \rfloor \}$}.\label{line:output}
\end{algorithm2e}

We begin the analysis of Algorithm~\ref{alg:post_processing} by analyzing its time complexity.
\begin{observation}
Algorithm~\ref{alg:post_processing} runs in $O(n \eps^{-1} \log \eps^{-1})$ time.
\end{observation}
\begin{proof}
Algorithm~\ref{alg:threshold_greedy} runs in $O(n \eps^{-1} \log \eps^{-1})$ time by Observation~\ref{obs:time_threshold_greedy}. Additionally, every iteration of the loop of Algorithm~\ref{alg:post_processing} runs in $O(n + \ell) = O(n)$ time, where the equality holds because the fact that some element is added to the solution of Algorithm~\ref{alg:threshold_greedy} whenever the index $k$ increases guarantees that $\ell$ is at most $n$. Therefore, to prove the observation it suffices to show that the loop of Algorithm~\ref{alg:post_processing} iterates only $O(\eps^{-1} \log \eps^{-1})$ times. Formally, the number of iterations of this loop is upper bounded by
\[
	1 + \log_{1 + \eps} \eps^{-1}
	=
	1 + \frac{\ln \eps^{-1}}{\ln(1 + \eps)}
	\leq
	1 + \frac{\ln \eps^{-1}}{\eps / 2}
	=
	O(\eps^{-1} \ln \eps^{-1})
	\enspace.
	\qedhere
\]
\end{proof}

In the rest of this section we concentrate on analyzing the approximation ratio guaranteed by Algorithm~\ref{alg:post_processing}. Consider the following inequality.
\begin{equation} \label{eq:assumption}
	f(S^{(i)+}) < \nicefrac{1}{2} \cdot f(OPT) \qquad \forall\; 0 \leq i \leq \log_{1 + \eps} \eps^{-1}
	\enspace.
\end{equation}
If this inequality does not hold for some integer $i$, then Algorithm~\ref{alg:post_processing} clearly obtains at least $\nicefrac{1}{2}$-approx\-imation. Therefore, it remains to consider the case in which Inequality~\eqref{eq:assumption} applies for every integer $0 \leq i \leq \log_{1 + \eps} \eps^{-1}$. One can observe that this inequality implies that $OPT$ is non-empty because $f(S^{(0)}) \geq f(\varnothing)$ by the monotonicity of $f$, and therefore, we can define $r$ to be an element of $OPT$ maximizing $c(r)$. Additionally, we assume below that $r$ is not the only element of $OPT$---if this assumption does not hold, then Algorithm~\ref{alg:post_processing} trivially returns an optimal solution.
\begin{lemma} \label{lem:huge_r}
If $c(r) \geq 1 - \eps$, then $f(S_\ell) \geq (\nicefrac{1}{2} - \eps) \cdot f(OPT)$.
\end{lemma}
\begin{proof}
If $f(\{r\}) \geq \nicefrac{1}{2} \cdot f(OPT)$, then the lemma follows immediately since $\{r\}$ is one of the sets considered for the output of Algorithm~\ref{alg:post_processing} on Line~\ref{line:output}. Therefore, we may assume that $f(\{r\}) \leq \nicefrac{1}{2} \cdot f(OPT)$, which implies $f(OPT - r) \geq f(OPT) - f(\{r\}) \geq \nicefrac{1}{2} \cdot f(OPT)$ by the submodularity of $f$. Since $c(OPT) \leq 1$ because $OPT$ is a feasible set, we get that $OPT - r$ is a set with a lot of value taking a very small part of the budget allowed (specifically, $c(OPT - r) = c(OPT) - c(r) \leq 1 - (1 - \eps) = \eps$). Below we show that the existence of such a set implies that $f(S_\ell)$ is large. We also assume below $f(S_\ell) \leq f(OPT - r)$ because otherwise the lemma follows immediately since $f(OPT - r) \geq \nicefrac{1}{2} \cdot f(OPT)$.

Consider first the case in which $c(S_\ell) \leq 1 - \eps$. Since $\max_{u \in OPT - r} c(u) \leq c(OPT - r) \leq \eps \leq 1 - c(S_\ell)$ in this case, Lemma~\ref{lem:unfilled_bound} gives us immediately, for $T = OPT - r$,
\[
	f(S_\ell)
	\geq
	f(OPT - r) - \frac{c(OPT - r)}{e} \cdot f(OPT)
	\geq
	\frac{f(OPT)}{2} - \frac{\eps \cdot f(OPT)}{e}
	\geq
	(\nicefrac{1}{2} - \eps) \cdot f(OPT)
	\enspace.
\]
It remains to consider the case of $c(S_\ell) \geq 1 - \eps$. Let $\bar{h}$ be the minimal $h$ such that $c(S_h) \geq 1 - \eps$. For every integer $0 \leq h < \bar{h}$, we have $\max_{u \in OPT - r} \leq c(OPT - r) \leq \eps \leq 1 - c(S_h)$, and therefore, applying Lemma~\ref{lem:greedy_bound} with $T = OPT - r$ yields
\[
	\frac{f(S_{h + 1}) - f(S_h)}{c(u_{h + 1})} 
	\geq
	\min\left\{(1 - \eps) \cdot \frac{f(OPT - r \mid S_h)}{c(OPT - r)}, \eps^{-1} \cdot f(OPT)\right\}
	\geq
	\frac{1 - \eps}{\eps} \cdot f(OPT - r \mid S_h)
	\enspace,
\]
where the second inequality holds because $c(OPT - r) \leq \eps$, and the monotonicity, submodularity and non-negativity of $f$ imply $f(OPT - r \mid S_h) \leq f(OPT \mid S_h) \leq f(OPT \mid \varnothing) \leq f(OPT)$. Rearranging the last inequality yields
\begin{align*}
	\mspace{36mu}&\mspace{-36mu}
	\left[1 - \frac{1 - \eps}{\eps} {}\cdot c(u_{h + 1})\right] \cdot [f(OPT - r) - f(S_h)]\\
	\geq{} &
	\left[1 - \frac{1 - \eps}{\eps} \cdot c(u_{h + 1})\right] \cdot f(OPT - r) + \frac{1 - \eps}{\eps} \cdot c(u_{h + 1}) \cdot f((OPT - r) \cup S_h) - f(S_{h + 1})
	\enspace,
\end{align*}
which by the monotonicity of $f$, our assumption that $f(S_h) \leq f(S_\ell) \leq f(OPT - r)$ and the inequality $1 - x \leq e^{-x}$ (which holds for every $x \in \bR$) implies
\[
	e^{-\eps^{-1}(1 - \eps) \cdot c(u_{h + 1})} \cdot [f(OPT - r) - f(S_h)]
	\geq
	f(OPT - r) - f(S_{h + 1})
	\enspace.
\]
Finally, by combining the last inequality for all integers $0 \leq h < \bar{h}$, we get
\begin{align*}
	f(OPT - r) - f(S_{\bar{h}})
	\leq{} &
	e^{-\eps^{-1}(1 - \eps) \cdot \sum_{0 \leq h < \bar{h}} c(u_{h + 1})} \cdot [f(OPT - r) - f(S_0)]\\
	\leq{} &
	e^{-\eps^{-1}(1 - \eps)^2} \cdot [f(OPT - r) - f(S_0)]\\
	\leq{} &
	e^{-\eps^{-1}(1 - \eps)^2} \cdot f(OPT - r)
	\leq
	2\eps \cdot f(OPT - r)
	\enspace,
\end{align*}
where the second inequality holds since $\sum_{0 \leq h < \bar{h}} c(u_{h + 1}) = c(S_{\bar{h}}) \geq 1 - \eps$ by the definition of $\bar{h}$ and the penultimate inequality follows from the non-negativity of $f$.

Rearranging the last inequality, we now get, using the monotonicity of $f$,
\[
	f(S_\ell)
	\geq
	f(S_{\bar{h}})
	\geq
	(1 - 2\eps) \cdot f(OPT - r)
	\geq
	(\nicefrac{1}{2} - \eps) \cdot f(OPT)
	\enspace.
	\qedhere
\]
\end{proof}

Notice that so far we have proved that the approximation ratio of Algorithm~\ref{alg:post_processing} is at least $\nicefrac{1}{2} - \eps$ when either Inequality~\eqref{eq:assumption} is violated or $c(r) \geq 1 - \eps$. The following lemma proves that the same approximation guarantee applies also when neither of these cases applies, which completes the proof of Theorem~\ref{thm:knapsack}.

\begin{lemma} \label{lem:bad_r_addition}
If $c(r) \leq 1 - \eps$, then Inequality~\eqref{eq:assumption} implies $f(S_\ell) \geq (\nicefrac{1}{2} - \eps) \cdot f(OPT)$.
\end{lemma}
\begin{proof} 
Let $i_r$ be the maximal integer such that $\eps (1 + \eps)^{i_r} \leq 1 - c(r)$. Since we assume that $c(r) \leq 1 - \eps$, $i_r$ is non-negative. On the other hand, we also have $i_r \leq \log_{1 + \eps} \eps^{-1} $ because otherwise $\eps (1 + \eps)^{i_r} > 1 \geq 1 - c(r)$. Therefore, the set $S^{(i_r)}$ is well-defined, and one can observe that
\begin{equation} \label{eq:bad_r_addition}
	f(S^{(i_r)} + r)
	=
	f(S^{(i_r)}) + f(r \mid S^{(i_r)})
	\leq
	f(S^{(i_r)}) + f(u^{(i_r)} \mid S^{(i_r)})
	=
	f(S^{(i_r)+})
	<
	\nicefrac{1}{2} \cdot f(OPT)
	\enspace,
\end{equation}
where the first inequality follows from the way in which Algorithm~\ref{alg:post_processing} selects $u^{(i_r)}$ because $c(S^{(i_r)}) \leq \eps (1 + \eps)^{i_r} \leq 1 - c(r)$, and the second inequality follows from Inequality~\eqref{eq:assumption}.

Let $h_r$ be the index for which $S_{h_r} = S^{(i_r)}$. If $h_r = \ell$, then $f(S_\ell) \geq (1 - 1/e) \cdot f(OPT)$ by Lemma~\ref{lem:unfilled_bound} (for $T = OPT$) because $\max_{u \in OPT} c(u) = c(r) \leq 1 - c(S^{(i_r)}) = 1 - c(S_\ell)$. Therefore, we may assume $h_r < \ell$ in the rest of this proof. In particular, this implies that for every $0 \leq h \leq h_r$, by plugging $T = OPT - r$ into Lemma~\ref{lem:greedy_bound}, we get
\begin{align*}
	\frac{f(S_{h + 1}) - f(S_h)}{c(u_{h + 1})}
	\geq{} &
	\min\left\{(1 - \eps) \cdot \frac{f(OPT - r \mid S_h)}{c(OPT - r)}, \eps^{-1} \cdot f(OPT)\right\}\\
	\geq{} &
	\min\left\{(1 - \eps) \cdot \frac{f(OPT \mid S_h + r)}{c(OPT - r)}, \eps^{-1} \cdot f(OPT)\right\}\\
	\geq{} &
	\min\left\{(1 - \eps) \cdot \frac{f(OPT) - f(S_{h_r} + r)}{c(OPT - r)}, \eps^{-1} \cdot f(OPT)\right\}\\
	\geq{} &
	(1 - \eps) \cdot \frac{f(OPT) - f(S_{h_r} + r)}{1 - c(r)}
	\enspace,
\end{align*}
where the second inequality follows from the submoduarity of $f$; the third inequality follows from $f$'s monotonicity; and the last inequality holds since $f(S_{h_r} + r) \geq 0$, and $\max\{c(OPT - r), \eps\} \leq 1 - c(r)$ by the condition of the lemma and the fact that $OPT$ is a feasible solution. Plugging now Inequality~\eqref{eq:bad_r_addition} into the last inequality (and recalling that $S_{h_r} = S^{(i_r)}$) yields
\[
	\frac{f(S_{h + 1}) - f(S_h)}{c(u_{h + 1})}
	\geq
	(1 - \eps) \cdot \frac{\nicefrac{1}{2} \cdot f(OPT)}{1 - c(r)}
	\enspace.
\]

Summing up the last inequality for all integers $0 \leq h \leq h_r$ now gives
\begin{align*}
	f(S_{h_r + 1}) - f(S_0)
	\geq{} &
	\sum_{h = 0}^{h_r} [f(S_{h + 1}) - f(S_h)]
	\geq
	(1 - \eps) \cdot \frac{\nicefrac{1}{2} \cdot f(OPT)}{1 - c(r)} \cdot \sum_{h = 0}^{h_r} c(u_{h + 1})\\
	={} &
	(1 - \eps) \cdot \frac{\nicefrac{1}{2} \cdot f(OPT)}{1 - c(r)} \cdot c(S_{h_r + 1})
	\geq
	(1 - \eps) \cdot \frac{\nicefrac{1}{2} \cdot f(OPT)}{1 - c(r)} \cdot \eps (1 + \eps)^{i_r}\\
	\geq{} &
	(1 - \eps) \cdot \frac{\nicefrac{1}{2} \cdot f(OPT)}{1 - c(r)} \cdot \frac{1 - c(r)}{1 + \eps}
	\geq
	(\nicefrac{1}{2} - \eps) \cdot f(OPT)
	\enspace,
\end{align*}
where the third inequality holds since the definition of $S^{(i_r)}$ implies that $h_r$ is the maximal index for which $c(S_{h_r}) \leq \eps (1 + \eps)^{i_r}$, and the penultimate inequality follows from the definition of $i_r$.
\end{proof}

\section{Information-Theoretic Inapproximability Results} \label{sec:inapproximability}

In this section we prove our inapproximability results. These results assume that the algorithm is able to access the objective function $f$ only through a value oracle that given a set $S \subseteq \cN$ returns $f(S)$. This assumption allows us to prove inapproximability results that are based on information theoretic arguments, and therefore, are independent of any complexity assumptions.

The first result that we show is inapproximability for {\SMC} (see Section~\ref{sec:SMK} for the definition of this problem) showing  that, when the number $k$ of elements allowed in the solution is constant, one cannot obtain a constant approximation ratio for this problem using $o(n)$ value oracle queries.
\begin{restatable}{theorem}{thmConstraintCardinalityLowerBound} \label{thm:constant_cardinality_lower_bound}
Any (possibly randomized) algorithm guaranteeing an approximation ratio of $\alpha \in (0, 1]$ for {\SMCFull} (\SMC) must use $\Omega(\alpha n / k)$ value oracle queries. In particular, this implies that the algorithm must make $\Omega(n)$ value oracle queries when $\alpha$ and $k$ are constants.
\end{restatable}

See Section~\ref{ssc:constant_cardinality_lower_bound} for the proof of Theorem~\ref{thm:constant_cardinality_lower_bound}. Let us now consider the approximability of {\SMC} in a different regime allowing larger values of $k$. Specifically, we consider the regime in which the ratio between $n$ and $k$ is some constant $\beta$. In this regime one can obtain $\beta$-approximation by simply outputting a uniformly random subset of $\cN$ of size $k$ (which requires no value oracle queries); thus, this regime cannot admit an extremely pessimistic inapproximability result like the one given by Theorem~\ref{thm:constant_cardinality_lower_bound}. Nevertheless, the following theorem shows that improving over the above-mentioned ``easy'' approximation guarantee requires a nearly-linear query complexity.

\begin{theorem} \label{thm:cardinality_lower_bound}
For any rational constant $\beta \in (0, 1)$ and constant $\eps > 0$, every (possibly randomized) algorithm for {\SMCFull} (\SMC) that guarantees an approximation ratio of $\beta + \eps$ for instances obeying $k = \beta n$ must use $\Omega(\tfrac{n}{\log n})$ value oracle queries. Moreover, this is true even when the objective function $f$ of {\SMC} is guaranteed to be a linear function.
\end{theorem}

We prove in Section~\ref{sec:cardinality_lower_deterministic} a restricted version of Theorem~\ref{thm:cardinality_lower_bound} that applies only to deterministic algorithms. Then, Section~\ref{sec:cardinality_lower_randomized} extends the proof to randomized algorithms, which implies Theorem~\ref{thm:cardinality_lower_bound}. Table~\ref{tbl:cardinality_state} summarizes our knowledge about the number of value oracle queries necessary for obtaining various approximation guarantees for {\SMC}.

\begin{table}
\begin{center}
\caption{Number of value oracle queries that are required and sufficient to obtain any given approximation ratio for {\SMCFull} (\SMC). The variable $\eps$ should be understood as an arbitrarily small constant.} \label{tbl:cardinality_state}
\medskip
\begin{tabular}{c|c|m{9cm}}
\multirow{2}{*}{\textbf{Regime}} & \textbf{Approximation Ratio} & \multicolumn{1}{c}{\textbf{Query Complexity}}\\
& \textbf{Range} &\\
\hline
\multirow[b]{2}{*}[-1mm]{Constant $k$} & $[\eps, 1 - 1/e - \eps]$ & Can be obtained using $O(n)$ queries by Theorem~\ref{thm:cardinality_alg} or~\cite{mirzasoleiman2015lazier}, and requires $\Omega(n)$ queries by Theorem~\ref{thm:constant_cardinality_lower_bound}. \\[5mm]
& $[1 - 1/e, 1]$ & Exact solution can be obtained using $O(n^k)$ queries via brute force; $\Omega(n)$ queries are necessary by Theorem~\ref{thm:constant_cardinality_lower_bound}. \\
\hline
\multirow[b]{2}{*}[-12mm]{$k = \beta n$} & $[0, \beta]$ & Can be obtained by outputting a uniformly random subset of the right size ($0$ queries).\\[5mm]
& $[\beta + \eps, 1 - 1/e - \eps]$ & Can be obtained using $O(n)$ queries by Theorem~\ref{thm:cardinality_alg} or~\cite{mirzasoleiman2015lazier}, and requires $\Omega(\frac{n}{\log n})$ queries by Theorem~\ref{thm:cardinality_lower_bound}.\\[5mm]
& $1 - 1/e$ & Can be obtained by the greedy algorithm using $O(nk)$ queries~\cite{nemhauser1078analysis}; requires $\Omega(\frac{n}{\log n})$ queries by Theorem~\ref{thm:cardinality_lower_bound}.\\[5mm]
& $[1 - 1/e + \eps, 1]$ & Requires exponentially many queries~\cite{nemhauser1978best}.
\end{tabular}
\end{center}
\end{table}

Using the technique of Theorem~\ref{thm:cardinality_lower_bound}, we can also prove another interesting result. In the {\USMFull} problem (\USM), the input consists of a non-negative (not necessarily monotone) submodular function $f\colon 2^\cN \to \nnR$. The objective is to output a set $S \subseteq \cN$ maximizing $f$. Our result for {\USM} is given by the next theorem. This theorem nearly completes our understanding of the minimum query complexity necessary to obtain various approximation ratios for {\USM} (see Table~\ref{tbl:usm_results} for more detail).
\begin{restatable}{theorem}{ThmUnconstrainedLowerBound} \label{thm:unconstrained_lower_bound}
For any constant $\eps > 0$, every algorithm for {\USMFull} (\USM) that guarantees an approximation ratio of $\nicefrac{1}{4} + \eps$ must use $\Omega(\tfrac{n}{\log n})$ value oracle queries.
\end{restatable}

The proof of Theorem~\ref{thm:unconstrained_lower_bound} can be found in Section~\ref{ssc:unconstrained_lower_bound}. The necessity of the $\log n$ term in Theorems~\ref{thm:cardinality_lower_bound} and~\ref{thm:unconstrained_lower_bound} is discussed in Appendix~\ref{app:technique_power}. Specifically, we show in this appendix an algorithm that can be used to solve (exactly) the hard problem underlying these theorems using $O(n / \log n)$ oracle queries. This algorithm is not efficient in terms of its time complexity, but it shows that one cannot prove a lower bound stronger than $\Omega(n / \log n)$ for this problem based on information theoretic arguments alone.

\begin{table}
\begin{center}
\caption{Number of value oracle queries required and sufficient to obtain any given approximation ratio for {\USMFull} (\USM). The variable $\eps$ should be understood as an arbitrarily small constant.} \label{tbl:usm_results}
\medskip
\begin{tabular}{c|c|c}
\textbf{Approximation Ratio} & \textbf{Required Query} & \textbf{Sufficient Query}\\
\textbf{Range} & \textbf{Complexity} & \textbf{Complexity}\\
\hline
$[0, \nicefrac{1}{4}]$ & $0$ & \multicolumn{1}{m{2.1in}}{\hfill$0$\hspace*{\fill}\newline An approximation ratio of $\nicefrac{1}{4}$ can be obtained by outputting a uniformly random subset of $\cN$~\cite{feige2011maximizing}.}\\
\hline
$[\nicefrac{1}{4} + \eps, \nicefrac{1}{2}]$ & $\Omega(\frac{n}{\log n})\quad$(Theorem~\ref{thm:unconstrained_lower_bound}) & \multicolumn{1}{m{2.1in}}{\hfill$O(n)$\hspace*{\fill}\newline Obtained by the Double Greedy algorithm of~\cite{buchbinder2015tight}.}\\
\hline
$[\nicefrac{1}{2} + \eps, 1]$ & \multicolumn{1}{m{2.1in}|}{Requires an exponential number of queries~\cite{feige2011maximizing}.} & --
\end{tabular}
\end{center}
\end{table}

\subsection{Inapproximability of {\SMC} for Constant \texorpdfstring{$k$}{k} Values} \label{ssc:constant_cardinality_lower_bound}

In this section we prove Theorem~\ref{thm:constant_cardinality_lower_bound}, which we repeat here for convinience.
\thmConstraintCardinalityLowerBound*
\begin{proof}
Let $t = 2k / \alpha$, and for every element $u \in \cN$, consider the set function $f_u\colon \cN \to \nnR$ defined, for every set $S \subseteq \cN$, as
\[
	f_u(S) =
	\begin{cases}
		\min\{t, |S|\} & \text{if $u \not \in S$} \enspace,\\
		t & \text{if $u \in S$} \enspace.
	\end{cases}
\]
Below we consider a distribution $\cD$ over instances of {\SMC} obtained by choosing $u$ uniformly at random out of $\cN$, and we show that this distribution is hard for every deterministic algorithm that uses $o(\alpha n / k)$ value oracle queries, which implies the theorem by Yao's principle.

Let $ALG$ be any deterministic algorithm using $o(\alpha n / k)$ value oracle queries, and let $S_1, S_2, \dotsc, S_\ell$ be the sets on which $ALG$ queries the value oracle of the objective function when this function is chosen to be $\min\{t, |S|\}$. We assume without loss of generality that $S_{\ell}$ is the output set of $ALG$. Let $E$ be the event that $u$ does not belong to any of sets $S_1, S_2, \dotsc, S_\ell$ that are of size less than $t$. 
Since $ALG$ is deterministic, whenever the event $E$ happens, $ALG$ makes the same set of value oracle queries when given either $\min\{t, |S|\}$ or $f_u$ as the objective function, and outputs $f_u(S_\ell) = \min\{t, |S|\} \leq |S_\ell| \leq k$ in both cases. Moreover, even when the event $E$ does not happen, $ALG$ still outputs a set of value at most $t$ since $f_u$ never returns larger value. Finally, observe that the probability of the event $E$ is at least $1 - \ell t / n$, and therefore, when the input instance is drawn from the distribution $\cD$, the expected value of the output of $ALG$ is at most
\[
	\frac{\ell t}{n} \cdot t + \left(1 - \frac{\ell t}{n} \right) \cdot k
	=
	\frac{\ell t^2}{n} + k - \frac{\ell tk}{n}
	\leq
	\frac{\ell t^2}{n} + k
	\enspace.
\]

It remains to observe that the optimal solution for every instance in the support of the distribution $\cD$ is $\{u\}$, whose value according to $f_u$ is $t$. Therefore, the approximation ratio of $ALG$ with respect to an instance of {\SMC} drawn from $\cD$ is no better than
\[
	\frac{\ell t^2/n + k}{t}
	=
	\frac{\ell t}{n} + \frac{k}{t}
	=
	\frac{2\ell k}{\alpha n} + \frac{\alpha}{2}
	=
	o(1) + \frac{\alpha}{2}
	\enspace,
\]
where the last equality holds since $ALG$ makes $o(\alpha n / k)$ queries, and $\ell$ is the number of queries that it makes given a particular input. Hence, $\cD$ is a hard distribution in the sense that $ALG$ is not an $\alpha$-approximation algorithm against an instance of {\SMC} drawn from $\cD$.
\end{proof}

\subsection{Inapproximability of {\SMC} for Deterministic Algorithms and Large \texorpdfstring{$k$}{k} Values} \label{sec:cardinality_lower_deterministic}

We begin this section by considering the following problem, termed {\SI}, which has a rational parameter $\beta \in (0, 1)$. An algorithm for this problem has access to a ground set $\cN$ of size $n$ and a non-empty collection $\cC$ of subsets of $\cN$ of size $k = \beta n$ (formally, $\cC \subseteq \{S \subseteq \cN \mid |S| = k\})$. Each instance of {\SI} also includes some hidden set $C^* \in \cC$, which the algorithm can access only via an oracle that answers the following type of queries: given a set $S \subseteq \cN$, the oracle returns the size of $S \cap C^*$. The objective of the algorithm is to output a set $S$ of size $k$ whose intersection with $C^*$ is as large as possible.

The majority of this section is devoted to proving the following proposition regarding deterministic algorithms for {\SI}. We then show that this proposition implies a version of Theorem~\ref{thm:cardinality_lower_bound} for deterministic algorithms.
\begin{proposition} \label{prop:set_identification}
For every $\alpha \in (\beta, 1]$ and set $\cC$, any deterministic $\alpha$-approximation algorithm for {\SI} must make $\frac{(k + 1) \cdot [2(\alpha - \beta)^2 + \ln \beta + (\beta^{-1} - 1) \ln (1 - \beta) - k^{-1}\ln n] + \ln |\cC|}{\ln (k + 1)}$ oracle queries.
\end{proposition}

Let $ALG$ be an arbitrary deterministic algorithm for {\SI}. To prove Proposition~\ref{prop:set_identification}, we need to design an adversary that selects a hidden set $C^*$ that is ``bad'' for $ALG$. Our adversary does not immediately commit to a particular set $C^*$. Instead, it maintains a list $\cL$ of candidate sets that are consistent with all the oracle queries that $ALG$ has performed so far. In other words, the set $\cL$ originally includes all the sets of $\cC$. Then, whenever $ALG$ queries the oracle on some set $S$, the adversary returns some answer $a$ and eliminates from $\cL$ every set $C$ for which $|S \cap C| \neq a$. To fully define the adversary, we still need to explain how it chooses the answer $a$ for the oracle query. As the adversary wishes to keep the algorithm in the dark for as long as possible, it chooses the answer $a$ to be the one that reduces the list $\cL$ by the least amount.

Let $\cL_i$ be the list $\cL$ after the adversary answers $i$ queries.
\begin{observation} \label{obs:query_reduction}
For every integer $i \geq 0$, $|\cL_i| \geq |\cC| / (k + 1)^i$.
\end{observation}
\begin{proof}
Whenever the adversary has to answer a query regarding some set $S$, the answer corresponding to each set in the list $\cL$ is some number in $\{0, 1, \dotsc, k\}$. Therefore, at least one of these answers is guaranteed to reduce the list only by a factor of $|\{0, 1, \dotsc, k\}| = k + 1$. Thus, after $i$ queries have been answered, the size of the list $\cL$ reduces by at most a factor of $(k + 1)^i$. To complete the proof of the observation, we recall that the original size of $\cL$ is $|\cC|$.
\end{proof}

Once $ALG$ terminates and outputs some set $S$ of size $k$, our adversary has to decide what set of $\cL$ is the ``real'' hidden set $C^*$, which it does by simply choosing the set $C^* \in \cL$ with the least intersection with $S$. If the list $\cL$ is still large at this point, then the adversary is guaranteed to be able to find in it a set with a low intersection with $S$. The following lemma quantifies this observation.
\begin{lemma} \label{lem:bad_set_guarantee}
For every set $S \subseteq \cN$ of size $k$, a non-empty subset $\cL \subseteq \cC$ must include a set $C \in \cL$ such that $|S \cap C| < \beta k + k\sqrt{(k^{-1}\ln n - \ln \beta - (\beta^{-1} - 1) \ln (1 - \beta) - \tfrac{1}{k + 1}\ln |\cL|) / 2}$.
\end{lemma}
\begin{proof}
Let $h = \beta k + k \cdot \sqrt{(k^{-1}\ln n - \ln \beta - (\beta^{-1} - 1) \ln (1 - \beta) - \tfrac{1}{k + 1}\ln |\cL|) / 2}$. We would like to upper bound the probability of a uniformly random subset $R \subseteq \cN$ of size $k$ to intersect at least $h$ elements of $S$. Since the expected size of the intersection of $R$ and $S$ is
\[
	\frac{|R| \cdot |S|}{|\cN|}
	=
	\frac{k^2}{n}
	=
	\beta k
	\enspace,
\]
and the distribution of $R$ is hyper-geometric, we get by Inequality~(10) in~\cite{skala2020hypergeometric} (developed based on the works of~\cite{chvatal1979tail,hoeffding1963probability}) that
\begin{equation} \label{eq:fraction_large_intersection}
	\Pr[|R \cap S| \geq h]
	=
	\Pr[|R \cap S| \geq \bE[|R \cap S|] + h - \beta k]
	\leq 
	e^{-2((h - \beta k) / k)^2 \cdot k}
	=
	e^{-2(h - \beta k)^2 / k}
	\enspace.
\end{equation}

By Stirling's approximation, the total number of subsets of $\cN$ whose size is $k$ is
\begin{align*}
	\binom{n}{k}
	={} &
	\frac{n!}{k! \cdot (n - k)!}
	\leq
	\frac{en(n/e)^n}{e(k/e)^k \cdot e((n - k)/e)^{n - k}}\\
	={} &
	\frac{n}{e\beta^k \cdot (1 - \beta)^{n - k}}
	\leq
	\frac{n}{\beta^k (1 - \beta)^{n - k}}
	\enspace.
\end{align*}
Since Inequality~\eqref{eq:fraction_large_intersection} implies that at most a fraction of $e^{-2(h - \beta k)^2 / k}$ out of these sets has intersection with $S$ of size at least $h$, we get that the number of subsets of $\cN$ of size $k$ whose intersection with $S$ is of size at least $h$ can be upper bounded by
\[
	e^{-2(h - \beta k)^2 / k} \cdot \binom{n}{k}
	\leq
	\frac{ne^{-2(h - \beta k)^2 / k}}{\beta^k (1 - \beta)^{n - k}}
	=
	\frac{ne^{-(\ln n - k \ln \beta - (n - k) \cdot \ln(1 - \beta) - \tfrac{k}{k + 1}\ln |\cL|)}}{\beta^k (1 - \beta)^{n - k}}
	=
	|\cL|^{\frac{k}{k + 1}}
	<
	|\cL|
	\enspace.
\]
The last bound implies that $|\cL|$ must include at least one set whose intersection with $S$ is less than $h$.
\end{proof}

We are now ready to prove Proposition~\ref{prop:set_identification}.
\begin{proof}[Proof of Proposition~\ref{prop:set_identification}]
By Observation~\ref{obs:query_reduction} and Lemma~\ref{lem:bad_set_guarantee}, if $ALG$ makes $i$ oracle queries, our adversary will be able to choose a hidden set $C^*$ such that the intersection between the output set $S$ of $ALG$ and $C^*$ is less than
\begin{align*}
	&
	\beta k + k\sqrt{(k^{-1}\ln n - \ln \beta - (\beta^{-1} - 1) \ln (1 - \beta) - \tfrac{1}{k + 1}\ln |\cL_i|) / 2}\\
	\leq{} &
	\beta k + k \cdot \sqrt{\frac{1}{2} \cdot \left(k^{-1}\ln n - \ln \beta - (\beta^{-1} - 1) \ln(1 - \beta) + \frac{i \cdot \ln(k + 1) - \ln |\cC|}{k + 1}\right)}
	\enspace.
\end{align*}
Since the optimal solution for {\SI} is the set $C^*$ itself, whose intersection with itself is $k$, the above inequality implies that the approximation ratio of $ALG$ is worse than
\[
	\beta + \sqrt{\frac{1}{2} \cdot \left(k^{-1}\ln n - \ln \beta - (\beta^{-1} - 1) \ln(1 - \beta) + \frac{i \cdot \ln(k + 1) - \ln |\cC|}{k + 1}\right)}
	\enspace.
\]

If the approximation ratio of $ALG$ is at least $\alpha$, then we get the inequality
\[
	\alpha \leq \beta + \sqrt{\frac{1}{2} \cdot \left(k^{-1}\ln n - \ln \beta - (\beta^{-1} - 1) \ln(1 - \beta) + \frac{i \cdot \ln(k + 1) - \ln |\cC|}{k + 1}\right)}
	\enspace,
\]
which implies that the number $i$ of oracle queries made by $ALG$ obeys
\[
	i
	\geq
	\frac{(k + 1) \cdot [2(\alpha - \beta)^2 + \ln \beta + (\beta^{-1} - 1) \ln (1 - \beta) - k^{-1}\ln n] + \ln |\cC|}{\ln (k + 1)}
	\enspace.
	\qedhere
\]
\end{proof}

Using Proposition~\ref{prop:set_identification}, we can now prove the promised version of Theorem~\ref{thm:cardinality_lower_bound} for deterministic algorithms.
\begin{corollary} \label{cor:cardinality_deterministic_lower_bound}
For any rational constant $\beta \in (0, 1)$ and (not necessarily constant) $\eps = \omega(\sqrt{\frac{\log n}{n}})$, every deterministic algorithm for {\SMC} that guarantees an approximation ratio of $\beta + \eps$ for instances obeying $k = \beta n$ must use $\Omega(\tfrac{\eps^2 n}{\log n})$ value oracle queries. Moreover, this is true even when the objective function $f$ of {\SMC} is guaranteed to be a linear function.
\end{corollary}
\begin{proof}
Let $\cC$ be the set of all subsets of $\cN$ of size $\beta n$, and let $k = \beta n$. By Stirling's approximation, the size of the set $\cC$ is at least
\[
	\binom{n}{\beta n}
	\geq
	\frac{e(n/e)^n}{e\beta n(\beta n / e)^{\beta n} \cdot e(1 - \beta) n((1 - \beta)n / e)^{(1 - \beta)n}}
	=
	\frac{1}{en^2 \beta^{\beta n + 1} \cdot (1 - \beta)^{(1 - \beta)n + 1}}
	\enspace.
\]
Therefore, the number of oracle queries made by any deterministic $(\beta + \eps)$-approximation algorithm for {\SI} with the above choice of parameters $\cC$ and $k$ is at least
\begin{align*}
	&
	\frac{(\beta n + 1) \cdot [2\eps^2 + \ln \beta + (\beta^{-1} - 1) \ln (1 - \beta) - (\beta n)^{-1}\ln n] + \ln |\cC|}{\ln (\beta n + 1)}\\
	\geq{} &
	\frac{(\beta n + 1) \cdot [2\eps^2 + \ln \beta + (\beta^{-1} - 1) \ln (1 - \beta) - (\beta n)^{-1}\ln n]}{\ln (\beta n + 1)} \\&- \frac{1 + 2 \ln n + (\beta n + 1)\ln \beta + ((1 - \beta)n + 1) \ln (1 - \beta)}{\ln (\beta n + 1)}\\
	={} &
	\frac{(\beta n + 1) \cdot [2\eps^2 - (\beta n)^{-1}\ln n] - 1 - 2 \ln n + (\beta^{-1} - 2) \ln (1 - \beta)}{\ln (\beta n + 1)}\\
	={} &
	\frac{\eps^2 \cdot \Theta(n) - \Theta(\log n)}{\Theta(\log n)}
	=
	\Omega\left(\frac{\eps^2 n}{\log n}\right)
	\enspace,
\end{align*}
where the last equality holds since we assume $\eps = \omega(\sqrt{\frac{\log n}{n}})$.

Assume now that we have an algorithm $ALG$ for {\SMC} with the guarantee stated in the corollary. We construct an algorithm $ALG'$ for {\SI} that consists of the following three steps.
\begin{compactenum}
	\item Construct an instance of {\SMC} by setting $k = \beta n$ and $f(S) = |S \cap C^*|$.
	\item Since $f$ is a linear function, we can use $ALG$ to find a set $S \subseteq \cN$ of size at most $\beta n$ such that $f(S) \geq (\beta + \eps) \cdot f(C^*) = (\beta + \eps)|C^*|$.
	\item Return any subset $S' \subseteq \cN$ of size $\beta n$ that includes $S$. Clearly, $|S' \cap C^*| \geq |S \cap C^*| = f(S) \geq (\beta + \eps)|C^*|$.
\end{compactenum}
The algorithm $ALG'$ obtains $(\beta + \eps)$-approximation for {\SI}, and therefore, by the above discussion, it must use $\Omega(\frac{\eps^2 n}{\log n})$ queries to the oracle of {\SI}. However, since $ALG'$ queries the oracle of {\SI} only once for every time that $ALG$ queries the value oracle of $f$, we get that $ALG$ must be using $\Omega(\frac{\eps^2 n}{\log n})$ value oracle queries.
\end{proof}

\subsection{Inapproximability of {\SMC} for Randomized Algorithms and Large \texorpdfstring{$k$}{k} Values} \label{sec:cardinality_lower_randomized}

In this section we extend the inapproximability result from Section~\ref{sec:cardinality_lower_randomized} to randomized algorithms. To do that, we start by proving the following version of Proposition~\ref{prop:set_identification} for randomized algorithms.
\begin{proposition} \label{prop:set_identification_randomized}
For every $\alpha \in (\beta, 1]$ and set $\cC$, any (possibly randomized) $\alpha$-approximation algorithm for {\SI} must make \[\frac{(k + 1) \cdot [(\alpha - \beta)^2/2 + \ln \beta + (\beta^{-1} - 1) \ln (1 - \beta) - k^{-1}\ln n] + \ln (|\cC|) + \ln (\frac{\alpha - \beta}{2 - \alpha - \beta})}{\ln (k + 1)}\] oracle queries in the worst case.\footnote{One can also prove in a similar way a version of this proposition bounding the expected number of oracle queries.}
\end{proposition}
\begin{proof}
Let $ALG$ be an algorithm of the kind described in the proposition, assume that the set $C^*$ is chosen uniformly at random out of $\cC$ and let $S_R$ be the output set of $ALG$. Using a Markov like argument, we get that since $ALG$ has an approximation guarantee of $\alpha$ and $|S_R \cap C^*|$ is always at most $|C^*| = k$,
\[
	\Pr\left[|S_R \cap C^*| \geq \frac{(\alpha + \beta)k}{2}\right]
	\geq
	\frac{\alpha - \beta}{2 - \alpha - \beta}
	\enspace,
\]
where the probability is over both the random choice of $C^*$ and the randomness of $ALG$ itself.

Recall now that $ALG$, as a randomized algorithm, can be viewed as a distribution over deterministic algorithms. This means that there must be at least one (deterministic) algorithm $ALG_D$ in the support of this distribution such that its output set $S_D$ obeys
\[
	\Pr\left[|S_D \cap C^*| \geq \frac{(\alpha + \beta)k}{2}\right]
	\geq
	\Pr\left[|S_R \cap C^*| \geq \frac{(\alpha + \beta)k}{2}\right]
	\geq
	\frac{\alpha - \beta}{2 - \alpha - \beta}
	\enspace.
\]
However, since $ALG_D$ is a deterministic algorithm, $S_D$ is a function of the set $C^*$ alone, which implies that there exists a subset $\cC' \subseteq \cC$ of size at least $\frac{\alpha - \beta}{2 - \alpha - \beta} \cdot |\cC|$ such that the algorithm $ALG_D$ guarantees $\frac{\alpha + \beta}{2}$-approximation whenever $C^* \in \cC'$.

By Proposition~\ref{prop:set_identification}, the number of oracle queries used by $ALG_D$ must be at least
\begin{align*}
	&
	\frac{(k + 1) \cdot [2(\frac{\alpha + \beta}{2} - \beta)^2 + \ln \beta + (\beta^{-1} - 1) \ln (1 - \beta) - k^{-1}\ln n] + \ln |\cC'|}{\ln (k + 1)}\\
	\geq{} &
	\frac{(k + 1) \cdot [(\alpha - \beta)^2/2 + \ln \beta + (\beta^{-1} - 1) \ln (1 - \beta) - k^{-1}\ln n] + \ln |\cC| + \ln (\frac{\alpha - \beta}{2 - \alpha - \beta})}{\ln (k + 1)}
	\enspace.
\end{align*}
The proposition now follows since the number of oracle queries made by $ALG_D$ is a lower bound on the number of oracle queries made by $ALG$ in the worst case.
\end{proof}

\begin{corollary} \label{cor:cardinality_randomized_lower_bound}
For any rational constant $\beta \in (0, 1)$ and (not necessarily constant) $\eps = \omega(\sqrt{\frac{\log n}{n}})$, every (possibly randomized) algorithm for {\SMC} that guarantees an approximation ratio of $\beta + \eps$ for instances obeying $k = \beta n$ must use $\Omega(\tfrac{\eps^2 n}{\log n})$ value oracle queries. Moreover, this is true even when the objective function $f$ of {\SMC} is guaranteed to be a linear function.
\end{corollary}
\begin{proof}
Let $\cC$ be the set of all subsets of $\cN$ of size $\beta n$, and let $k = \beta n$. As was proved in the proof of Corollary~\ref{cor:cardinality_deterministic_lower_bound}, 
\[
	\binom{n}{\beta n}
	\geq
	\frac{1}{en^2 \beta^{\beta n + 1} \cdot (1 - \beta)^{(1 - \beta)n + 1}}
	\enspace,
\]
and therefore, by Proposition~\ref{prop:set_identification_randomized}, the number of oracle queries made by any $(\beta + \eps)$-approximation algorithm for {\SI} with the above choice of parameters $\cC$ and $k$ is at least
\begin{align*}
	&
	\frac{(\beta n + 1) \cdot [\eps^2/2 + \ln \beta + (\beta^{-1} - 1) \ln (1 - \beta) - (\beta n)^{-1}\ln n] + \ln |\cC| + \ln (\frac{\eps}{2 - 2\beta - \eps})}{\ln (\beta n + 1)}\\
	\geq{} &
	\frac{(\beta n + 1) \cdot [\eps^2/2 + \ln \beta + (\beta^{-1} - 1) \ln (1 - \beta) - (\beta n)^{-1}\ln n] + \ln (\eps / 2)}{\ln (\beta n + 1)} \\&- \frac{1 + 2 \ln n + (\beta n + 1)\ln \beta + ((1 - \beta)n + 1) \ln (1 - \beta)}{\ln (\beta n + 1)}\\
	={} &
	\frac{(\beta n + 1) \cdot [\eps^2/2 - (\beta n)^{-1}\ln n] - 1 - 2 \ln n + (\beta^{-1} - 2) \ln (1 - \beta) + \ln(\eps / 2)}{\ln (\beta n + 1)}\\
	={} &
	\frac{\eps^2 \cdot \Theta(n) - \Theta(\log n) - \Theta(\log \eps^{-1})}{\Theta(\log n)}
	=
	\Omega\left(\frac{\eps^2 n}{\log n}\right)
	\enspace,
\end{align*}
where the last equality holds since we assume $\eps = \omega(\sqrt{\frac{\log n}{n}})$.

The rest of the proof is completely identical to the corresponding part in the proof of Corollary~\ref{cor:cardinality_deterministic_lower_bound} (up to the need to add some expectations signs), and therefore, we omit it.
\end{proof}

Theorem~\ref{thm:cardinality_lower_bound} is the special case of Corollary~\ref{cor:cardinality_randomized_lower_bound} in which $\eps$ is a positive constant, and therefore, obeys $\eps = \omega(\sqrt{\frac{\log n}{n}})$.

\subsection{Inapproximability for {\USM}} \label{ssc:unconstrained_lower_bound}

In this section we reuse the machinery developed in Section~\ref{sec:cardinality_lower_randomized} to get a query complexity lower bound for {\USM} and prove Theorem~\ref{thm:unconstrained_lower_bound}. Towards this goal, we assume throughout the section the following parameters for {\SI}. The parameter $k$ is set to be $n/2$, and $\cC$ is the collection of all subsets of $\cN$ of size $k$. One can verify that the proof of Corollary~\ref{cor:cardinality_randomized_lower_bound} also implies the following lemma (as a special case for $\beta = 1/2$).
\begin{lemma} \label{lem:SI_hardness_parameters}
Given the above parameters, for any $\eps = \omega(\sqrt{\frac{\log n}{n}})$, every (possibly randomized) $(\nicefrac{1}{2} + \eps)$-approximation algorithm for {\SI} must use $\Omega(\frac{\eps^2 n}{\log n})$ oracle queries.
\end{lemma}

Our next objective is to show a reduction from {\SI} to {\USM}. Let $ALG$ be a $(\nicefrac{1}{4} + \eps)$-approximation algorithm for {\USM}. Using $ALG$, one can design the algorithm for {\SI} (with the parameters we assume) that appears as Algorithm~\ref{alg:usm_si_reduction}. It is important to observe that the value oracle of the set function $f$ defined on Line~\ref{line:function} of the algorithm can be implemented using one query to the oracle of {\SI}. Furthermore, $f$ is the cut-function of a directed graph, and therefore, it is non-negative and submodular. This allows Algorithm~\ref{alg:usm_si_reduction} to use $ALG$ to construct a set $T$ of large expected value. The set $T$ is then replaced by it complement $\cN \setminus T$ if this increases the value of $f(T)$. Naturally, this replacement can only increase the expected value of the set $T$, and also gives it some deterministic properties that we need. Algorithm~\ref{alg:usm_si_reduction} completes by converting the set $T$ into an output set of size $n/2$ in one of two ways. If $T$ is too small, then a uniformly random subset of $\cN \setminus T$ of the right size is added to it. Otherwise, if $T$ is too large, then a uniformly random subset of it of the right size is picked.
\begin{algorithm2e}
\caption{Reduction from {\SI} to {\USM}} \label{alg:usm_si_reduction}
\DontPrintSemicolon
Define a function $f\colon 2^\cN \to \nnR$ by $f(S) = |S \cap C^*| \cdot (\tfrac{n}{2} - |S \setminus C^*|)$.\label{line:function}\\
Use $ALG$ to find a set $T$ that such that $\bE[f(T)] \geq (\nicefrac{1}{4} + \eps) \cdot \max_{S \subseteq \cN} f(S) = (\nicefrac{1}{4} + \eps) \cdot n^2/4$.\\
\lIf{$f(T) < f(\cN \setminus T)$}{Update $T \gets \cN \setminus T$. \label{line:manual_switch}}
\If{$|T| \leq n/2$}
{
	Pick a uniformly at random subset $R$ of $\cN \setminus T$ of size $n/2 - |T|$.\label{line:increase_T}\\
	\Return {$T \cup R$}.
}
\Else
{
	Pick a uniformly at random subset $R$ of $T$ of size $n/2$.\label{line:decrease_T}\\
	\Return{$R$}.
}
\end{algorithm2e}

Let $T$ and $\tilde{T}$ denote the values of the set $T$ before and after Line~\ref{line:manual_switch} of Algorithm~\ref{alg:usm_si_reduction}. The following observation states the properties of the set $\tilde{T}$ that we need below.
\begin{observation} \label{obs:basic_properties}
The set $\tilde{T}$ obeys
\begin{compactitem}
	\item $\bE[f(\tilde{T})] \geq (\nicefrac{1}{4} + \eps) \cdot n^2/4$.
	\item $|\tilde{T} \cap C^*| \geq |\tilde{T} \setminus C^*|$.
\end{compactitem}
\end{observation}
\begin{proof}
Line~\ref{line:manual_switch} of Algorithm~\ref{alg:usm_si_reduction} guarantees that $f(\tilde{T}) \geq f(T)$, which implies
\[
	\bE[f(\tilde{T})]
	\geq
	\bE[f(T)]
	\geq
	(\nicefrac{1}{4} + \eps) \cdot n^2/4
	\enspace.
\]

To see that the second part of the observation holds as well, note that Line~\ref{line:manual_switch} also guarantees $f(\tilde{T}) \geq f(\cN \setminus \tilde{T})$, which, by the definition of $f$, implies
\[
	|\tilde{T} \cap C^*|(\tfrac{n}{2} - |\tilde{T} \setminus C^*|)
	\geq
	|\tilde{T} \setminus C^*|(\tfrac{n}{2} - |\tilde{T} \cap C^*|)
	\enspace,
\]
Rearranging this inequality now gives
\[
	\tfrac{n}{2}|\tilde{T} \cap C^*|
	\geq
	\tfrac{n}{2}|\tilde{T} \setminus C^*|
	\enspace,
\]
and dividing this inequality by $n/2$ yields the required inequality.
\end{proof}

One consequence of the last lemma is given by the next corollary.
\begin{corollary} \label{cor:difference_expectation}
$\bE[|\tilde{T} \cap C^*| - |\tilde{T} \setminus  C^*|] \geq n\eps/2$.
\end{corollary}
\begin{proof}
Observe that
\begin{align*}
	&
	\tfrac{n}{2} \cdot \bE[|\tilde{T} \cap C^*| - |\tilde{T} \setminus  C^*|]
	\geq
	\bE[(|\tilde{T} \cap C^*| - |\tilde{T} \setminus  C^*|)(\tfrac{n}{2} - |\tilde{T} \setminus  C^*|)]\\\
	\geq{} &
	\bE[(|\tilde{T} \cap C^*| - |\tilde{T} \setminus  C^*|)(\tfrac{n}{2} - |\tilde{T} \setminus  C^*|)] + \bE[|\tilde{T} \setminus  C^*|(\tfrac{n}{2} - |\tilde{T} \setminus  C^*|)] - \frac{n^2}{16}\\
	={} &
	\bE[|\tilde{T} \cap C^*|(\tfrac{n}{2} - |\tilde{T} \setminus  C^*|)] - \frac{n^2}{16}
	=
	\bE[f(\tilde{T})] - \frac{n^2}{16}
	\geq
	\frac{n^2\eps}{4}
	\enspace,
\end{align*}
where the first inequality holds since Observation~\ref{obs:basic_properties} guarantees that $|\tilde{T} \cap C^*| - |\tilde{T} \setminus  C^*|$ is always non-negative, and the last inequality follows also from Observation~\ref{obs:basic_properties}. The corollary now follows by dividing this inequality by $n/2$.
\end{proof}


We are now ready to prove the following lemma, which analyzes the performance guarantee of Algorithm~\ref{alg:usm_si_reduction}.
\begin{lemma}
If $T'$ is the output set of Algorithm~\ref{alg:usm_si_reduction}, then $\bE[|T' \cap C^*|] \geq (\nicefrac{1}{2} + \nicefrac{\eps}{4}) \cdot n/2$, where the expectation is over the randomness of $ALG$ and Algorithm~\ref{alg:usm_si_reduction}.
\end{lemma}
\begin{proof}
Let us begin this proof by fixing the set $\tilde{T}$. In other words, until we unfix this set, all the expectations we use are assumed to be only over the randomness of Lines~\ref{line:increase_T} and~\ref{line:decrease_T} of Algorithm~\ref{alg:usm_si_reduction}. Whenever Line~\ref{line:decrease_T} of the algorithm is used, we have
\begin{align*}
	\bE[|T' \cap C^*|]
	={} &
	\frac{n/2}{|\tilde{T}|} \cdot |\tilde{T} \cap C^*|
	=
	\frac{n}{4} \cdot \left(1 + \frac{2|\tilde{T} \cap C^*| - |\tilde{T}|}{|\tilde{T}|}\right)\\
	={} &
	\frac{n}{4} \cdot \left(1 + \frac{|\tilde{T} \cap C^*| - |\tilde{T} \setminus C^*|}{|\tilde{T}|}\right)
	\geq
	\frac{n}{4} \cdot \left(1 + \frac{|\tilde{T} \cap C^*| - |\tilde{T} \setminus C^*|}{n}\right)
	\enspace,
\end{align*}
where the inequality holds since Observation~\ref{obs:basic_properties} guarantees that $|\tilde{T} \cap C^*| - |\tilde{T} \setminus  C^*| \geq 0$. Similarly, whenever Line~\ref{line:increase_T} of Algorithm~\ref{alg:usm_si_reduction} is used, we get
\begin{align*}
	\bE[|T' \cap C^*|]
	={} &
	|\tilde{T} \cap C^*| + \bE[|R \cap C^*|]
	=
	|\tilde{T} \cap C^*| + \frac{\tfrac{n}{2} - |\tilde{T}|}{|\cN \setminus \tilde{T}|} \cdot |C^* \setminus \tilde{T}|\\
	={} &
	|\tilde{T} \cap C^*| + \frac{\tfrac{n}{2} - |\tilde{T}|}{n - |\tilde{T}|} \cdot (\tfrac{n}{2} - |\tilde{T} \cap C^*|)
	=
	\frac{\tfrac{n}{2} |\tilde{T} \cap C^*| + \frac{n}{2}(\tfrac{n}{2} - |\tilde{T}|)}{n - |\tilde{T}|}\\
	={} &
	\frac{n}{4} \cdot \left(1 + \frac{|\tilde{T} \cap C^*| - |\tilde{T} \setminus C^*|}{n - |\tilde{T}|}\right)
	\geq
	\frac{n}{4} \cdot \left(1 + \frac{|\tilde{T} \cap C^*| - |\tilde{T} \setminus C^*|}{n}\right)
	\enspace.
\end{align*}
Therefore, even if we unfix $\tilde{T}$, and take expectation also over the randomness of this set, we still get by the law of total expectation that
\[
	\bE[|T' \cap C^*|]
	\geq
	\frac{n}{4} \cdot \left(1 + \frac{\bE[|\tilde{T} \cap C^*| - |\tilde{T} \setminus C^*|]}{n}\right)
	\geq
	\frac{n}{4} \cdot \left(1 + \frac{n\eps/2}{n}\right)
	=
	\frac{n}{2} \cdot \left(\frac{1}{2} + \frac{\eps}{4}\right)
	\enspace,
\]
where the second inequality follows from Corollary~\ref{cor:difference_expectation}.
\end{proof}

We can now prove Theorem~\ref{thm:unconstrained_lower_bound}, which we repeat here for convenience.
\ThmUnconstrainedLowerBound*
\begin{proof}
By Lemma~\ref{lem:SI_hardness_parameters}, Algorithm~\ref{alg:usm_si_reduction} must use $\Omega(\frac{\eps^2 n}{\log n})$ queries to the oracle of {\SI}. However, aside from the queries used by $ALG$, Algorithm~\ref{alg:usm_si_reduction} queries this oracle only a constant number of times, which implies that $ALG$ must be using $\Omega(\frac{\eps^2 n}{\log n})$ such queries. Recall now that $ALG$, as an algorithm for {\USM}, queries the oracle of {\SI} only by querying the value oracle of $f$, and every query to this value oracle results in a single query to the oracle of {\SI}. Therefore, $ALG$ must be using $\Omega(\frac{\eps^2 n}{\log n})$ value oracle queries. This completes the proof of the theorem since $ALG$ was chosen as an arbitrary $(\nicefrac{1}{4} + \eps)$-approximation algorithm for {\USM}.
\end{proof}
\section{Set System and Multiple Knapsacks Constraints} \label{sec:SMKS}

In this section we consider the {\SMKSFull} problem (\SMKS). In this problem we are given a non-negative monotone submodular function $f\colon 2^\cN \to \nnR$, a $p$-set system $\cM = (\cN, \cI)$, $d \geq 1$ non-negative cost functions $c_i\colon \cN \to \nnR$ (one function for every integer $1 \leq i \leq d$) and $d$ positive values $B_1, B_2, \dotsc, B_d$. We say that a set $S \subseteq \cN$ is \emph{feasible} if it is independent in $\cM$ and also obeys $c_i(S) \leq B_i$ for every $1 \leq i \leq d$ (where $c_i(S) \triangleq \sum_{u \in S} c_i(u)$). The objective of the problem is to find a set maximizing $f$ among all feasible sets $S \subseteq \cN$. Below, we denote by $r$ the rank of the set system $\cM$. We also make a few simplifying assumptions.
\begin{itemize}
	\item We assume that $\cM$ does not include any self-loops or elements $u \in \cN$ such that $c_i(u) > B_i$ for some integer $1 \leq i \leq d$. Any element violating these assumptions can be simply discarded since it cannot belong to any feasible set.
	\item We assume that the sum $\sum_{i = 1}^d c_i(u)$ is strictly positive for every element $u \in \cN$. Elements violating this assumption can be added to any solution, and therefore, it suffices to solve the problem without such elements, and then add them to the obtained solution at the every end of the algorithm.
	\item We assume that $B_i = 1$ for every integer $1 \leq i \leq d$. This can be guaranteed by scaling the cost functions $c_i$.
\end{itemize}

Recall that Section~\ref{ssc:related_work} demonstrated a tradeoff between the time complexity and approximation guarantee of state-of-the-art algorithms for {\SMKS}. The next theorem further studies this tradeoff, and improves over the state-of-the-art approximation for nearly-linear time algorithms. The $\tilde{O}$ notation in this theorem suppresses factors that are poly-logarithmic in $n$, $d$ and $\eps^{-1}$.
\begin{theorem} \label{thm:SMKS}
For every $\eps > 0$, there exist algorithms that achieve $[(1 + O(\eps))(p + 1 + \tfrac{7}{4}d)]^{-1}$-approximation and $(p + 1.5556 + \tfrac{13}{9}d + \eps)^{-1}$ for {\SMKSFull} (\SMKS) in $\tilde{O}(nd + n/\eps)$ and $\tilde{O}(n^2 + nd)$ time, respectively.
\end{theorem}

In Section~\ref{ssc:general_algorithm} we present and analyze a basic version of the algorithm that we use to prove Theorem~\ref{thm:SMKS} (in Appendix~\ref{app:nearly_linear_main} we explain how to make this basic algorithm nearly-linear). The basic version of our algorithm assumes access to an estimate $\rho$ of the density of the small elements of an optimal solution for the problem. In Section~\ref{ssc:rho_search} we explain how the dependence of the algorithm on $\rho$ can be dropped without increasing the time complexity of the algorithm by too much. Finally, in Section~\ref{ssc:big_elements_algorithms} we show how our algorithm can be used to derive the results stated in Theorem~\ref{thm:SMKS}.

\subsection{Basic Algorithm for {\SMKS}} \label{ssc:general_algorithm}

In this section we present and analyze a basic version algorithm of the algorithm that we use to prove Theorem~\ref{thm:SMKS}. This algorithm is given as Algorithm~\ref{alg:main}, and it gets two parameters. The first of these parameters is an integer $\lambda \geq 1$. Elements that have a value larger than $\lambda^{-1}$ with respect to at least one function $c_i$ are considered big elements, and are stored in the set $B$ of the algorithm. The other elements of $\cN$ are considered small elements. The algorithm never considers any solution that includes both big and small elements. Instead, it creates one candidate solution $S_B$ from the big elements, and one candidate solution from the small elements, and then outputs the better among the two (technically, in some cases the algorithm outputs directly the candidate solution based on the small elements without comparing it to $S_B$). The candidate solution $S_B$ is constructed using a procedure called $\BigAlg$ that gets the set $B$ as input and outputs a feasible set whose value is at least $\alpha \cdot f(OPT \cap B)$, where $\alpha$ is a some value in $(0, 1]$ and $OPT$ is an arbitrary optimal solution. At this point we ignore the implementation of $\BigAlg$, and leave the value of $\alpha$ unspecified. These gaps are filled in Section~\ref{ssc:big_elements_algorithms}.

Most of Algorithm~\ref{alg:main} is devoted to constructing the candidate solution out of small elements, which we refer to below as the ``small elements solution''. In the construction of this solution, Algorithm~\ref{alg:main} uses its second parameter, which is a value $\rho \geq 0$ that intuitively should represent the density of the small elements of $OPT$. The algorithm initializes the small elements solution to be empty, and then iteratively adds to it the element with the largest marginal contribution among the small elements that have two properties: (i) their addition to the solution does not make it dependent in $\cM$, and (ii) their density (the ratio between their marginal contribution and cost according to the linear constraints) is at least $\rho$. This process of growing the small elements solution can end in one of two ways. One option is that the process ends because no additional elements can be added to the solution (in other words, no element has the two properties stated above). In this case the better among $S_B$ and the small elements solution obtained $S_k$ is returned. The other way in which the process of growing the small elements solution can end is when it starts violating at least one linear constraint. When this happens, Algorithm~\ref{alg:main} uses a procedure called {\SetExtract} to get a subset of the small elements solution that is feasible and also has a good value, and this subset is returned.

\begin{algorithm2e}
\DontPrintSemicolon
\caption{\texttt{Basic Algorithm}$(\lambda, \rho)$} \label{alg:main}
\tcp{Build the set of big elements, and find a candidate solution based on them.}
Let $B \gets \{u \in \cN \mid \exists_{1 \leq i \leq d}\; c_i(u) > \lambda^{-1}\}$.\\
Let $S_B$ be the output set of $\BigAlg(B)$.\\

\BlankLine

\tcp{Construct a solution from the small elements.}
Let $S_0 \gets \varnothing$, $k \gets 0$.\\
\While{there exists an element $u \in \cN \setminus (S_k \cup B)$ such that $S_k + u \in \cI$ and $f(u \mid S_k) \geq \rho \cdot \sum_{i = 1}^d c_i(u)$ \label{line:candidates_condition}}
{
	Let $v_{k + 1}$ be an element maximizing $f(u \mid S_k)$ among all the elements obeying the condition of the loop.\\
	Let $S_{k + 1} \gets S_k + v_{k + 1}$.\\
	\lIf{$\max_{1 \leq i \leq d} c_i(S_{k + 1}) \leq 1$ \label{line:budget_condition}}
	{
		Increase $k$ by $1$.
	}
	\lElse
	{
		\Return{the output set of {\SetExtract}$(\lambda, S_{k + 1})$}.\label{line:exceed_budget}
	}
}
\Return{the better set among $S_B$ and $S_k$}.
\end{algorithm2e}

Let us now describe the procedure {\SetExtract}, which appears as Algorithm~\ref{alg:set_extract}. As explained above, this procedure gets a set $S$ of small elements that violates at least one of the linear constraints. Its objective is to output a subset $T$ of $S$ that does not violate any linear constraint, but is not very small in terms of the linear constraints. Since the set $S$ passes by Algorithm~\ref{alg:main} to {\SetExtract} contains only elements of density at least $\rho$, this implies that the output set $T$ of {\SetExtract} has a significant value. Algorithm~\ref{alg:set_extract} does its job by constructing $\lambda + 1$ subsets $T_1, T_2, \dotsc, T_{\lambda + 1}$ of $S$, and then outputting the subset with the maximum size with respect to the linear constraint. The first subset $T_1$ is constructed by starting with the empty set, and then simply adding elements of $S$ to $T_1$ one by one, in an arbitrary order, until some element $u_1$ cannot be added because adding it will result in a set that violates some linear constraint. The algorithm then constructs the second subset $T_2$ in essentially the same way, but makes sure to include $u_1$ in it by starting with the set $\{u_1\}$ and then adding elements of $S$ to $T_2$ one by one in an arbitrary order, until some element $u_2$ cannot be added because adding it will result in a set that violates some linear constraint. The set $T_3$ is then constructed in the same way starting from the set $\{u_1, u_2\}$, and in general the set $T_j$ is constructed by starting from the set $\{u_1, u_2, \dotsc, u_{j - 1}\}$ and then adding to it elements of $S$ in an arbitrary order until some element $u_j$ cannot be added because adding it will result in a set that violates some linear constraint. Intuitively, this method of constructing the subsets $T_1, T_2, \dotsc, T_{\lambda + 1}$ guarantees that every element $u_j$ is rejected at most once from a subset due to the linear constraints.

\begin{algorithm2e}
\caption{\SetExtract$(\lambda, S)$} \label{alg:set_extract}
\DontPrintSemicolon
\For{$j = 1$ \KwTo $\lambda + 1$ \label{line:output_set_loop}}
{
	Let $T_j \gets \{u_1, u_2, \dotsc, u_{j - 1}\}$.\\
	\For{every element $u \in S$ \label{line:scan_loop}}
	{
		\lIf{$\max_{1 \leq i \leq d} c_i(T_j + u) \leq 1$ \label{line:knapsack_check_main}}{Add $u$ to $T_j$.}
		\lElse{Denote the element $u$ by $u_j$ from this point on, and exit the loop of Line~\ref{line:scan_loop}. \label{line:exit_loop}}
	}
}
\Return{the set maximizing $\sum_{i = 1}^d c_i(T)$ among all sets $T \in \{T_1, T_2, \dotsc, T_{\lambda + 1}\}$}.
\end{algorithm2e}

A formal statement of the guarantee of {\SetExtract} is given by the next lemma.
\begin{lemma} \label{lem:set_extract}
Assuming the input set $S$ of {\SetExtract} obeys
\begin{compactitem}
	\item $c_i(S) > 1$ for some integer $1 \leq i \leq d$ and
	\item $\max_{1 \leq i \leq d} c_i(u) \leq \lambda^{-1}$ for every element $u \in S$,
\end{compactitem}
then the output set $T$ is a subset of $S$ such that $\max_{1 \leq i \leq d} c(T) \leq 1$, but $\sum_{i = 1}^d c_i(T) \geq \frac{\lambda}{\lambda + 1}$.
\end{lemma}
\begin{proof}
For any $1 \leq j \leq \lambda + 1$, {\SetExtract} initializes the set $T_j$ to contain $j - 1 \leq \lambda$ elements, which implies that immediately after this initialization the set $T_j$ obeyed $\max_{1 \leq i \leq d} c_i(T_j) \leq 1$ because of the second condition of the lemma. After the initialization of $T_j$, Algorithm~\ref{alg:set_extract} grows it by adding to it only elements whose addition does not make $\max_{1 \leq i \leq d} c_i(T_j)$ exceed $1$. This method of growth guarantees that the set $T_j$ keeps obeying $\max_{1 \leq i \leq d} c_i(T_j) \leq 1$ throughout the execution of the algorithm. Therefore, since $T$ is chosen as the set $T_j$ for some integer $j$, it obeys $\max_{1 \leq i \leq d} c_i(T) \leq 1$.

Consider any iteration of the loop on Line~\ref{line:output_set_loop} of {\SetExtract}. If this loop never reaches Line~\ref{line:exit_loop}, then we are guaranteed that all the elements of $S$ are added to $T_j$, which contradicts the inequality $\max_{1 \leq i \leq d} c_i(T) \leq 1$ that we have proved above because we are guaranteed that $c_i(S) > 1$ for some integer $1 \leq i \leq d$. Therefore, {\SetExtract} reaches Line~\ref{line:exit_loop} in every iteration of the outter loop. Specifically, for any integer $1 \leq j \leq \lambda$, since the algorithm reached Line~\ref{line:exit_loop} in iteration number $j$ of this loop, we must have $\max_{1 \leq i \leq d} c_i(T_j + u_j) > 1$. Hence,
\[
	\sum_{i = 1}^d c_i(T_j)
	=
	\sum_{i = 1}^d [c_i(T_j + u_j) - c_i(u_j)]
	\geq
	\max_{1 \leq i \leq d} c_i(T_j + u_j) - \sum_{i = 1}^d c_i(u_j)
	>
	1 - \sum_{i = 1}^d c_i(u_j)
	\enspace.
\]
Using this inequity and the observation that $T_{\lambda + 1}$ includes all the elements $u_1, u_2, \dotsc, u_\lambda$, we get
\begin{align*}
	\sum_{i = 1}^d c_i(T)
	={} &
	\max_{1 \leq j \leq \lambda + 1} \sum_{i = 1}^d c_i(T_j)
	\geq
	\frac{\sum_{j = 1}^{\lambda + 1} \sum_{i = 1}^d c_i(T_j)}{\lambda + 1}\\
	\geq{} &
	\frac{\sum_{j = 1}^{\lambda} [1 - \sum_{i = 1}^d c_i(u_j)] + \sum_{j = 1}^\lambda \sum_{i = 1}^d c_i(u_j)}{\lambda + 1}
	=
	\frac{\lambda}{\lambda + 1}
	\enspace.
	\qedhere
\end{align*}
\end{proof}

We now get to the analysis of the full Algorithm~\ref{alg:main}. Let $\ell$ be the final value of the variable $k$ of Algorithm~\ref{alg:main}.
\begin{observation} \label{obs:feasible}
Algorithm~\ref{alg:main} outputs a feasible set.
\end{observation}
\begin{proof}
One can observe that the set $S_\ell$ is feasible because the conditions of Lines~\ref{line:candidates_condition} and~\ref{line:budget_condition} of Algorithm~\ref{alg:main} guarantee that it assigns to $S_k$ only feasible sets for any $0 \leq k \leq \ell$. Furthermore, $S_B$ is feasible by the properties we assume for the procedure {\BigAlg}. Therefore, to prove the observation it only remains to show that the output set $T$ of {\SetExtract}$(\lambda, S_{\ell + 1})$ is feasible when Line~\ref{line:exceed_budget} is executed.

To see that this is indeed the case, we note that $T$ is a subset of $S_{\ell + 1}$, which is an independent set of $\cM$, and therefore, $T$ is also independent in $\cM$. Additionally, Lemma~\ref{lem:set_extract} implies that $\max_{1 \leq i \leq d} c_i(T) \leq 1$ because (i) we are guaranteed by the condition of Line~\ref{line:budget_condition} that there exists an integer $1 \leq d \leq k$ such that $c_i(S_{\ell + 1}) = c_i(S_\ell + v_{\ell + 1}) > 1$ when Line~\ref{line:exceed_budget} is executed, and (ii) we are guaranteed that $\max_{u \in S_{\ell + 1}} \max_{1 \leq i \leq d} c_i(u) \leq \lambda^{-1}$ since $S_{\ell + 1}$ contains only small elements.
\end{proof}

Next, we analyze the time complexity of Algorithm~\ref{alg:main}. 
\begin{lemma}
Algorithm~\ref{alg:main} has a time complexity of $O(\lambda n d + nr + T_B)$, where $T_B$ is the time complexity of $\BigAlg$.
\end{lemma}
\begin{proof}
The construction of the set $B$ requires $O(dn)$ time, and therefore, the time complexity required for the entire Algorithm~\ref{alg:main} except for the loop starting on Line~\ref{line:candidates_condition} is $O(dn + T_B)$. In the rest of this proof we show that this loop requires $O(\lambda n d + nr)$ time, which implies the lemma.

The loop starting on Line~\ref{line:candidates_condition} of Algorithm~\ref{alg:main} runs at most $r$ times because the size of the set $S_k$ grows by at least one after every such iteration (and this set always remains feasible in $\cM$). To understand the time complexity of the iterations of the last loop, we can observe that, assuming we maintain the values $c_i(S_k)$, each such iteration takes $O(n + d)$ time, with the exception of the following operations.
\begin{itemize}
	\item In Line~\ref{line:candidates_condition} we need to calculate the sum $\sum_{i = 1}^d c_i(u)$ for multiple elements $u$, which takes $O(d)$ time per element. However, we can pre-calculate this sum for all the elements of $\cN$ in $O(nd)$ time.
	\item Executing $\SetExtract(\lambda, S_{k + 1})$ requires $O(\lambda |S_{k + 1}| d) = O(\lambda r d)$ time. However, this procedure is executed at most once by Algorithm~\ref{alg:main}.
\end{itemize}
Combining all the above, we get that the loop starting on Line~\ref{line:tau_loop} of Algorithm~\ref{alg:main} requires only $O(nd + \lambda r d + r(n + d)) = O(\lambda n d + nr)$ time.
\end{proof}

Our next objective is to analyze the approximation ratio of Algorithm~\ref{alg:main}. Let $E$ be the event that Algorithm~\ref{alg:main} returns through Line~\ref{line:exceed_budget}. We begin the analysis of Algorithm~\ref{alg:main} by looking separately at the case in which the event $E$ happens and at the case in which it does not happen. When the event $E$ happens, the output set of Algorithm~\ref{alg:main} is the output set of {\SetExtract}. This set contains only high density elements and is large in terms of the linear constraints (by Lemma~\ref{lem:set_extract}), which provides a lower bound on its value. The following lemma formalizes this argument.
\begin{lemma} \label{lem:E_guarantee}
If the event $E$ happens, then Algorithm~\ref{alg:main} returns a solution of value at least $\frac{\lambda \rho}{\lambda + 1}$.
\end{lemma}
\begin{proof}
Let $T$ be the output set of {\SetExtract}$(\lambda, S_{\ell + 1})$. We need to show that $f(T) \geq \frac{\lambda \rho}{\lambda + 1}$. Since $T$ is a subset of $S_{\ell + 1} = \{v_1, v_2, \dotsc, v_{\ell + 1}\}$,
\begin{align*}
	f(T)
	={} &
	\sum_{v_{k} \in T} f(v_{k} \mid T \cap S_{k - 1})
	\geq
	\sum_{v_{k} \in T} f(v_{k} \mid S_{k - 1})\\
	\geq{} &
	\sum_{v_{k} \in T} \left(\rho \cdot \sum_{i = 1}^d c_i(v_{k})\right)
	=
	\rho \cdot \sum_{i = 1}^d c_i(T)
	\geq
	\frac{\lambda \rho}{\lambda + 1}
	\enspace,
\end{align*}
where the first inequality follows from the submodularity of $f$, the second inequality follows from the definition of $v_{k}$, and the last inequality follows from the guarantee of Lemma~\ref{lem:set_extract}.
\end{proof}

Handling the case in which the event $E$ does not happen is somewhat more involved. Towards this goal, let us recursively define a set $O_k$ for every $0 \leq k \leq \ell$. The base of the recursion is that for $k = \ell$ we define $O_\ell = OPT \setminus (B \cup \{u \in OPT \mid f(u \mid S_\ell) < \rho \cdot \sum_{i = 1}^d c_i(u)\})$. Assuming $O_{k + 1}$ is already defined for some $0 \leq k < \ell$, we define $O_k$ as follows. Let $D_k = \{u \in O_{k + 1} \setminus S_k \mid S_k + u \in \cI\}$. If $|D_k| \leq p$, we define $O_k = O_{k + 1} \setminus D_k$. Otherwise, we let $D'_k$ be an arbitrary subset of $D_k$ of size $p$, and we define $O_k = O_{k + 1} \setminus D'_k$.
\begin{lemma} \label{lem:O_null_empty}
Assuming $E$ does not happen, $O_0 = \varnothing$.
\end{lemma}
\begin{proof}
We prove by a downward induction the stronger claim that $|O_k| \leq pk$ for every $0 \leq k \leq \ell$. To prove this inequality for $k = \ell$, we note that the fact that $E$ did not happen and still Algorithm~\ref{alg:main} terminated after $\ell$ iterations implies that no element of $O_\ell \setminus S_\ell$ can be added to $S_\ell$ without violating independence in $\cM$. Therefore, $S_\ell$ is a base of $S_\ell \cup O_\ell$ in this $p$-system, which implies $|O_\ell| \leq p|S_\ell| = p\ell$ by the definition of a $p$-system since $O_\ell \subseteq OPT$ is independent in $\cM$.

Assume now that the inequality $|O_{k + 1}| \leq p(k + 1)$ holds for some $0 \leq k < \ell$, and let us prove $|O_k| \leq pk$. There are two cases to consider. If $O_k = O_{k + 1} \setminus D'_k$, then $|O_k| = |O_{k + 1}| - |D'_k| \leq p(k + 1) - p = pk$. Otherwise, if $O_k = O_{k + 1} \setminus D_k$, then, by the definition of $D_k$, there are no elements of $O_k \setminus S_k$ that can be added to $S_k$ without violating independence in $\cM$. Therefore, $S_k$ is a base of $S_k \cup O_k$ in this $p$-system, which implies $|O_k| \leq p|S_k| = pk$ by the definition of a $p$-system since $O_k \subseteq O_\ell \subseteq OPT$ is independent in $\cM$.
\end{proof}
\begin{corollary} \label{cor:p-system_analysis}
Assuming $E$ does not happen, $f(S_\ell) \geq \frac{f(O_\ell \cup S_\ell)}{p + 1}$.
\end{corollary}
\begin{proof}
We prove by induction the stronger claim that, for every integer $0 \leq k \leq \ell$,
\begin{equation} \label{eq:guarantee_k}
	f(S_k)
	\geq
	\frac{f(O_k \cup S_k)}{p + 1}
	\enspace.
\end{equation}
For $k = 0$ this inequality follows from the non-negativity of $f$ since $S_0 = O_0 = \varnothing$ by Lemma~\ref{lem:O_null_empty}. Assume now that Inequality~\eqref{eq:guarantee_k} holds for some value $k - 1$ obeying $0 \leq k - 1 < \ell$, and let us prove it for $k$.

Consider the set $\Delta_k = O_k \setminus O_{k - 1}$. The construction of $O_{k - 1}$ guarantees that every element of $\Delta_k$ can be added to $S_{k - 1}$ without violating independence in $\cM$. Furthermore, since $\Delta_k \subseteq O_\ell$, every element $u \in \Delta_k$ also obeys $f(u \mid S_k) \geq f(u \mid S_\ell) \geq \rho \cdot \sum_{i = 1}^d c_i(u)$. Therefore, every element of $\Delta_k$ obeys the condition of the loop on Line~\ref{line:candidates_condition} of Algorithm~\ref{alg:main} in the $k$-th iteration of the loop, and by the definition of $v_k$, this implies
\begin{align*}
	f(S_k)
	={} &
	f(S_{k - 1}) + f(v_k \mid S_{k - 1})
	\geq
	f(S_{k - 1}) + \frac{f(v_k \mid S_{k - 1}) + \sum_{u \in \Delta_k} f(u \mid S_{k - 1})}{|\Delta_k| + 1}\\
	\geq{} &
	\frac{f(O_{k - 1} \cup S_{k - 1})}{p + 1} + \frac{f(\Delta_k + v_k \mid S_{k - 1})}{|\Delta_k| + 1}\\
	\geq{} &
	\frac{f(O_{k - 1} \cup S_{k - 1})}{p + 1} + \frac{f(\Delta_k + v_k \mid S_{k - 1})}{p + 1}
	\geq
	\frac{f(O_{k} \cup S_{k})}{p + 1}
	\enspace,
\end{align*}
where the second inequality follows from the induction hypothesis and the submodularity of $f$, the penultimate inequality follows from the monotonicity of $f$ and the observation that the construction of $O_{k - 1}$ guarantees $|\Delta_k| \leq p$, and the last inequality follows again from the submodularity of $f$.
\end{proof}

To use the last corollary, we need a lower bound on $O_\ell \cup S_\ell$, which is given by the next lemma.
\begin{lemma} \label{lem:Oell_lower_bound}
$f(O_\ell \cup S_\ell) \geq f(OPT) - f(S_B) / \alpha - \rho \cdot \left[d - \frac{|OPT \cap B|}{\lambda} \right]$.
\end{lemma}
\begin{proof}
Observe that
\begin{align} \label{eq:Oell_bound}
	f(O_\ell \cup S_\ell)
	={} &
	f(OPT \setminus (B \cup \{u \in OPT \mid f(u \mid S_\ell) < \rho \cdot \sum\nolimits_{i = 1}^d c_i(u)\}) \cup S_\ell)\\ \nonumber
	\geq{} &
	f(OPT) - f(OPT \cap B) - f(\{u \in OPT \setminus B \mid f(u \mid S_\ell) < \rho \cdot \sum\nolimits_{i = 1}^d c_i(u)\} \mid S_\ell)\\ \nonumber
	\geq{} &
	f(OPT) - f(S_B) / \alpha - f(\{u \in OPT \setminus B \mid f(u \mid S_\ell) < \rho \cdot \sum\nolimits_{i = 1}^d c_i(u)\} \mid S_\ell)
	\enspace,
\end{align}
where the first inequality follows from the submodularity and monotonicity of $f$, and the second inequality follows from the definition of $S_B$. To lower bound the rightmost side of the last inequality, we need to upper bound the last term in it.
\begin{align*}
	f(\{u \in{} & OPT \setminus B \mid f(u \mid S_\ell) < \rho \cdot \sum\nolimits_{i = 1}^d c_i(u)\} \mid S_\ell)
	\leq
	\sum_{\substack{u \in OPT \setminus B \\ f(u \mid S_\ell) < \rho \cdot \sum\nolimits_{i = 1}^d c_i(u)}} \mspace{-36mu} f(u \mid S_\ell)\\
	\leq{} &
	\rho \cdot \sum_{u \in OPT \setminus B} \sum_{i = 1}^d c_i(u)
	=
	\rho \cdot \left[\sum_{u \in OPT} \sum_{i = 1}^d c_i(u) - \sum_{u \in OPT \cap B} \sum_{i = 1}^d c_i(u)\right]
	\leq
	\rho \cdot \left[d - \frac{|OPT \cap B|}{\lambda} \right]
	\enspace,
\end{align*}
where the last inequality holds since $OPT$ is a feasible set and every element of $B$ is big. Plugging the last inequality into Inequality~\eqref{eq:Oell_bound} completes the proof of the lemma.
\end{proof}

We are now ready to lower bound the value of the solution produced by Algorithm~\ref{alg:main}. While reading the following proposition, it useful to keep in mind that $|OPT \cap B| \leq d(\lambda - 1)$ because (i) the feasibility of $OPT$ implies $c_i(OPT \cap B) \leq 1$ for every integer $1 \leq i \leq d$ and (ii) every element $u \in OPT \cap B$ obeys $c_i(u) > \lambda^{-1}$ for at least one such $i$.
\begin{proposition} \label{prop:main_guarantee}
Let
\[
	\rho^* = \frac{f(OPT)}{(p + 1 + \alpha^{-1}) / (1 + \lambda^{-1}) + d - |OPT \cap B| / \lambda}
	\enspace.
\]
If $\rho \geq \rho^*$ and the event $E$ happened, or $\rho \leq \rho^*$ and the event $E$ did not happen, then Algorithm~\ref{alg:main} outputs a solution of value at least $\lambda \rho^* / (\lambda + 1)$. Furthermore, if $\rho \in [(1 - \delta)\rho^*, \rho^*]$ for some value $\delta \in (0, 1]$, then, regardless of the realization of the event $E$, Algorithm~\ref{alg:main} outputs a solution of value at least $(1 - \delta)\lambda \rho^* / (\lambda + 1)$. Finally, Algorithm~\ref{alg:main} runs in $O(\lambda n d + nr + T_B)$ time, where $T_B$ is the time complexity of $\BigAlg$.
\end{proposition}
\begin{proof}
If the event $E$ happened, then Lemma~\ref{lem:E_guarantee} guarantees that the output of Algorithm~\ref{alg:main} is of value at least $\lambda \rho / (\lambda + 1)$. Therefore, to complete the proof of the proposition, it suffices to show that when the event $E$ does not happen and $\rho^* \geq \rho$, the value of the output of the Algorithm~\ref{alg:main} is at least $\lambda \rho^* / (\lambda + 1)$; and the rest of this proof is devoted to showing that this is indeed the case.

In the last case, Corollary~\ref{cor:p-system_analysis} and Lemma~\ref{lem:Oell_lower_bound} imply together
\[
	f(S_\ell)
	\geq
	\frac{f(OPT) - f(S_B) / \alpha - \rho \cdot \left[d - \frac{|OPT \cap B|}{\lambda}\right]}{p + 1}
	\enspace,
\]
and therefore, the output set of Algorithm~\ref{alg:main} is of value at least
\begin{align*}
	\max\{f(S_\ell), f(S_B)\}
	\geq{} &
	\frac{(p + 1) \cdot f(S_\ell) + \alpha^{-1} \cdot f(S_B)}{p + 1 + \alpha^{-1}}
	\geq
	\frac{f(OPT) - \rho(d - |OPT \cap B|/\lambda)}{p + 1 + \alpha^{-1}}\\
	\geq{} &
	\frac{f(OPT) - \rho^*(d - |OPT \cap B|/\lambda)}{p + 1 + \alpha^{-1}}
	=
	\frac{\lambda \rho^*}{\lambda + 1}
	\enspace.
	\qedhere
\end{align*}
\end{proof}

In Appendix~\ref{app:nearly_linear_main} we show a modified version of Algorithm~\ref{alg:main} that appears as Algorithm~\ref{alg:main_nearly_linear}, accepts a quality control parameter $\eps \in (0, 1/4)$ and employs the thresholding speedup technique due to~\cite{badanidiyuru2014fast}. Formally, the properties of Algorithm~\ref{alg:main_nearly_linear} are given by the following variant of Proposition~\ref{prop:main_guarantee}.

\begin{restatable}{proposition}{propMainGuaranteeNearlyLinear} \label{prop:main_guarantee_nearly_linear}
Let
\[
	\rho^* = \frac{(1 - \eps)f(OPT)}{((1 + \eps)p + 1 + \alpha^{-1}) / (1 + \lambda^{-1}) + d - |OPT \cap B| / \lambda}
	\enspace.
\]
There exists an event $\tilde{E}$ such that, if $\rho \geq \rho^*$ and the event $\tilde{E}$ happened, or $\rho \leq \rho^*$ and the event $\tilde{E}$ did not happen, then Algorithm~\ref{alg:main_nearly_linear} outputs a solution of value at least $\lambda \rho^* / (\lambda + 1)$. Furthermore, if $\rho \in [(1 - \delta)\rho^*, \rho^*]$ for some value $\delta \in (0, 1]$, then, regardless of the realization of the event $\tilde{E}$, Algorithm~\ref{alg:main_nearly_linear} outputs a solution of value at least $(1 - \delta)\lambda \rho^* / (\lambda + 1)$. Finally, Algorithm~\ref{alg:main_nearly_linear} runs in $O(\lambda n d + n\eps^{-1}(\log n + \log \eps^{-1}) + T_B)$ time, where $T_B$ is the time complexity of $\BigAlg$.
\end{restatable}

\subsection{Guessing \texorpdfstring{$\rho$}{Rho}} \label{ssc:rho_search}

To use Algorithm~\ref{alg:main} (or the faster Algorithm~\ref{alg:main_nearly_linear}), one must supply a value for $\rho$. Furthermore, according to Proposition~\ref{prop:main_guarantee}, it is best if the supplied value of $\rho$ is close to $\rho^*$ (throughout this section, unless stated explicitly otherwise, $\rho^*$ refers to its value as defined in Proposition~\ref{prop:main_guarantee}). In this section we present an algorithm that manages to supply such a value for $\rho$ using binary search. However, before presenting this algorithm, we first need to find a relatively small range which is guaranteed to include $\rho^*$. Let $\underline{\alpha}$ be a known lower bound on the value of $\alpha$.
\begin{observation} \label{obs:rho_star_limits}
It always holds that
\[
	\frac{1}{p + 1 + \underline{\alpha}^{-1} + d}
	\leq
	\frac{\rho*}{\max_{u \in \cN} f(u)}
	\leq
	\frac{2n}{p}
	\enspace.
\]
\end{observation}
\begin{proof}
According to the definition of $\rho^*$,
\begin{align*}
	\rho^*
	={} &
	\frac{f(OPT)}{(p + 1 + \alpha^{-1}) / (1 + \lambda^{-1}) + d - |OPT \cap B| / \lambda}\\
	\leq{} &
	\frac{n \cdot \max_{u \in \cN} f(\{u\})}{p / (1 + \lambda^{-1}) + d - |OPT \cap B| / \lambda}
	\leq
	\frac{2n \cdot \max_{u \in \cN} f(\{u\})}{p}
	\enspace,
\end{align*}
where the first inequity follows from the submodularity and non-negativity of $f$, and the second inequality follows from the upper bound on $|OPT \cap B|$ given in the discussion before Proposition~\ref{prop:main_guarantee} and the inequality $\lambda \geq 1$.

Similarly,
\begin{align*}
	\rho^*
	={} &
	\frac{f(OPT)}{(p + 1 + \alpha^{-1}) / (1 + \lambda^{-1}) + d - |OPT \cap B| / \lambda}\\
	\geq{} &
	\frac{\max_{u \in \cN} f(\{u\})}{(p + 1 + \alpha^{-1}) / (1 + \lambda^{-1}) + d - |OPT \cap B| / \lambda}
	\geq
	\frac{\max_{u \in \cN} f(\{u\})}{p + 1 + \underline{\alpha}^{-1} + d}
	\enspace,
\end{align*}
where the first inequality holds since every singleton is a feasible set by our assumptions, and the second inequality holds since $1 + \lambda^{-1} \geq 1$ and $|OPT \cap B| / \lambda \geq 0$.
\end{proof}

We are now ready to present, as Algorithm~\ref{alg:complete}, the algorithm that avoids the need to guess $\rho$. This algorithm gets a quality control parameter $\delta \in (0, 1)$ in addition to the parameter $\lambda$  of Algorithm~\ref{alg:main}. In Algorithm~\ref{alg:complete} we use the shorthand $\rho(i) \triangleq (1 + \delta)^i \cdot \max_{u \in \cN} f(\{u\}) / (p + 1 + \underline{\alpha}^{-1} + d)$.
\begin{algorithm2e}
\DontPrintSemicolon
\caption{\texttt{$\rho$ Guessing Algorithm}$(\lambda, \delta)$} \label{alg:complete}
Let $\underline{i} \gets 0$, $\bar{i} \gets \left\lceil \log_{1 + \delta} \frac{2n}{p} - \log_{1 + \delta}  \frac{1}{p + 1 + \underline{\alpha}^{-1} + d} \right\rceil$ and $k \gets 0$.\\
\While{$\bar{i} - \underline{i} > 1$}
{
	Update $k \gets k + 1$.\\
	Let $i_k \gets \lceil (\underline{i} + \bar{i}) / 2 \rceil$.\\
	Execute Algorithm~\ref{alg:main} with $\rho = \rho(i_k)$. Let $A_k$ denote the output set of this execution of Algorithm~\ref{alg:main}, and let $E_k$ denote the event $E$ for the execution.\\
	\lIf{the event $E_k$ happened}{Update $\underline{i} \gets i_k$.}
	\lElse{Update $\bar{i} \gets i_k$.}
}
Execute Algorithm~\ref{alg:main} with $\rho = \rho(\underline{i})$. Let $A'$ denote the output set of this execution of Algorithm~\ref{alg:main_nearly_linear}.\\
\Return{the set maximizing $f$ in $\{A'\} \cup \{A_{k'} \mid 1 \leq k' \leq k\}$}. 
\end{algorithm2e}

\begin{proposition} \label{prop:complete_guarantee}
Algorithm~\ref{alg:complete} has a time complexity of $O((\lambda n d + nr + T_B) \cdot (\log \delta^{-1} + \log(\log n + \log (\underline{\alpha}^{-1} + d)))$ and outputs a set of value at least
\[
	\frac{(1 - \delta)\lambda \rho^*}{\lambda + 1}
	\geq
	\frac{(1 - \delta)\lambda \cdot f(OPT)}{\lambda(p + 1 + \alpha^{-1}) + (\lambda + 1)(d - |OPT \cap B|/\lambda)}
	\enspace.
\]
\end{proposition}
\begin{proof}
The number of iterations done by Algorithm~\ref{alg:complete} is at most 
\begin{align*}
	&
	2 + \log_2 \left[\log_{1 + \delta} \frac{2n}{p} - \log_{1 + \delta}  \frac{1}{p + 1 + \underline{\alpha}^{-1} + d}\right]\\
	={} &
	2 + \log_2 [\ln (2n) - \ln(p) + \ln(p + 1 + \underline{\alpha}^{-1} + d)] - \log_2 \ln(1 + \delta)\\
	\leq{} &
	3 + \log_2 [\ln n + \ln 2 + \ln(2 + \underline{\alpha}^{-1} + d)] + \log_2 \delta^{-1}
	=
	O(\log \delta^{-1} + \log(\log n + \log (\underline{\alpha}^{-1} + d)))
	\enspace.
\end{align*}
The time complexity guaranteed by the proposition now follows by multiplying the last expression with the time complexity of Algorithm~\ref{alg:main} given by Proposition~\ref{prop:main_guarantee}.

It remains to prove the approximation guarantee stated in the proposition. Assume towards a contradiction that this approximation guarantee does not hold, and let us show that this results in a contradiction. Our first objective is to show that this assumption implies that the inequality
\begin{equation} \label{eq:good_local_search}
	\rho(\underline{i})
	\leq
	\rho^*
	\leq
	\rho(\bar{i})
\end{equation}
holds throughout the execution of Algorithm~\ref{alg:complete}. Immediately after the initialization of $\underline{i}$ and $\bar{i}$, this inequality holds by Observation~\ref{obs:rho_star_limits}. Assume now that Inequality~\eqref{eq:good_local_search} held before some iteration of the loop of Algorithm~\ref{alg:complete}, and let us explain why it holds also after the iteration. By Proposition~\ref{prop:main_guarantee}, our assumption (that the approximation guarantee does not hold) implies that, in every iteration of the loop of Algorithm~\ref{alg:complete}, the event $E_k$ happened if and only if $\rho(i_k) < \rho^*$, and therefore, Algorithm~\ref{alg:complete} updates $\underline{i}$ in the iteration when $\rho(i_k) < \rho^*$ and updates $\bar{i}$ in the iteration when $\rho(i_k) \geq \rho^*$.

By Inequality~\eqref{eq:good_local_search} and the observation that $\bar{i} - \underline{i} \leq 1$ when Algorithm~\ref{alg:complete} terminates, we get
\[
	(1 - \delta)\rho^*
	\leq
	(1 - \delta)\rho(\bar{i})
	\leq
	\rho(\underline{i})
	\leq
	\rho^*
	\enspace.
\]
However, the last inequality implies $f(A') \geq (1 - \delta)\lambda \rho^* / (\lambda + 1)$ by Proposition~\ref{prop:main_guarantee}, and hence, completes the proof.
\end{proof}

In Appendix~\ref{app:nearly_linear_complete} we give, as Algorithm~\ref{alg:complete_nearly_linear}, a modified version of Algorithm~\ref{alg:complete} that is based on Algorithm~\ref{alg:main_nearly_linear} instead of Algorithm~\ref{alg:main}. We prove for this algorithm the following proposition.

\begin{proposition} \label{prop:complete_guarantee_nearly_linear}
Algorithm~\ref{alg:complete_nearly_linear} has a time complexity of $O((\lambda n d + n\eps^{-1}(\log n + \log \eps^{-1}) + T_B) \cdot (\log \delta^{-1} + \log(\log n + \log (\underline{\alpha}^{-1} + d)))$ and outputs a set of value at least
\[
	\frac{(1 - \delta)\lambda \rho^*}{\lambda + 1}
	\geq
	\frac{(1 - \delta - \eps)\lambda \cdot f(OPT)}{\lambda((1 + \eps)p + 1 + \alpha^{-1}) + (\lambda + 1)(d - |OPT \cap B|/\lambda)}
	\enspace,
\]
where $\rho^*$ represents here the value stated in Proposition~\ref{prop:main_guarantee_nearly_linear}.
\end{proposition}

\subsection{Proof of Theorem~\ref{thm:SMKS}} \label{ssc:big_elements_algorithms}

In this section we use the machinery developed so far to prove Theorem~\ref{thm:SMKS}.
To use Algorithms~\ref{alg:complete} or~\ref{alg:complete_nearly_linear} one must choose the algorithm {\BigAlg}. A very simple choice is to use an algorithm that just returns the best singleton subset of $B$. This leads to the following theorem.
\begin{theorem} \label{thm:smks_nearly_linear}
For every value $\eps \in (0, 1/4]$, there is a polynomial time $[(1 + O(\eps))(p + 1 + \tfrac{7}{4}d)]^{-1}$-approximation algorithm for {\SMKS} that runs in $\tilde{O}(nd + n/\eps)$ time.
\end{theorem}
\begin{proof}
Given the above choice for {\BigAlg}, we can set $\alpha = |OPT \cap B|^{-1}$ because the submodularity of $f$ implies
\[
	\max_{u \in B} f(\{u\})
	\geq
	\frac{\sum_{u \in OPT \cap B} f(\{u\})}{|OPT \cap B|}
	\geq
	\frac{f(OPT \cap B)}{|OPT \cap B|}
	\enspace.
\]
Thus, for $\lambda \geq 2$, it is valid to choose $\underline{\alpha}^{-1} = d(\lambda - 1)$ because of the upper bound on $|OPT \cap B|$ explained in the discussion before Proposition~\ref{prop:main_guarantee}.

If we now choose to use Algorithm~\ref{alg:complete_nearly_linear} and set $\lambda = 2$ and $\delta = \eps$, we get by Proposition~\ref{prop:complete_guarantee_nearly_linear} that the inverse of the approximation ratio of the algorithm we consider is at most
\begin{align*}
	&
	\frac{\lambda((1 + \eps)p + 1 + \alpha^{-1}) + (\lambda + 1)(d - |OPT \cap B|/\lambda)}{(1 - \delta - \eps)\lambda}\\
	={} &
	\frac{(1 + \eps)p + 1 + \alpha^{-1}(1 - \lambda^{-1} - \lambda^{-2}) + d(1 + \lambda^{-1})}{1 - \delta - \eps}\\
	\leq{} &
	\frac{(1 + \eps)p + 1 + d(\lambda - 1)(1 - \lambda^{-1} - \lambda^{-2}) + d(1 + \lambda^{-1})}{1 - \delta - \eps}
	=
	\frac{(1 + \eps)p + 1 + d(\lambda - 1 + \lambda^{-1} + \lambda^{-2})}{1 - \delta - \eps}\\
	={} &
	\frac{(1 + \eps)p + 1 + \frac{7}{4}d}{1 - \delta - \eps}
	\leq
	(1 + 3\delta + 3\eps)(p + 1 + \tfrac{7}{4}d)
	=
	(1 + 6\eps)(p + 1 + \tfrac{7}{4}d)
	\enspace,
\end{align*}
and its time complexity is
\begin{align*}
	&
	O((\lambda n d + n\eps^{-1}(\log n + \log \eps^{-1}) + T_B) \cdot (\log \delta^{-1} + \log(\log n + \log (\underline{\alpha}^{-1} + d)))\\
	={} &
	O((n d + n\eps^{-1} (\log n + \log \eps^{-1}) + T_B) \cdot (\log \eps^{-1} + \log(\log n + \log (2d)))\\
	={} &
	\tilde{O}(nd + n/\eps + T_B)
	=
	\tilde{O}(nd + n/\eps)
	\enspace,
\end{align*}
where the last equality holds because {\BigAlg} can be implemented by simply scanning all the $n$ possible singleton sets, which requires a time complexity of $T_B = O(n)$.
\end{proof}

A slightly more involved option for {\BigAlg} is an algorithm that enumerates over all sets of up to two big elements, and outputs the best such set that is feasible. This leads to the following theorem.
\begin{theorem} \label{thm:smks_quadratic}
For every value $\eps \in (0, 1/4]$, there is a polynomial time $(p + 1.5556 + \tfrac{13}{9}d + \eps)$-approximation algorithm for {\SMKS} that runs in $\tilde{O}(n^2 + nd)$ time.
\end{theorem}
\begin{proof}
In this proof we need the following known lemma.
\begin{lemma}[Lemma~2.2 of~\cite{feige2011maximizing}] \label{lem:sampling}
Let $f\colon 2^\cN \to \bR$ be a submodular function, let $A$ be an arbitrary subset of $\cN$, and let $A(p)$ be a random subset of $A$ containing every element of $A$ with probability $p$ (not necessarily independently). Then,
\[
	\bE[f(A(p))]
	\geq
	(1 - p) \cdot f(\varnothing) + p \cdot f(A)
	\enspace.
\]
\end{lemma}

Next, we show that one is allowed to choose $\alpha = \min\{2 / |OPT \cap B|, 1\}$ for the above described {\BigAlg}. If $|OPT \cap B| \leq 2$ then this is trivial. Otherwise, by choosing $R$ to be a uniformly random subset of $OPT \cap B$ of size $2$, we get that the value of the output of {\BigAlg} is at least
\[
	\bE[f(R)]
	\geq
	\left(1 - \frac{2}{|OPT \cap B|}\right) \cdot f(\varnothing) + \frac{2}{|OPT \cap B|} \cdot f(OPT \cap B)
	\geq
	\frac{2}{|OPT \cap B|} \cdot f(OPT \cap B)
	\enspace,
\]
where the first inequality follows from Lemma~\ref{lem:sampling} and the observation that $R$ includes every element of $OPT \cap B$ with probability exactly $2 / |OPT \cap B|$. Hence, it is valid to set $\underline{\alpha}^{-1} = \max\{d(\lambda - 1) / 2, 1\}$ because of the upper bound on $|OPT \cap B|$ explained in the discussion before Proposition~\ref{prop:main_guarantee}.

Consider now the algorithm obtained by plugging $\delta = \frac{\eps}{2p + 4d + 4}$, $\lambda = 3$ and the above choice for {\BigAlg} into Algorithm~\ref{alg:complete}. We would like to show that this algorithm has all the properties guaranteed by the theorem. By Proposition~\ref{prop:complete_guarantee}, the time complexity of this algorithm is
\begin{align*}
	&
	O((\lambda n d + nr + T_B) \cdot (\log \delta^{-1} + \log(\log n + \log (\underline{\alpha}^{-1} + d)))\\
	={} &
	O((n d + nr + T_B) \cdot (\log (p + d) + \log \eps^{-1} + \log(\log n + \log d))\\
	={} &
	\tilde{O}(nr + nd + T_B)
	=
	\tilde{O}(n^2 + nd)
	\enspace,
\end{align*}
where the last equality holds since a brute force implementation of {\BigAlg} runs in $O(n^2)$ time.

It remains to analyze the approximation ratio of our suggested algorithm. Towards this goal, we need to consider a few cases. If $|OPT \cap B| \geq 2$, then, by Proposition~\ref{prop:complete_guarantee}, the inverse of the approximation ratio of the above algorithm is no worse than
\begin{align*}
	&
	\frac{\lambda(p + 1 + \alpha^{-1}) + (\lambda + 1)(d - |OPT \cap B|/\lambda)}{(1 - \delta)\lambda}\\
	={} &
	\frac{p + 1 + |OPT \cap B| \cdot (1/2 - \lambda^{-1} - \lambda^{-2}) + d(1 + \lambda^{-1})}{1 - \delta}\\
	\leq{} &
	\frac{p + 1 + d(\lambda - 1) \cdot(1/2 - \lambda^{-1} - \lambda^{-2}) + d(1 + \lambda^{-1})}{1 - \delta}
	=
	\frac{p + 1 + d(\lambda/2 - 1/2 + \lambda^{-1} + \lambda^{-2})}{1 - \delta}\\
	={} &
	\frac{p + 1 + \frac{13}{9}d}{1 - \delta}
	\leq
	(1 + 2\delta)(p + 1 + \tfrac{13}{9}d)
	\leq
	p + 1 + \tfrac{13}{9}d + \eps
	\enspace.
\end{align*}
Otherwise, if $|OPT \cap B| \leq 1$, then, by the same proposition, the inverse of the approximation ratio of the above algorithm is at most
\begin{align*}
	&
	\frac{\lambda(p + 1 + \alpha^{-1}) + (\lambda + 1)(d - |OPT \cap B|/\lambda)}{(1 - \delta)\lambda}
	\leq
	\frac{p + 2 + d(1 + \lambda^{-1})}{1 - \delta}\\
	={} &
	\frac{p + 2 + \tfrac{4}{3}d}{1 - \delta}
	\leq
	(1 + 2\delta)(p + 2 + \tfrac{4}{3}d)
	\leq
	p + 2 + \tfrac{4}{3}d + \eps
	\enspace.
\end{align*}

The above inequalities prove an approximation ratio which is a bit weaker than the one guaranteed by the theorem. To get exactly the approximation ratio guaranteed by the theorem, we need to be a bit more careful with the last case. First, for $|OPT \cap B| = 1$,
\begin{align*}
	&
	\frac{\lambda(p + 1 + \alpha^{-1}) + (\lambda + 1)(d - |OPT \cap B|/\lambda)}{(1 - \delta)\lambda}
	=
	\frac{p + 2 - 1/\lambda - 1/\lambda^2 + d(1 + \lambda^{-1})}{1 - \delta}\\
	={} &
	\frac{p + \tfrac{14}{9} + \tfrac{4}{3}d}{1 - \delta}
	\leq
	(1 + 2\delta)(p + 1.5556 + \tfrac{4}{3}d)
	\leq
	p + 1.5556 + \tfrac{4}{3}d + \eps
	\enspace.
\end{align*}
Handling the case $OPT \cap B = \varnothing$ is a bit more involved. Since there are no big elements in this case in $OPT$, there is no need to take $S_B$ into account in the analysis of Algorithm~\ref{alg:complete}. It can be observed that by repeating the analysis, but ignoring this set, we can get that the value of the output set of this algorithm is also lower bounded by
\[
	\frac{(1 - \delta)\lambda \cdot f(OPT \setminus B)}{\lambda(p + 1) + (\lambda + 1)(d - |OPT \cap B|/\lambda)}
	\enspace,
\]
which in our case implies that the inverse of the approximation ratio of our algorithm is at most
\[
	\frac{\lambda(p + 1) + d(\lambda + 1)}{(1 - \delta)\lambda}
	=
	\frac{p + 1 + d(1 + \lambda^{-1})}{1 - \delta}
	=
	\frac{p + 1 + \tfrac{4}{3}d}{1 - \delta}
	\leq
	p + 1 + \tfrac{4}{3}d + \eps
	\enspace.
	\qedhere
\]
\end{proof}

Theorem~\ref{thm:SMKS} now follows immediately from Theorems~\ref{thm:smks_nearly_linear} and~\ref{thm:smks_quadratic}.

\noindent \textbf{Remark:} It is natural to consider also candidates for {\BigAlg} that enumerate over larger subsets of $B$. However, this will require $\Omega(n^3)$ time, and is, therefore, of little interest as one can obtain a clean $(p + d + 1)$-approximation for {\SMKS} in $\tilde{O}(n^3)$ time (see Section~\ref{ssc:related_work}).
\section{Experimental Results} \label{sec:experiments}
We have studied the submodular maximization problem subject to different constraints. 
In this section, we compare our proposed algorithms with the state-of-the-art algorithms under the following constraints: (i) a carnality constraint (\cref{sec:experiment-cardinality}), (ii) a single knapsack constraint (\cref{sec:experiment-knapsack}), and (iii) combination of a $p$-set system and $d$ knapsack constraints (\cref{sec:experiment-system}). 

\subsection{Cardinally Constraint} \label{sec:experiment-cardinality}
Cardinality constraint is the most widely studied setup in the submodular maximization literature. We compare our algorithm, \AlgFTG (abbreviated with \FTG in the figures) with \AlgLG and \AlgBR \cite{kuhnle2021quick} under the cardinality constraint. \AlgLG is an efficient implementation of the na\"ive \AlgG algorithm in which the diminishing returns property of submodular functions is used to avoid oracle queries that are known to provide little gain~\cite{mirzasoleiman2015lazier}. It is well-known that \AlgLG leads to several orders of magnitude speedups over \AlgG in practice.

In our first experiment, we implement a movie recommender system by finding a diverse set of movies for a user. We adopt the approach of \citet{lindgren2015sparse} to extract features for each movie by using ratings from the MovieLens dataset \cite{harper2015movielens}.
For a given set of movies $\cN$ $(|\cN| = n)$, let vector $v_i$ represent the attributes of the $i$-th movie. The similarity matrix $M_{n \times n}$ between movies is defined by $M_{ij} = e^{-\lambda \cdot  \text{dist}(v_i,v_j)}$, where $\text{dist}(v_i,v_j)$ is the euclidean distance between the vectors $v_i, v_j \in \cN$. 
For this application, we used the following monotone and submodular function to quantify the diversity of a given set of movies $S$: $f(S) = \log \det (\mathbf{I} + \alpha M_S)$, where $\mathbf{I}$ is the identity matrix and $M_S$ is a square sub-matrix of $M$ consisting of the rows and columns corresponding to the movies in the set $S$ \cite{herbrich2003fast}. Our objective is to maximize $f$ under the cardinality constraint $k$. In \cref{fig:LG-movie-utility}, we compare the utility of the algorithms on this instance. We observe that \AlgFTG with $\varepsilon \in \{0.1, 0.2 \}$ performs as good as \AlgLG, and $\AlgBR$ performs slightly worse. In \cref{fig:LG-movie-oracle}, we observe that the query complexity of \AlgFTG is significantly less than that of \AlgLG.
It is interesting to note that, as is guaranteed by our theoretical results, the number of oracle calls for \AlgFTG is (almost) not increasing with $k$. 

In our second experiment, we consider a location summarization application. The end goal is to find a representative summary of tens of thousands of Yelp business locations (from Charlotte, Edinburgh, Las Vegas, Madison, Phoenix and  Pittsburgh) by using their related attributes \cite{badanidiyuru2020submodular}.
To quantify how well a business at a location $i$ represents another business at a location $j$ we use a similarity measure $M_{i, j}= e^{-\lambda \cdot  \text{dist}(v_i,v_j)}$, where $v_i$ and $v_j$ are vectors of attributes of facilities at locations $i$ and $j$.
To select a good representation of all the locations $\cN = \{1, \dots , n \} $, we use the monotone and submodular facility location function \cite{krause12survey}
\begin{equation} \label{eq:facility-location}
f(S) =\frac{1}{n} \sum_{i=1}^{n} \max_{j \in S} M_{i,j} \enspace.
\end{equation}
In \cref{fig:LG-yelp-utility}, we observe that \AlgLG performs slightly better than \AlgFTG in this experiment. On the other hand, \cref{fig:LG-yelp-oracle} shows that the query complexity of \AlgFTG is much lower.

\begin{figure*}[htb!] 
	\centering  
	\subfloat[Movie Recommendation]{\includegraphics[height=32mm]{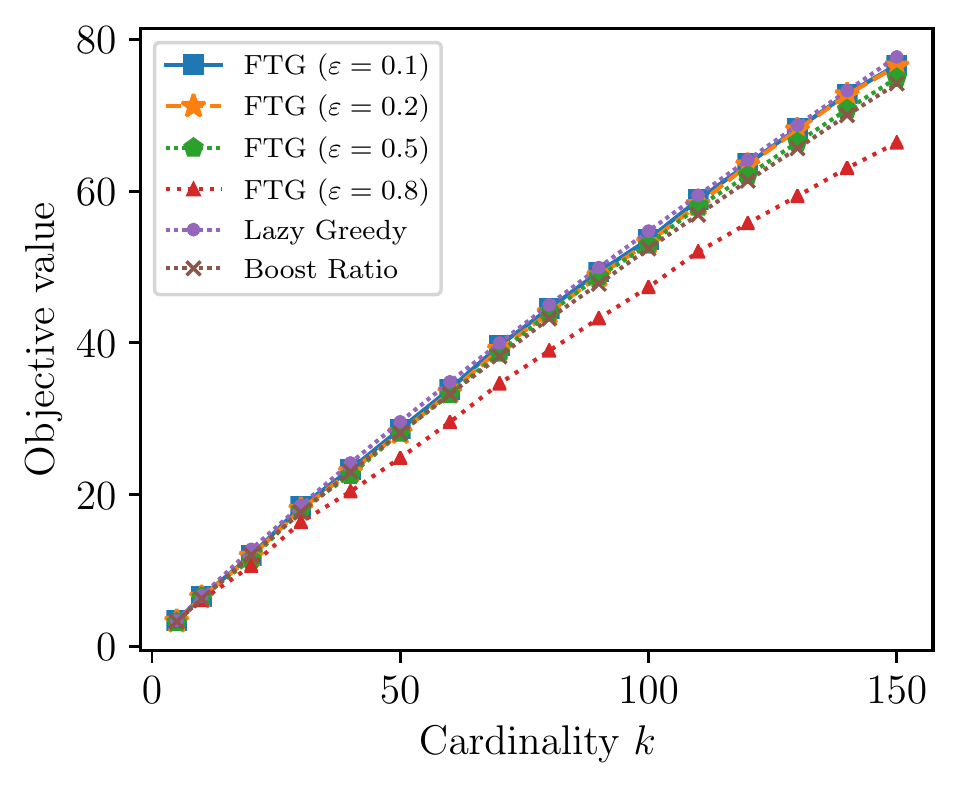}\label{fig:LG-movie-utility}}
	\subfloat[Movie Recommendation]{\includegraphics[height=32mm]{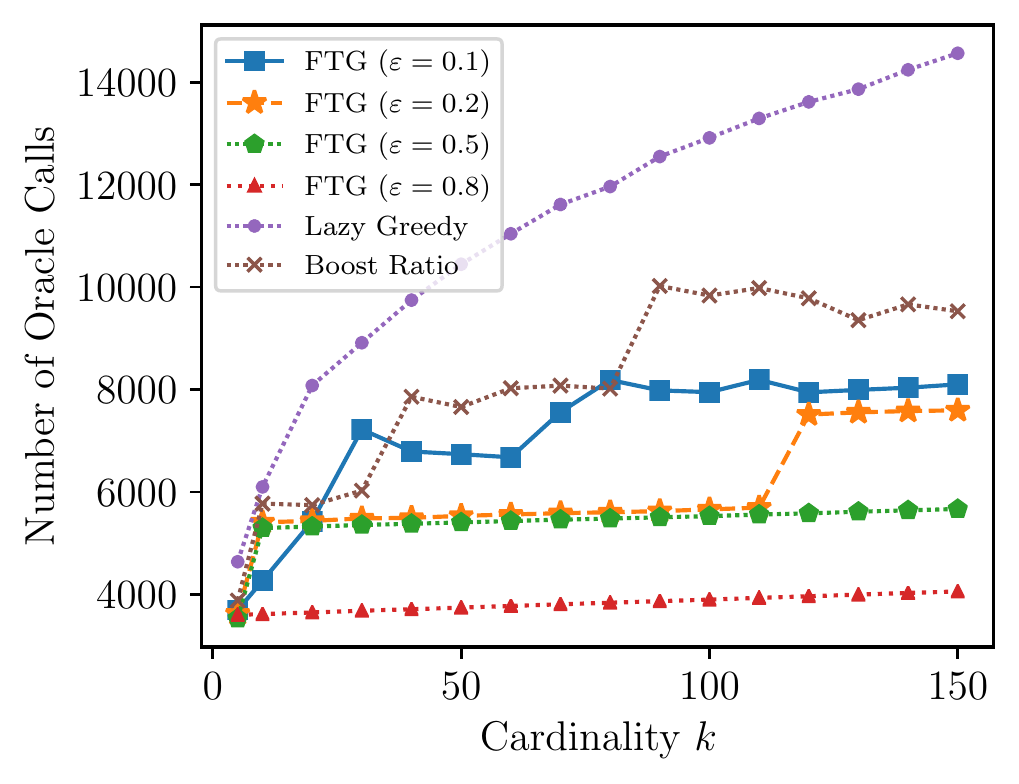}\label{fig:LG-movie-oracle}}
	\subfloat[Location Summarization]{\includegraphics[height=32mm]{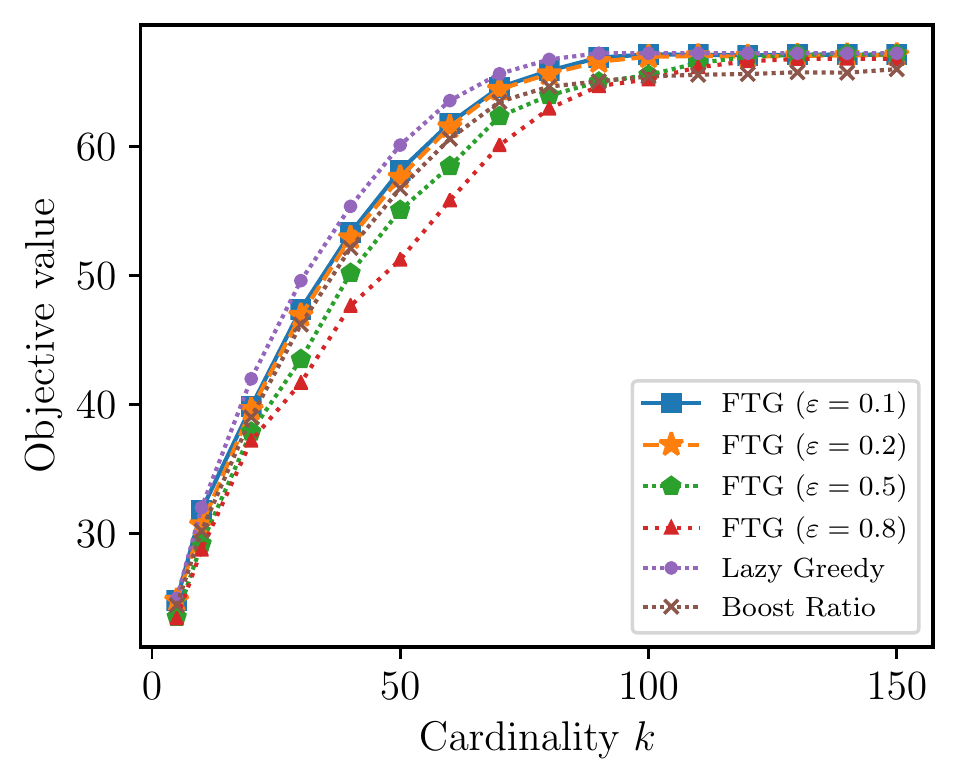}\label{fig:LG-yelp-utility}}
	\subfloat[Location Summarization]{\includegraphics[height=32mm]{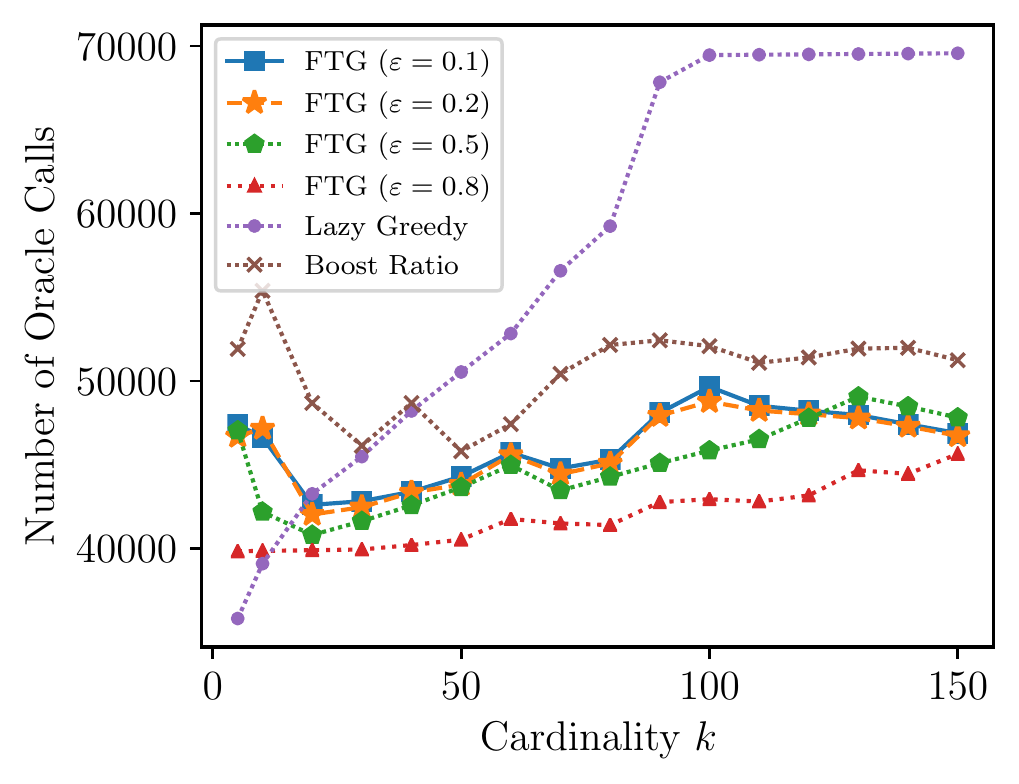}\label{fig:LG-yelp-oracle}}
	\caption{Comparing \AlgFTG (\cref{alg:threshold_greedy}) with \AlgLG and \AlgBR \cite{kuhnle2021quick} under a cardinality constraint.}
	\label{fig:LG}
\end{figure*}

\AlgSG is a fast but randomized approach for maximizing submodular functions subject to a cardinality constraint \cite{mirzasoleiman2015lazier}. In the next experiment, we compare \AlgFTG with \AlgSG under the cardinality constraint. We consider the monotone and submodular vertex cover function. This function is defined over vertices of a (directed) graph $G = (V, E)$. For a given vertex set $S \subseteq V$, let $N(S)$ be the set of vertices which are pointed to by $S$, i.e., $N(S) \triangleq \{v \in V \mid \exists u \in S \text{ s.t. } (u, v) \in E\}$. The vertex cover function $f:2^{V} \to \bR_{\geq 0}$ is then defined as $f(S) = | N(S) \cup S|$. We use two different graphs: (i) a random graph of $n = 10^6$ nodes with an average degree of $2$, where $20$ additional nodes with degrees $50$ are added (the neighbors of these high degree nodes are chosen randomly); (ii) Slashdot social network \cite{snapnets}.
In \cref{fig:sg}, we observe that the utility of \AlgFTG is significantly better than that of \AlgSG. We also observe a high variability in the utility of solutions returned by \AlgSG. Furthermore, \AlgFTG $(\varepsilon = 0.8)$ outperforms \AlgSG $(\varepsilon \in \{ 0.1, 0.2 \})$ in terms of both utility and query complexity.
Note that, while \AlgSG is performing quite well in many practical scenarios, the theoretical guarantee of this algorithm holds only in expectation, and there are cases resulting in high variance. In these high variance cases, one has to run \AlgSG multiple times, which diminishes the benefit from the algorithm. 

\begin{figure*}[htb!] 
	\centering  
	\subfloat[Random graph]{\includegraphics[height=31mm]{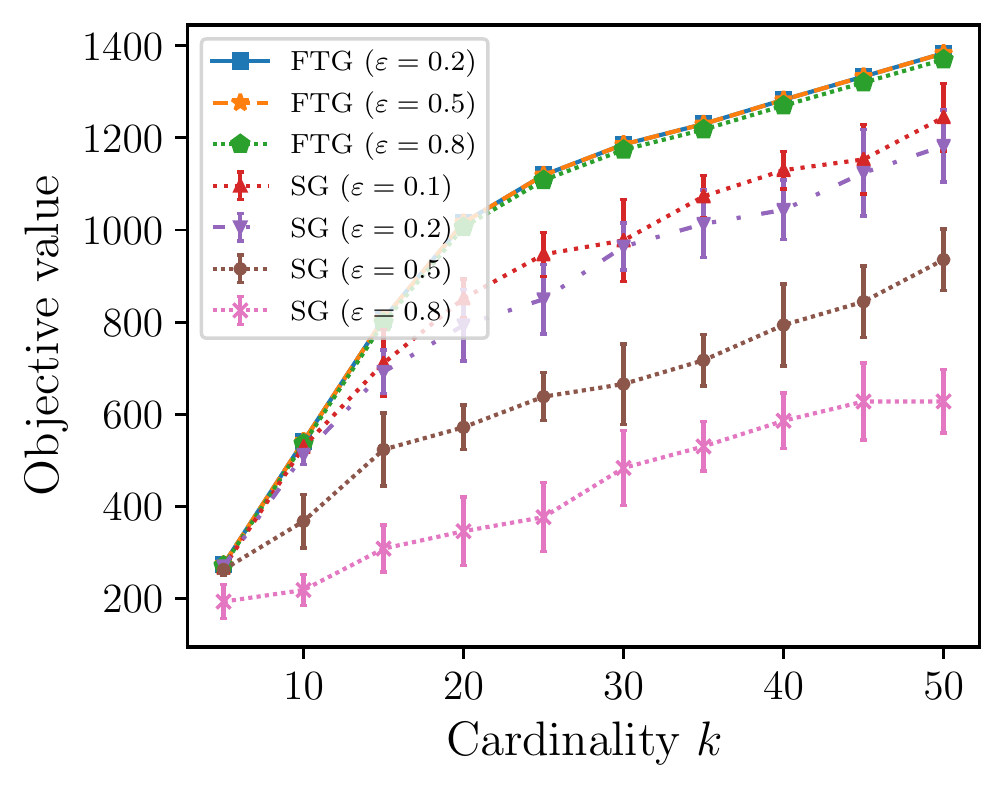}\label{fig:sg-random-utility}}
	\subfloat[Random graph]{\includegraphics[height=31mm]{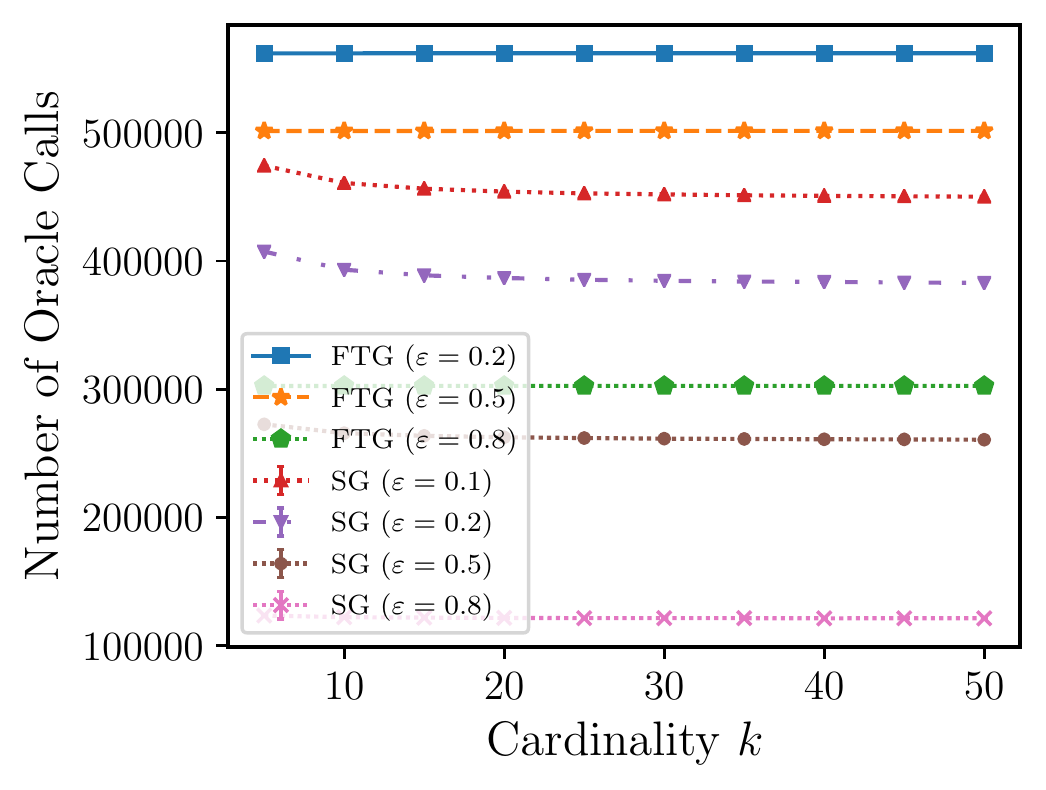}\label{fig:sg-random-call}}
	\subfloat[Slashdot network] {\includegraphics[height=31mm]{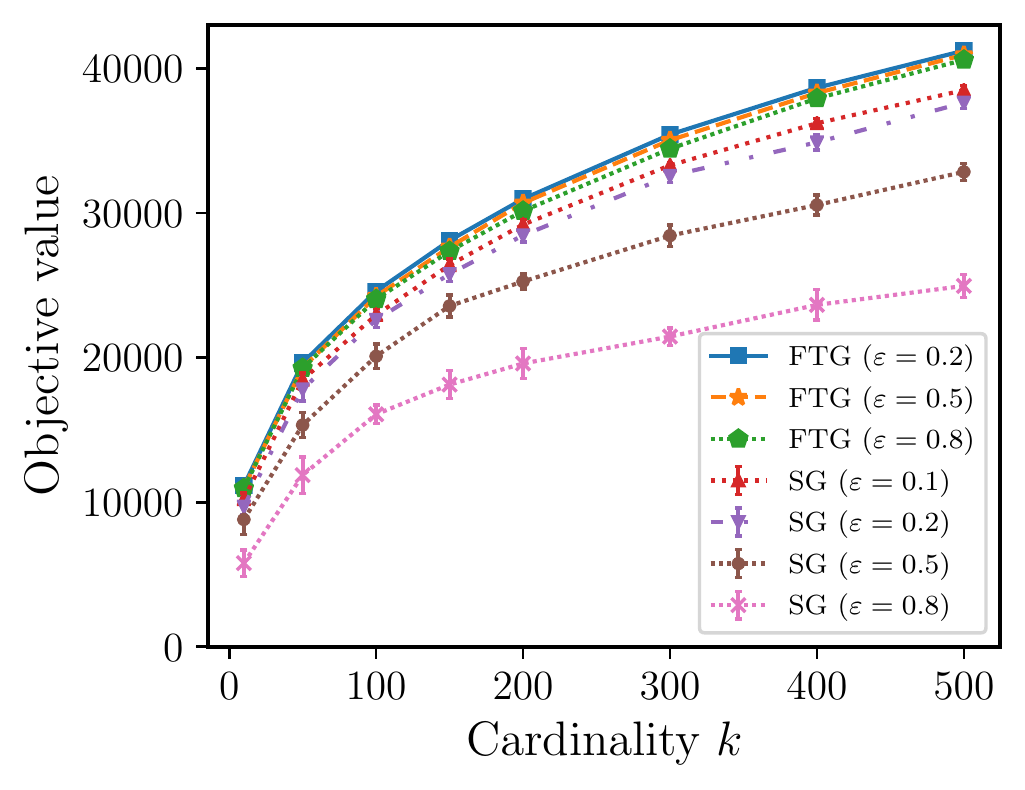}\label{fig:sg-slashdot-utility}}
	\subfloat[Slashdot network] {\includegraphics[height=31mm]{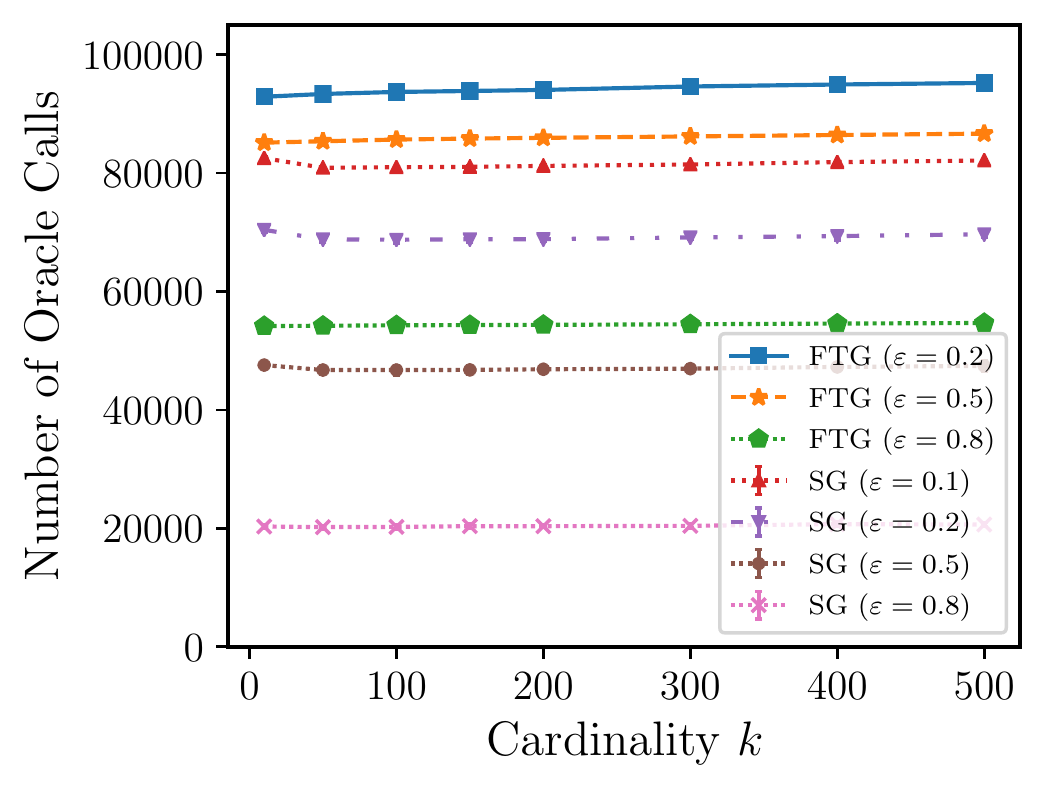}\label{fig:sg-slashdot-call}}
	\caption{Comparing \AlgFTG (\cref{alg:threshold_greedy}) with \AlgSG on vertex cover problems under a cardinality constraint.}
	\label{fig:sg}
\end{figure*}

\subsection{Single Knapsack Constraint} \label{sec:experiment-knapsack}
In this section, we evaluate the performance of \cref{alg:post_processing} (referred to as \FTGP in the figures) with that of \AlgMRT \cite{huang2021improved} and \AlgDG under a single knapsack constraint. \AlgDG greedily adds elements that
maximize the ratio between their marginal gain and knapsack cost (but ignoring elements whose addition will result in a violation of the knapsack constraint). In our first experiment for this constraint, we consider the movie recommendation application from \cref{sec:experiment-cardinality}. The cost of each movie is defined to be proportional to the absolute difference between the rating of that movie and 10 (the maximum rating in iMDB). In this application, the goal is to find a diverse set of movies while guaranteeing that the total rating of the picked movies is high enough \cite{badanidiyuru2020submodular}. From \cref{fig:knapsack-recomm-utiltiy,knapsack-recomm-recall}, we observe that \cref{alg:post_processing} significantly outperforms the other two algorithms with respect to both the utility and number of oracle calls metrics.

For our second experiment under the single knapsack constraint, we consider a Twitter text summarization application with the goal of producing a representative summary of Tweets around the first day of January 2019 \cite{haba2020streaming}. In this application, the objective is to find a diverse summary from the Twitter feeds of news agencies. The monotone and submodular function $f$ used in this task is defined over a ground set $\cN$ of tweets as follows: $f(S) = \sum_{w \in \cW} \sqrt{\sum_{e \in S} \text{score}(w,e)}$, where $\cW$ is the set of all the English words \cite{kazemi2019submodular}. If word $w$ is in Tweet $e$, then we have $\text{score}(w,e) = $ number of retweets $e$ has received. Otherwise, we define $\text{score}(w,e) = 0$.
The cost of each tweet is proportional to the time difference (in months) between the date of that Tweet and the first day of January 2019 \cite{haba2020streaming}. Again, from \cref{fig:knapsack-text-utiltiy,knapsack-text-recall}, it is evident that \cref{alg:post_processing} surpasses the baseline algorithms.

\begin{figure*}[htb!] 
	\centering  
	\subfloat[Movie Recommendation]{\includegraphics[height=32mm]{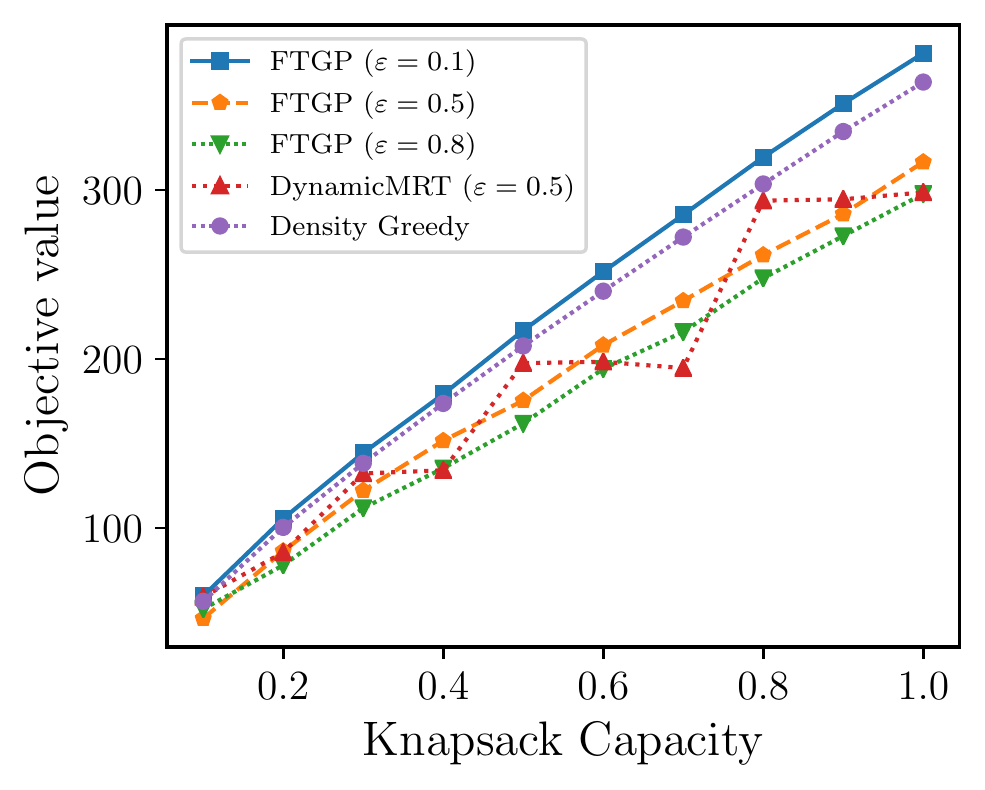}\label{fig:knapsack-recomm-utiltiy}}
	\subfloat[Movie Recommendation]{\includegraphics[height=32mm]{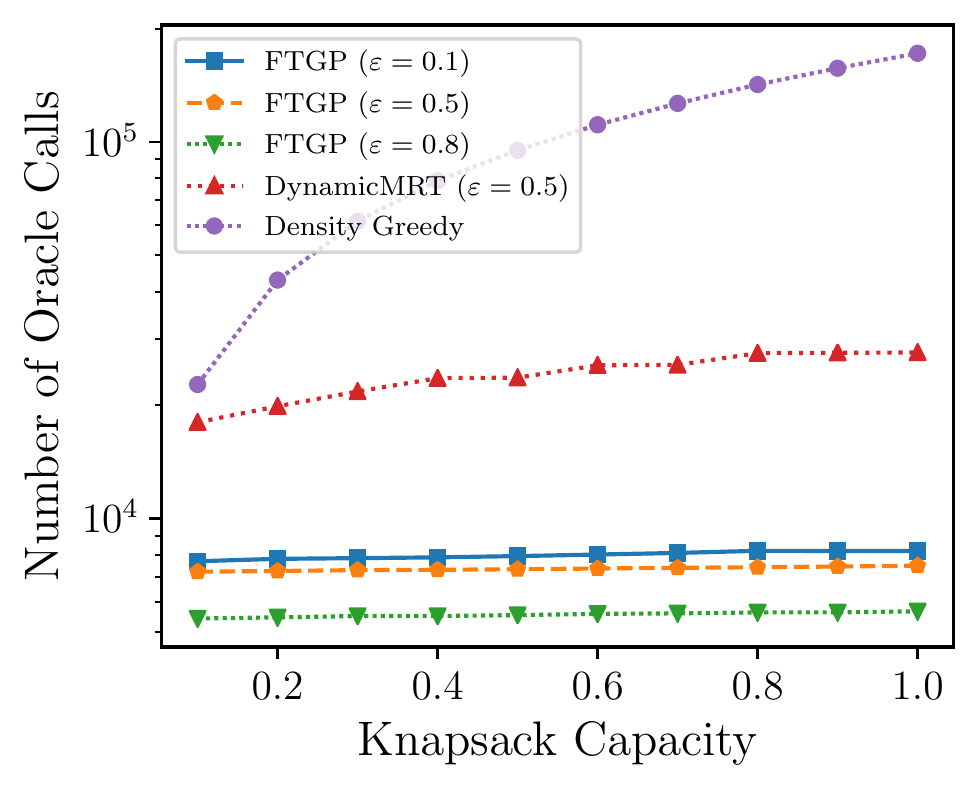}\label{knapsack-recomm-recall}}
	\subfloat[Text Summarization]{\includegraphics[height=32mm]{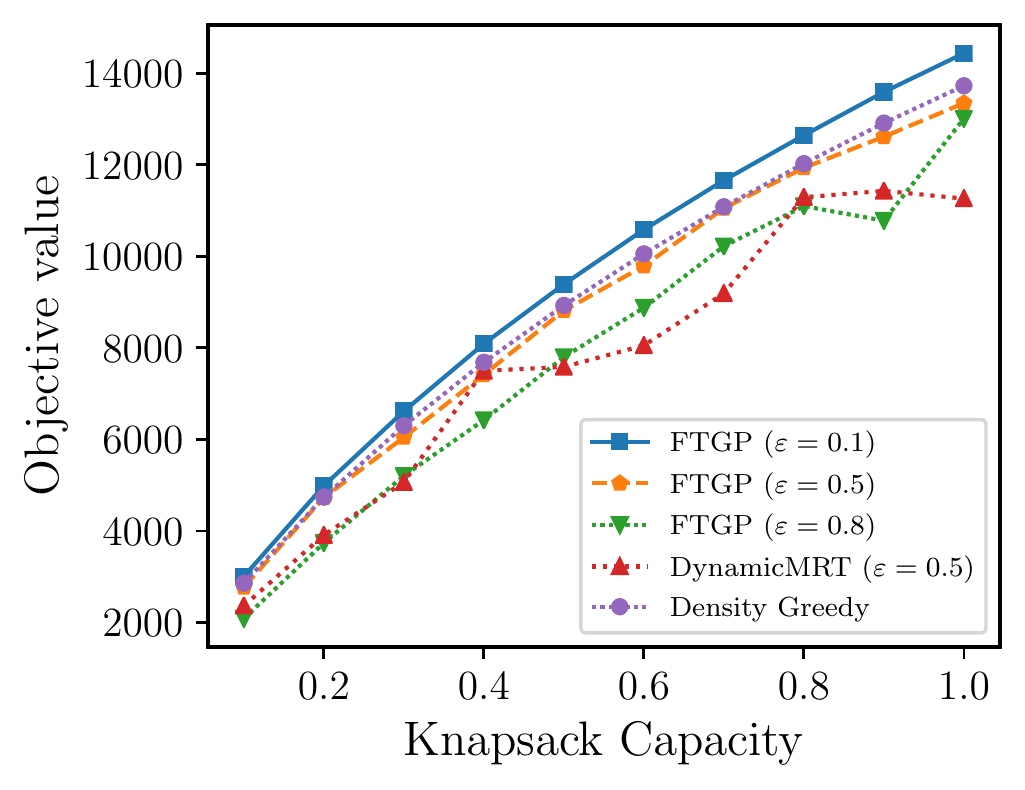}\label{fig:knapsack-text-utiltiy}}
	\subfloat[Text Summarization]{\includegraphics[height=32mm]{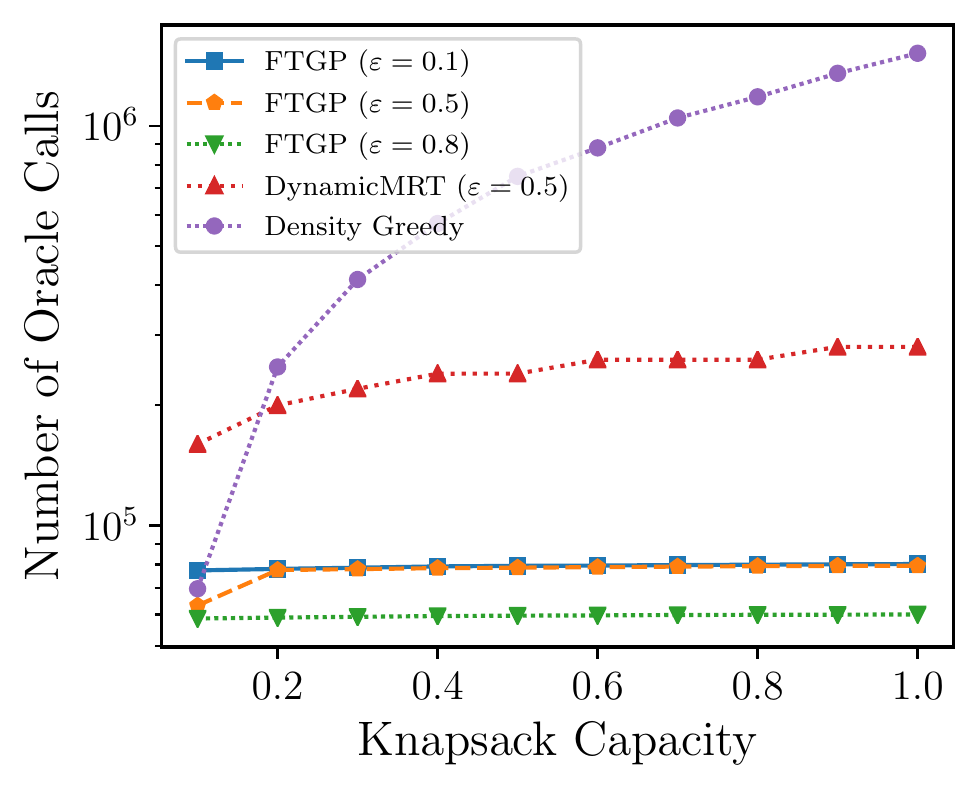}\label{knapsack-text-recall}}
	\caption{Comparing \cref{alg:post_processing} (referred to as \FTGP) with \AlgMRT \cite{huang2021improved} and \AlgDG under a single knapsack constraint.}
	\label{fig:knapsack}
\end{figure*}

\subsection{\texorpdfstring{$p$}{p}-Set System and \texorpdfstring{$d$}{d} Knapsack Constraints} \label{sec:experiment-system}

In the last set of experiments, we compare the performance of \cref{alg:complete_nearly_linear} with several other algorithms under the combination of a $p$-system and $d$ knapsack constraints. We consider \AlgBG \cite{badanidiyuru2020submodular}, \AlgFast \cite{badanidiyuru2014fast}, \AlgG and \AlgDG as our baselines. \AlgG keeps adding elements one by one according to their marginal gains as long as the $p$-system and knapsack constraints allow it. \AlgDG is very similar to \AlgG with the only difference being that it picks elements according to the ratio between their marginal gain and their total knapsack cost. Note that \AlgG and \AlgDG are heuristic algorithms without any theoretical guarantees for the setup of this experiment. 

In the first experiment of this section, we consider the location summarization application from \cref{sec:experiment-cardinality}. The goal is to maximize \eqref{eq:facility-location} subject to the combination of the following constraints: (i) the maximum number of locations from each city is $5$, (ii) the total size of the summary is at most 20, and (iii) two knapsacks $c_1$ and $c_2$, where $c_i(j) = \textrm{distance}(j, \textrm{POI}_i)$ is the normalized distance of location $j$ to a given point of interest in city $j$ (for $c_1$, $\textrm{POI}_1$ is the down-town; and for $c_2$, $\textrm{POI}_2$ is the international airport in that city). Note that the distances are normalized so that one unit of knapsack budget is equivalent to 100km. 
\cref{fig:SMKS-yelp-utility,fig:SMKS-yelp-call} compare the algorithms for varying values of knapsack budget. In terms of utility, we observe that \AlgBG is the best performing algorithm followed by our algorithm (\cref{alg:complete_nearly_linear}). Despite the competitive performance of \cref{alg:complete_nearly_linear}, its query complexity is almost an order of magnitude faster than \AlgBG.

The second application is a video summarization task over a collation of videos from the VSUMM dataset \cite{de2011vsumm}.\footnote{The dataset is available for download from \url{https://sites.google.com/site/vsummsite/}.} The features for each frame of a video are generated by a pre-trained ResNet-18 model \cite{he2016deep,kazemi2021regularized}. The similarity between two frames $i$ and $j$ is defined by $e^{-\lambda \cdot \textrm{dist}(x_i,x_j)}$, where $\textrm{dist}(x_i,x_j)$ is the Euclidean distance between the corresponding features of the two frames. 
Similarly to the movie recommendation applications of \cref{sec:experiment-cardinality,sec:experiment-knapsack}, the goal of this new summarization task is to maximize the monotone and submodular function $f(S) = \log \det (\mathbf{I} + \alpha M_S)$
subject to the combination of the following constraints: (i) the maximum allowed cardinality of the final summary is $k$ (in \cref{fig:SMKS-video-utility,fig:SMKS-video-call}, we compare algorithms by varying this value), (ii) the maximum number of allowed frames from each video is 5, and (iii) a single knapsack constraint. The knapsack cost for each frame $u$ is defined as $\nicefrac{\mathrm{H}(u)}{20}$, where $\mathrm{H}(u)$ is the entropy of $u$.. This extra knapsack constraint allows us to bound the required bits to encode the produced summary by using the entropy of each frame as a proxy for its encoding size. We again observe that \AlgBG and our algorithm (\cref{alg:complete_nearly_linear}) produces the summaries with the highest utilities (see \cref{fig:SMKS-video-utility}). Furthermore, in \cref{fig:SMKS-video-call}, we observe \cref{alg:complete_nearly_linear} is the fastest algorithm. Particularly, \cref{alg:complete_nearly_linear} is several orders of magnitudes faster than \AlgBG and \AlgFast.

\begin{figure*}[htb!] 
	\centering  
	\subfloat[Location Summarization]{\includegraphics[height=32mm]{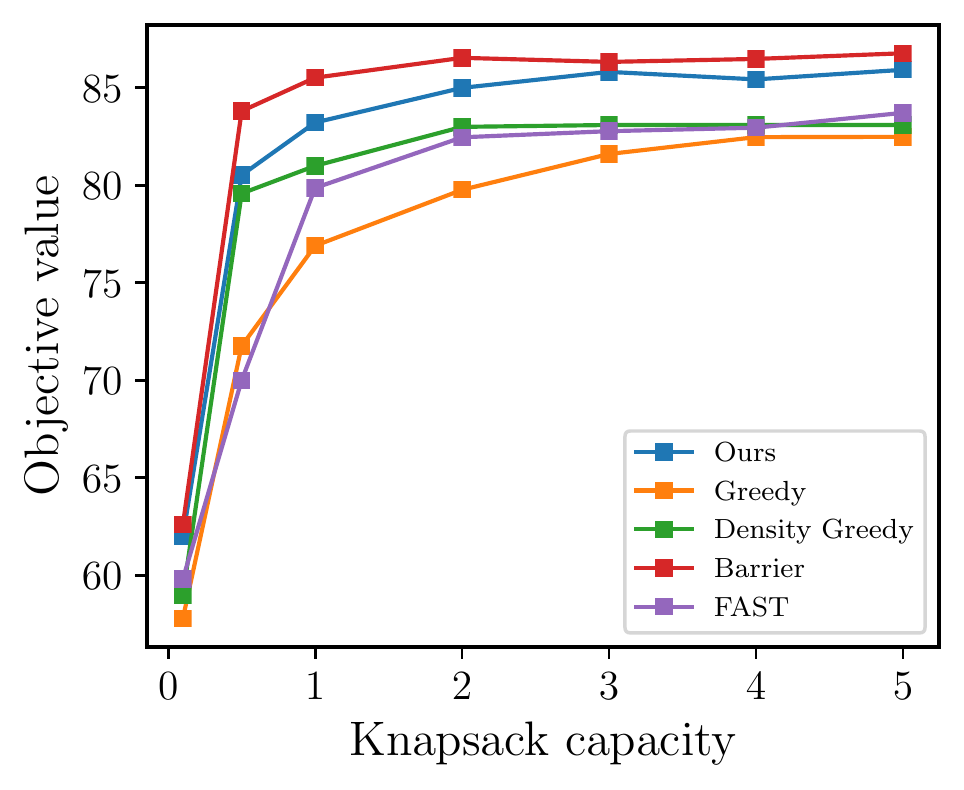}\label{fig:SMKS-yelp-utility}}
	\hspace{2pt}
	\subfloat[Location Summarization]{\includegraphics[height=32mm]{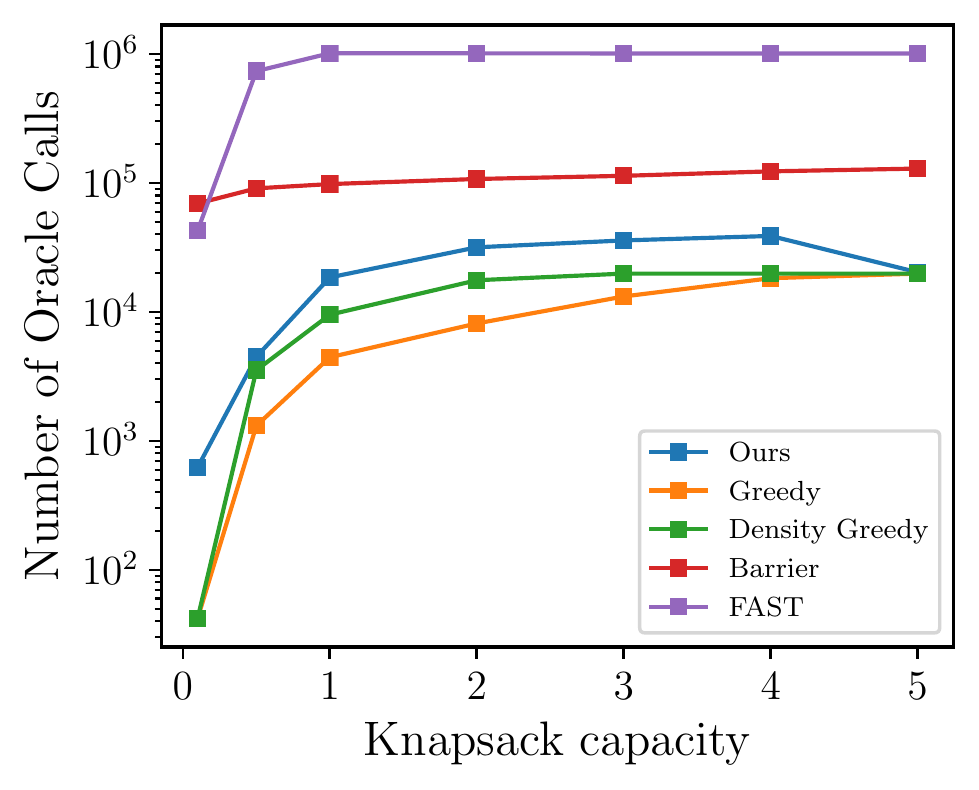}\label{fig:SMKS-yelp-call}}
	\hspace{2pt}
	\subfloat[Video Summarization]{\includegraphics[height=32mm]{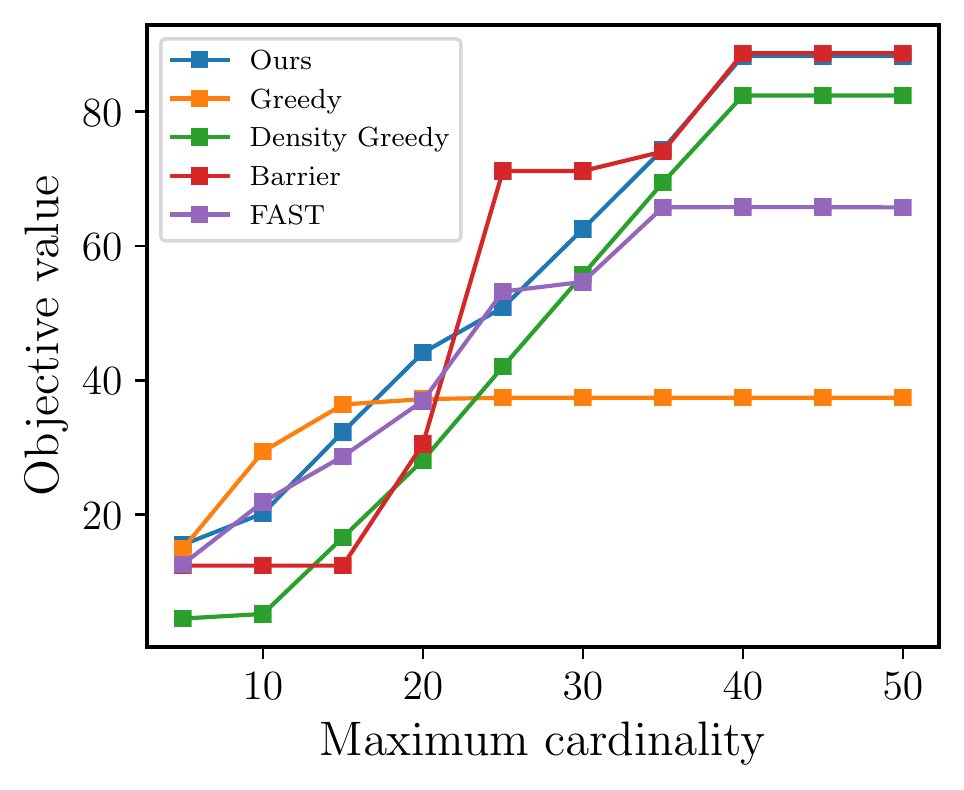}\label{fig:SMKS-video-utility}}
	\hspace{2pt}
	\subfloat[Video Summarization]{\includegraphics[height=32mm]{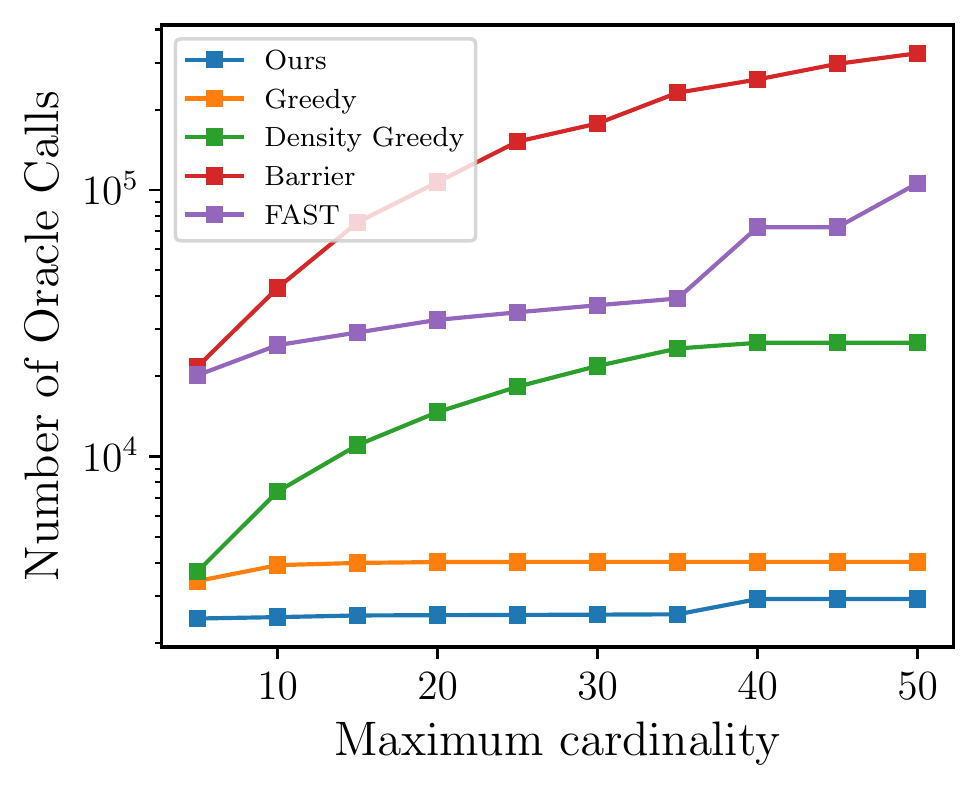}\label{fig:SMKS-video-call}}
	\caption{Comparing \cref{alg:complete_nearly_linear} with the state-of-the-art algorithms subject to the combination of a $p$-system and $d$ knapsack constraints.}
 	\label{fig:SMKS}
\end{figure*}

\section{Conclusion}

An algorithm for maximizing a submodular function is conventionally measured in terms of two quantities: a) approximation guarantee, i.e., how well it performs with respect to the  (exponential time) optimum  algorithm, and b) query complexity, i.e., how many function evaluations it requires. 
In their seminal work, \citet{nemhauser1978best} and \citet{nemhauser1078analysis} resolved part (a).   In this paper, after 44 years and building on a large body of prior work, we nearly resolved part (b) and portrayed a nearly complete picture of the landscape. Specifically, we developed a clean linear-time algorithm for maximizing a monotone submodular function subject to a cardinality constraint (and more generally a knapsack constraint). We also provided information-theoretic lower bounds on the  query complexity of both constrained and unconstrained (non-monotone) submodular maximization. Finally, we studied the tradeoff between the time complexity and approximation ratio for maximizing a monotone submodular function subject to a $p$-set system and $d$ knapsack constraints.

\bibliographystyle{plainnat}
\bibliography{tex/KnapsackPlusSystem}

\appendix

\section{Solving \texorpdfstring{\SI}{SI} using \texorpdfstring{$O(n / \log n)$}{O(n / log n)} Oracle Queries} \label{app:technique_power}

Recall that the inapproximability results proved in Section~\ref{sec:inapproximability} (namely, Theorems~\ref{thm:cardinality_lower_bound} and~\ref{thm:unconstrained_lower_bound}) are based on a reduction to a problem named {\SI} (defined in Section~\ref{sec:cardinality_lower_deterministic}). In Section~\ref{sec:cardinality_lower_randomized}, we prove that any (possibly randomized) algorithm for {\SI} must use $\Omega(n / \log n)$ oracle queries to significantly improve over some ``easy'' approximation ratio. In this section we show that this result is tight in the sense that $O(n / \log n)$ oracle queries suffice to solve {\SI} exactly. The algorithm we use for this purpose is given as Algorithm~\ref{alg:set_identification}. This algorithm is not efficient in terms of its time complexity. Nevertheless, it shows that one cannot prove an inapproximability result requiring $\omega(n / \log n)$ oracle queries for {\SI} based on information theoretic arguments only.

For simplicity, we assume in Algorithm~\ref{alg:set_identification} that the ground set $\cN$ is simply the set $[n]$. Given this assumption, we are able to base the algorithm on the following known lemma.
\begin{lemma}[Due to~\cite{lev2011size}] \label{lem:special_matrix}
If $q > (2 \log_2 3 + o(1))\frac{n}{\log n}$, then there exists a binary matrix $\mathbf{Q} \in \{0, 1\}^{q \times n}$ such that, for every vector $\mathbf{s}$, the equation $\mathbf{Q} \cdot \mathbf{x} = \mathbf{s}$ has at most one binary solution $\mathbf{x} \in \{0, 1\}^n$.
\end{lemma}
Algorithm~\ref{alg:set_identification} begins by finding such a matrix $\mathbf{Q}$ (this can be done by brute-force enumeration since we do not care about the time complexity). Then, the product $\mathbf{Q} \cdot \characteristic_{C^*}$ is calculated using $q = O(n / \log n)$ oracle queries, where $\characteristic_{C^*}$ is the characteristic vector of the set $C^*$ (\ie, a vector that includes $1$ in coordinate $i$ if $i \in C^*$, and otherwise, includes $0$ in this coordinate). Once Algorithm~\ref{alg:set_identification} has the product $\mathbf{Q} \cdot \characteristic_{C^*}$, Lemma~\ref{lem:special_matrix} guarantees that it is possible to recover the vector $\characteristic_{C^*}$ itself, which the algorithm can do, again, using brute-force.
\begin{algorithm2e}
\caption{Algorithm for {\SI}} \label{alg:set_identification}
Let $q = O(n / \log n)$ be a large enough value so that it obeys the requirement of Lemma~\ref{lem:special_matrix}.\\
Let $\mathbf{Q} \in \{0, 1\}^{q \times n}$ be a binary matrix with the properties stated in Lemma~\ref{lem:special_matrix}.\\
\For{$j = 1$ \KwTo $q$}
{
	Let $Q_j$ be a set such that $\characteristic^T_{Q_j}$ is equal to the $j$-th line of $\mathbf{Q}$.\\
	Let $s_j = |Q_j \cap C^*|$. \tcp{Can be calculated using a single oracle query.}
}
Let $\mathbf{s}$ be the vector whose $j$-th coordinate, for every integer $1 \leq j \leq q$ is $s_j$. \tcp{Note that $\mathbf{s} = \mathbf{Q} \cdot \characteristic_{C^*}$.}
Let $\mathbf{x} \in \{0, 1\}^n$ be a solution for $\mathbf{Q} \cdot \mathbf{x} = \mathbf{s}$. \tcp{By Lemma~\ref{lem:special_matrix}, $\mathbf{x} = \characteristic_{C^*}$ is the sole solution for this equation in $\{0, 1\}^n$.}
\Return{the set whose characteristic vector is $\mathbf{x}$.}
\end{algorithm2e}
\section{Fast Versions of Algorithms from Section~\ref{sec:SMKS}}

\subsection{Nearly Linear Time Version of Algorithm~\ref{alg:main}} \label{app:nearly_linear_main}

In this section we present a version of Algorithm~\ref{alg:main} that runs in nearly linear time. This version appears as Algorithm~\ref{alg:main_nearly_linear}, and it gets a quality control parameter $\eps \in (0, 1/4)$ (in addition to the parameters $\lambda$ and $\rho$ of Algorithm~\ref{alg:main}). The speedup in this version of the algorithm is obtained using a technique due to~\cite{badanidiyuru2014fast} which employs a decreasing threshold $\tau$. In every iteration of the loop starting on Line~\ref{line:scan_tau} of Algorithm~\ref{alg:main_nearly_linear}, the algorithm looks for elements whose marginal value exceeds this threshold. This guarantees that the element selected by the algorithm has an almost maximal marginal among the elements that can be selected because the previous iteration of the same loop already selected every element that could be selected when $\tau$ was larger.

\begin{algorithm2e}
\DontPrintSemicolon
\caption{\texttt{Nearly Linear Time General Algorithm}$(\lambda, \rho, \eps)$} \label{alg:main_nearly_linear}
\tcp{Build the set of big elements, and find a candidate solution based on them.}
Let $B \gets \{u \in \cN \mid \exists_{1 \leq i \leq d}\; c_i(u) > \lambda^{-1}\}$.\\
Let $S_B$ be the output set of $\BigAlg(B)$.\\

\BlankLine

\tcp{Construct a solution from the small elements.}
Let $S_0 \gets \varnothing$, $k \gets 0$, $M \gets \max_{u \in \cN} f(\{u\})$ and $\tau \gets M$.\\
\While{$\tau \geq \eps M / [(1 + \eps)n]$ \label{line:tau_loop}}
{
	\For{every element $u \in \cN \setminus (S_k \cup B)$ \label{line:scan_tau}}
	{
		\If{$S_k + u \in \cI$ and $f(u \mid S_k) \geq \max\{\tau, \rho \cdot \sum_{i = 1}^d c_i(u)\}$ \label{line:element_valid}}
		{
			Let $v_{k + 1} \gets u$, and let $S_{k + 1} \gets S_k + v_{k + 1}$.\\
			\lIf{$\max_{1 \leq i \leq d} c_i(S_{k + 1}) \leq 1$ \label{line:knapsack_check}}
			{
				Increase $k$ by $1$.
			}
			\lElse
			{
			\Return{the output set of {\SetExtract}$(\lambda, S_{k + 1})$}. \label{line:exceed_budget_nearly_linear}
			}
		}
	}
	Update $\tau \gets \tau / (1 + \eps)$.
}
\Return{the better set among $S_B$ and $S_k$}.
\end{algorithm2e}

As explained in Section~\ref{ssc:general_algorithm}, the properties we would like to prove for Algorithm~\ref{alg:main_nearly_linear} are summarized by Proposition~\ref{prop:main_guarantee_nearly_linear}, which we repeat here for convenience.

\propMainGuaranteeNearlyLinear*

We begin the proof of Proposition~\ref{prop:main_guarantee_nearly_linear} with the following lemma, which proves that Algorithm~\ref{alg:main_nearly_linear} has the time complexity stated in the proposition.
\begin{lemma}
Algorithm~\ref{alg:main_nearly_linear} has a time complexity of $O(\lambda n d + n\eps^{-1}(\log n + \log \eps^{-1}) + T_B)$, where $T_B$ is the time complexity of $\BigAlg$.
\end{lemma}
\begin{proof}
The construction of the set $B$ requires $O(dn)$ time, and therefore, the time complexity required for the entire Algorithm~\ref{alg:main_nearly_linear} except for the loop starting on Line~\ref{line:tau_loop} is $O(dn + T_B)$. In the rest of this proof we show that this loop requires $O(\lambda n d + n\eps^{-1}(\log n + \log \eps^{-1}))$ time, which implies the lemma.

The loop starting on Line~\ref{line:tau_loop} of Algorithm~\ref{alg:main_nearly_linear} runs at most the number of times that is required to decrease $\tau$ from $M$ to $\eps M / [(1 + \eps)n]$, which is
\[
	\lceil \log_{1 + \eps} (n/\eps) \rceil + 1
	\leq
	\frac{\ln n - \ln \eps}{\eps/2} + 2
	=
	O(\eps^{-1}(\log n + \log \eps^{-1}))
	\enspace.
\]
Within each iteration of the loop on Line~\ref{line:tau_loop}, the loop starting on Line~\ref{line:scan_tau} executes at most $n$ times. To understand the time complexity of the iterations of the last loop, we can observe that each such iteration takes a constant time with the exception of the following operations.
\begin{itemize}
	\item In Line~\ref{line:element_valid} we need to calculate the sum $\sum_{i = 1}^d c_i(u)$, which takes $O(d)$ time. However, we can pre-calculate this sum for all the elements of $\cN$ in $O(nd)$ time.
	\item Checking the condition on Line~\ref{line:knapsack_check} requires $O(d)$ time (assuming we maintain the values $c_i(S_k)$). However, this line is executed only when an element is added to the solution of the algorithm, which happens at most $O(r) = O(n)$ times.
	\item Executing $\SetExtract(\lambda, S_{k + 1})$ requires $O(\lambda |S_{k + 1}| d) = O(\lambda r d) = O(\lambda n d)$ time. However, this procedure is executed at most once by Algorithm~\ref{alg:main_nearly_linear}.
\end{itemize}
Combining all the above, we get that the loop starting on Line~\ref{line:tau_loop} of Algorithm~\ref{alg:main_nearly_linear} requires only $O(\lambda n d + n\eps^{-1}(\log n + \log \eps^{-1}))$ time.
\end{proof}

Like in Section~\ref{ssc:general_algorithm}, we use $\ell$ below to denote the final value of the variable $k$. Furthermore, one can observe that the proof of Observation~\ref{obs:feasible} applies to Algorithm~\ref{alg:main_nearly_linear} up to some natural modifications, and therefore, we are guaranteed that the output set of Algorithm~\ref{alg:main_nearly_linear} is feasible. The rest of this section is devoted to bounding the approximation guarantee of this output set.

Let $\tilde{E}$ be the event that Algorithm~\ref{alg:main_nearly_linear} returns through Line~\ref{line:exceed_budget_nearly_linear}. We consider separately the case in which the event $\tilde{E}$ happens and at the case in which it does not happen. When the event $\tilde{E}$ happens, the output set of Algorithm~\ref{alg:main} is the output set of {\SetExtract}. This set contains only high density elements and is large in terms of the linear constraints (by Lemma~\ref{lem:set_extract}), which provides a lower bound on its value. This argument can be formalized, leading to the following lemma, whose proof is omitted since it is analogous to the proof of Lemma~\ref{lem:E_guarantee}.
\begin{lemma} \label{lem:E_guarantee_fast}
If the event $\tilde{E}$ happens, then Algorithm~\ref{alg:main} returns a solution of value at least $\frac{\lambda \rho}{\lambda + 1}$.
\end{lemma}

Handling the case in which the event $\tilde{E}$ does not happen is somewhat more involved. Towards this goal, we recursively define a set $O_k$ for every $0 \leq k \leq \ell$ (note that the definition we give here is slightly different compared to the one given in Section~\ref{ssc:general_algorithm}). The base of the recursion is that for $k = \ell$ we define $O_\ell = OPT \setminus (B \cup \{u \in OPT \mid f(u \mid S_\ell) < \max\{\rho \cdot \sum_{i = 1}^d c_i(u)), \eps M / n\}\}$. Assuming $O_{k + 1}$ is already defined for some $0 \leq k < \ell$, we define $O_k$ as follows. Let $D_k = \{u \in O_{k + 1} \setminus S_k \mid S_k + u \in \cI\}$. If $|D_k| \leq p$, we define $O_k = O_{k + 1} \setminus D_k$. Otherwise, we let $D'_k$ be an arbitrary subset of $D_k$ of size $p$, and we define $O_k = O_{k + 1} \setminus D'_k$.
\begin{lemma} \label{lem:O_null_empty_fast}
Assuming $\tilde{E}$ does not happen, $O_0 = \varnothing$.
\end{lemma}
\begin{proof}[Proof Sketch]
The proof of this lemma is very similar to the proof of Lemma~\ref{lem:O_null_empty}. The only difference is that arguing why $O_\ell$ is a base of $S_\ell \cup O_\ell$ is a bit more involved now. Specifically, consider the last iteration of the loop starting on Line~\ref{line:tau_loop} of Algorithm~\ref{alg:main_nearly_linear}. In this iteration the value of $\tau$ was at most $\eps M / n$, and therefore, every element of $O_\ell$ would have been added to $S_\ell$ during this iteration unless this addition violates independence in $\cM$.
\end{proof}

Using the last lemma, we can now prove the following corollary, which corresponds to Corollary~\ref{cor:p-system_analysis} from Section~\ref{ssc:general_algorithm}. 
\begin{corollary} \label{cor:p-system_analysis_fast}
Assuming $\tilde{E}$ does not happen, $f(S_\ell) \geq \frac{f(O_\ell \cup S_\ell)}{(1 + \eps)p + 1}$.
\end{corollary}
\begin{proof}
We prove by induction the stronger claim that, for every integer $0 \leq k \leq \ell$,
\begin{equation} \label{eq:guarantee_k_fast}
	f(S_k)
	\geq
	\frac{f(O_k \cup S_k)}{(1 + \eps)p + 1}
	\enspace.
\end{equation}
For $k = 0$ this inequality follows from the non-negativity of $f$ since $S_0 = O_0 = \varnothing$ by Lemma~\ref{lem:O_null_empty_fast}. Assume now that Inequality~\eqref{eq:guarantee_k_fast} holds for some value $k - 1$ obeying $0 \leq k - 1 < \ell$, and let us prove it for $k$.

Consider the set $\Delta_k = O_k \setminus O_{k - 1}$. Let $u'$ be an element of $\Delta_k$ maximizing $f(u' \mid S_{k - 1})$, and let $\tau'$ be the maximum value that $\tau$ takes in any iteration of the loop starting on Line~\ref{line:tau_loop} of Algorithm~\ref{alg:main_nearly_linear} that is not larger than $f(u' \mid S_k)$. Such a value exists because the inclusion $\Delta_k \subseteq O_\ell$ implies that every element $u \in \Delta_k$ obeys $f(u \mid S_{k - 1}) \geq f(u \mid S_\ell) \geq \max\{\rho \cdot \sum_{i = 1}^d c_i(u), \eps M / n\}$. Observe now that $u'$ cannot belong to $S_\ell$ because $u' \in \Delta_k \subseteq O_\ell$ guarantees $f(u' \mid S_\ell) > 0$, which implies that, during the iteration of of the loop starting on Line~\ref{line:tau_loop} that corresponds to $\tau'$, the element $u'$ was not added to the solution of Algorithm~\ref{alg:main_nearly_linear}. Let us study the reason that $u'$ was not added. By the definition of $\tau'$, $f(u' \mid S_{k - 1}) \geq \tau'$. Additionally, by the construction of $O_{k - 1}$, every element of $\Delta_k$ can be added to $S_{k - 1}$ without violating independence in $\cM$. These two facts imply that the reason that $u'$ was not added must have been that, by the time Algorithm~\ref{alg:main_nearly_linear} considers the element $u'$ in the iteration correspond to $\tau'$, the solution of Algorithm~\ref{alg:main_nearly_linear} already contained at least $k$ elements. Since up to this time the algorithms adds to its solution only elements whose marginal contribution is at least $\tau'$, this implies
\[
	f(S_{k}) - f(S_{k - 1})
	\geq
	\tau'
	\geq
	\frac{f(u' \mid S_{k - 1})}{1 + \eps}
	\geq
	\frac{\sum_{u \in \Delta_k} f(u \mid S_{k - 1})}{(1 + \eps)|\Delta_k|}
	\enspace.
\]

Recall now that $f(S_k) - f(S_{k - 1}) = f(v_k \mid S_{k - 1})$. Combining this equality with the previous one, we get
\begin{align*}
	f(S_k)
	\geq{} &
	f(S_{k - 1}) + \frac{f(v_k \mid S_{k - 1}) + \sum_{u \in \Delta_k} f(u \mid S_{k - 1})}{(1 + \eps)|\Delta_k| + 1}
	\geq
	\frac{f(O_{k - 1} \cup S_{k - 1})}{(1 + \eps)p + 1} + \frac{f(\Delta_k + v_k \mid S_{k - 1})}{(1 + \eps)|\Delta_k| + 1}\\
	\geq{} &
	\frac{f(O_{k - 1} \cup S_{k - 1})}{(1 + \eps)p + 1} + \frac{f(\Delta_k + v_k \mid S_{k - 1})}{(1 + \eps)p + 1}
	\geq
	\frac{f(O_{k} \cup S_{k})}{(1 + \eps)p + 1}
	\enspace,
\end{align*}
where the second inequality follows from the induction hypothesis and the submodularity of $f$, the penultimate inequality follows from the monotonicity of $f$ and the observation that the construction of $O_{k - 1}$ guarantees $|\Delta_k| \leq p$, and the last inequality follows again from the submodularity of $f$.
\end{proof}

To use the last corollary, we need a lower bound on $O_\ell \cup S_\ell$, which is given by the next lemma (and corresponds to Lemma~\ref{lem:Oell_lower_bound}).
\begin{lemma} \label{lem:Oell_lower_bound_fast}
$f(O_\ell \cup S_\ell) \geq (1 - \eps) f(OPT) - f(S_B) / \alpha - \rho \cdot \left[d - \frac{|OPT \cap B|}{\lambda} \right]$.
\end{lemma}
\begin{proof}
Observe that
\begin{align} \label{eq:Oell_bound_fast}
	f(O_\ell \cup S_\ell)
	={} &
	f(OPT \setminus (B \cup \{u \in OPT \mid f(u \mid S_\ell) < \max\{\rho \cdot \sum\nolimits_{i = 1}^d c_i(u), \eps M / n \}\}) \cup S_\ell)\\ \nonumber
	\geq{} &
	f(OPT) - f(OPT \cap B) - f(\{u \in OPT \setminus B \mid f(u \mid S_\ell) < \rho \cdot \sum\nolimits_{i = 1}^d c_i(u)\} \mid S_\ell)\\ \nonumber
	& - f(\{u \in OPT \setminus B \mid f(u \mid S_\ell) < \eps M / n\} \mid S_\ell)\\\nonumber
	\geq{} &
	f(OPT) - f(S_B) / \alpha - f(\{u \in OPT \setminus B \mid f(u \mid S_\ell) < \rho \cdot \sum\nolimits_{i = 1}^d c_i(u)\} \mid S_\ell)\\\nonumber
	& - f(\{u \in OPT \setminus B \mid f(u \mid S_\ell) < \eps M / n\} \mid S_\ell)
	\enspace,
\end{align}
where the first inequality follows from the submodularity and monotonicity of $f$, and the second inequality follows from the definition of $S_B$. To lower bound the rightmost side of the last inequality, we need to upper bound the two last terms in it.
\begin{align*}
	f(\{u \in{} & OPT \setminus B \mid f(u \mid S_\ell) < \rho \cdot \sum\nolimits_{i = 1}^d c_i(u)\} \mid S_\ell)
	\leq
	\sum_{\substack{u \in OPT \setminus B \\ f(u \mid S_\ell) < \rho \cdot \sum\nolimits_{i = 1}^d c_i(u)}} \mspace{-36mu} f(u \mid S_\ell)\\
	\leq{} &
	\rho \cdot \sum_{u \in OPT \setminus B} \sum_{i = 1}^d c_i(u)
	=
	\rho \cdot \left[\sum_{u \in OPT} \sum_{i = 1}^d c_i(u) - \sum_{u \in OPT \cap B} \sum_{i = 1}^d c_i(u)\right]
	\leq
	\rho \cdot \left[d - \frac{|OPT \cap B|}{\lambda} \right]
	\enspace,
\end{align*}
where the last inequality holds since $OPT$ is a feasible set and every element of $B$ is big. Additionally,
\[
	f(\{u \in OPT \setminus B \mid f(u \mid S_\ell) < \eps M / n\} \mid S_\ell)
	\leq
	\sum_{\substack{u \in OPT \setminus B \\ f(u \mid S_\ell) < \eps M / n}} \mspace{-18mu} f(u \mid S_\ell)
	\leq
	\eps M
	\leq
	\eps \cdot f(OPT)
	\enspace,
\]
where the second inequality holds since $|OPT \cap B|$ is a subset of $\cN$, and therefore, is of size at most $n$; and the last inequality holds since every element of $\cN$ is a feasible solution by our assumption. Plugging the last inequalities into Inequality~\eqref{eq:Oell_bound_fast} completes the proof of the lemma.
\end{proof}

We are now ready to prove Proposition~\ref{prop:main_guarantee_nearly_linear}.
\begin{proof}[Proof of Proposition~\ref{prop:main_guarantee_nearly_linear}]
If the event $\tilde{E}$ happened, then Lemma~\ref{lem:E_guarantee_fast} guarantees that the output of Algorithm~\ref{alg:main_nearly_linear} is of value at least $\lambda \rho / (\lambda + 1)$. Therefore, to complete the proof of the lemma, it suffices to show that when the event $E$ does not happen and $\rho^* \geq \rho$, the value of the output of the Algorithm~\ref{alg:main_nearly_linear} is at least $\lambda \rho^* / (\lambda + 1)$; and the rest of this proof is devoted to showing that this is indeed the case.

In the last case, Corollary~\ref{cor:p-system_analysis_fast} and Lemma~\ref{lem:Oell_lower_bound_fast} imply together
\[
	f(S_\ell)
	\geq
	\frac{(1 - \eps)f(OPT) - f(S_B) / \alpha - \rho \cdot \left[d - \frac{|OPT \cap B|}{\lambda}\right]}{(1 + \eps)p + 1}
	\enspace,
\]
and therefore, the output set of Algorithm~\ref{alg:main_nearly_linear} is of value at least
\begin{align*}
	\max\{f(S_\ell), f(S_B)\}
	\geq{} &
	\frac{((1 + \eps)p + 1) \cdot f(S_\ell) + \alpha^{-1} \cdot f(S_B)}{(1 + \eps)p + 1 + \alpha^{-1}}\\
	\geq{} &
	\frac{(1 - \eps)f(OPT) - \rho(d - |OPT \cap B|/\lambda)}{(1 + \eps)p + 1 + \alpha^{-1}}\\
	\geq{} &
	\frac{(1 - \eps)f(OPT) - \rho^*(d - |OPT \cap B|/\lambda)}{(1 + \eps)p + 1 + \alpha^{-1}}
	=
	\frac{\lambda \rho^*}{\lambda + 1}
	\enspace.
	\qedhere
\end{align*}
\end{proof}

\subsection{Nearly-Linear Time Version of Algorithm~\ref{alg:complete}} \label{app:nearly_linear_complete}

In this section we describe a fast version of Algorithm~\ref{alg:complete}. This version appears as Algorithm~\ref{alg:complete_nearly_linear}. There are only three differences between the two algorithms: (1) Algorithm~\ref{alg:complete_nearly_linear} uses Algorithm~\ref{alg:main_nearly_linear} instead of Algorithm~\ref{alg:main}, (2) the value of $\bar{\tau}$ is slightly higher in Algorithm~\ref{alg:main_nearly_linear}, and (3) the definition $\rho(i)$ is slightly modified in Algorithm~\ref{alg:main_nearly_linear} to be $\rho(i) \triangleq (1 - 2\eps)(1 + \delta)^i \cdot \max_{u \in \cN} f(\{u\}) / (p + 1 + \underline{\alpha}^{-1} + d)$.
\begin{algorithm2e}
\DontPrintSemicolon
\caption{\texttt{$\rho$ Guessing Algorithm}$(\lambda, \eps, \delta)$} \label{alg:complete_nearly_linear}
Let $\underline{i} \gets 0$, $\bar{i} \gets \left\lceil \log_{1 + \delta} \frac{2n}{p} - \log_{1 + \delta}  \frac{1 - 2\eps}{p + 1 + \underline{\alpha}^{-1} + d} \right\rceil$ and $k \gets 0$.\\
\While{$\bar{i} - \underline{i} > 1$}
{
	Update $k \gets k + 1$.\\
	Let $i_k \gets \lceil (\underline{i} + \bar{i}) / 2 \rceil$.\\
	Execute Algorithm~\ref{alg:main_nearly_linear} with $\rho = \rho(i_k)$. Let $A_k$ denote the output set of this execution of Algorithm~\ref{alg:main_nearly_linear}, and let $E_k$ denote the event $\tilde{E}$ for the execution.\\
	\lIf{the event $E_k$ happened}{Update $\underline{i} \gets i_k$.}
	\lElse{Update $\bar{i} \gets i_k$.}
}
Execute Algorithm~\ref{alg:main_nearly_linear} with $\rho = \rho(\underline{i})$. Let $A'$ denote the output set of this execution of Algorithm~\ref{alg:main_nearly_linear}.\\
\Return{the set maximizing $f$ in $\{A'\} \cup \{A_{k'} \mid 1 \leq k' \leq k\}$}. 
\end{algorithm2e}

The following observation corresponds to Observation~\ref{obs:rho_star_limits}. The proofs of the two observations are very similar.

\begin{observation} \label{obs:rho_star_limits_nearly_linear}
For the value $\rho^*$ stated in Proposition~\ref{prop:main_guarantee_nearly_linear},
\[
	\frac{1 - 2\eps}{p + 1 + \underline{\alpha}^{-1} + d}
	\leq
	\frac{\rho*}{\max_{u \in \cN} f(u)}
	\leq
	\frac{2n}{p}
	\enspace.
\]
\end{observation}
\begin{proof}
According to the definition of $\rho^*$,
\begin{align*}
	\rho^*
	={} &
	\frac{(1 - \eps)f(OPT)}{((1 + \eps)p + 1 + \alpha^{-1}) / (1 + \lambda^{-1}) + d - |OPT \cap B| / \lambda}\\
	\leq{} &
	\frac{(1 - \eps) n \cdot \max_{u \in \cN} f(\{u\})}{(1 + \eps) p / (1 + \lambda^{-1}) + d - |OPT \cap B| / \lambda}\\
	\leq{} &
	\frac{(1 - \eps)2n \cdot \max_{u \in \cN} f(\{u\})}{(1 + \eps)p}
	\leq
	\frac{2n \cdot \max_{u \in \cN} f(\{u\})}{p}
	\enspace,
\end{align*}
where the first inequity follows from the submodularity and non-negativity of $f$, and the second inequality follows from the upper bound on $|OPT \cap B|$ given in the discussion before Proposition~\ref{prop:main_guarantee} and the inequality $\lambda \geq 1$.

Similarly,
\begin{align*}
	\rho^*
	={} &
	\frac{(1 - \eps)f(OPT)}{((1 + \eps)p + 1 + \alpha^{-1}) / (1 + \lambda^{-1}) + d - |OPT \cap B| / \lambda}\\
	\geq{} &
	\frac{(1 - \eps)\max_{u \in \cN} f(\{u\})}{((1 + \eps)p + 1 + \alpha^{-1}) / (1 + \lambda^{-1}) + d - |OPT \cap B| / \lambda}\\
	\geq{} &
	\frac{(1 - \eps) \max_{u \in \cN} f(\{u\})}{(1 + \eps)p + 1 + \underline{\alpha}^{-1} + d}
	\geq
	\frac{(1 - 2\eps) \max_{u \in \cN} f(\{u\})}{p + 1 + \underline{\alpha}^{-1} + d}
	\enspace,
\end{align*}
where the first inequality holds since every singleton is a feasible set by our assumption, and the second inequality holds since $1 + \lambda^{-1} \geq 1$ and $|OPT \cap B| / \lambda \geq 0$.
\end{proof}

Given the last observation, the proof of Proposition~\ref{prop:complete_guarantee_nearly_linear} is identical to the proof of Proposition~\ref{prop:complete_guarantee} up to the following changes.
\begin{itemize}
	\item One has to use Proposition~\ref{prop:main_guarantee_nearly_linear} and Observation~\ref{obs:rho_star_limits_nearly_linear} instead of Proposition~\ref{prop:main_guarantee} and Observation~\ref{obs:rho_star_limits}.
	\item The calculation needed to bound the number of iterations performed by the algorithm becomes a bit more involved. Specifically, we can upper bound it by
	\begin{align*}
	&
	2 + \log_2 \left[\log_{1 + \delta} \frac{2n}{p} - \log_{1 + \delta}  \frac{1 - 2\eps}{p + 1 + \underline{\alpha}^{-1} + d}\right]\\
	={} &
	2 + \log_2 [\ln (2n) - \ln(p) - \ln(1 - 2\eps) + \ln(p + 1 + \underline{\alpha}^{-1} + d)] - \log_2 \ln(1 + \delta)\\
	\leq{} &
	4 + \log_2 [\ln n + \ln 2 + \ln(2 + \underline{\alpha}^{-1} + d)] + \log_2 \delta^{-1}\\
	={} &
	O(\log \delta^{-1} + \log(\log n + \log (\underline{\alpha}^{-1} + d)))
	\enspace.
\end{align*}
\end{itemize}

\end{document}